\documentclass[
nonacm,rgb,acmsmall]{acmart}
\AtBeginDocument{%
  }

\setcopyright{acmcopyright}
\copyrightyear{2018}
\acmYear{2018}
\acmDOI{XXXXXXX.XXXXXXX}

\acmConference[Conference acronym 'XX]{Make sure to enter the correct
  conference title from your rights confirmation emai}{June 03--05,
  2018}{Woodstock, NY}
\acmPrice{15.00}
\acmISBN{978-1-4503-XXXX-X/18/06}

\usepackage{tikz}
\usepackage{mathtools}                     
 \usepackage[all]{xy}  
\usepackage{array}
\usepackage{fixltx2e}
\usepackage{url}
 
\usepackage{amsmath}
\usepackage{cmll}
\usepackage{amsthm}
\usepackage{stmaryrd}
\usepackage{mathpartir}
\usepackage{color}

\usepackage{bm}
\usepackage{calc,xifthen} 
\usepackage{calrsfs}

\newcommand{\ESfig}[3][!ht]{\begin{figure}[#1]
		\ifthenelse{\isempty{#2}}{}{\caption{#2}}
		\begin{displaymath}\xymatrix@=2em{#3}\end{displaymath} \end{figure}}
\newcommand{\ESfighalf}[3][!ht]{\begin{figure}[#1]
		\ifthenelse{\isempty{#2}}{}{\caption{#2}} 
		\begin{displaymath}\xymatrix@=1em{#3}\end{displaymath} \end{figure}}
\newcommand{\ESfigtiny}[3][!ht]{\begin{figure}[#1]
		\ifthenelse{\isempty{#2}}{}{\caption{#2}}
		\begin{displaymath}\xymatrix@=0.5em{#3}\end{displaymath} \end{figure}}

\newcommand{\posn}{\hbox{{\it do}}[A,B]}
\newcommand{\dop}{\hbox{{\it do}}}

\theoremstyle{acmplain}
   \newtheorem{notation}[subsection]{Notation}

\newcommand{\pfrule}[2]{{{\hbox{{$#1$}}\over{\strut \hbox{{$#2$}}}}}}
 \def\fatbar{\mathop{[\!]}}
\newcommand{\Eta}{{\rm H}}

\newcommand{\stcomp}[3]{\exists #1.\, [\, #2 \mathrel{
\vvbar} #3\, ]}

\newcommand{\pin}{\hbox{ {\rm in} }}
\newcommand{\plet}{\hbox{ {\rm let }}}
\newcommand{\larrow}{\Leftarrow}

\newcommand{\tri}{\trianglelefteq}

\newcommand{\Strat}{{\mathbf{Strat}}}
\newcommand{\exi}{{1}}
\newcommand{\foral}{{2}}

\newcommand{\dem}{{\it dem}}
\newcommand{\out}{{\it out }}
\newcommand{\sPi}{\Pi^s}
\newcommand{\curry}{{\rm curry\,}}
\newcommand{\app}{{\rm apply}}
\newcommand{\DOP}{{\bf dOp}}

\newcommand{\comp}{{\circ}}
\newcommand{\Optic}{{\rm optic}}
\newcommand{\pto}{\rightharpoonup}
\newcommand{\cov}{{{\mathrel-\joinrel\subset}}}

\newcommand{\eswp}{event structure with polarity}
\newcommand{\esswp}{event structures with polarity}

\newdir{C}{%
!/4.5pt/@{( }*:(1,-.2)@{}*:(1,+.2)@_{}}
\newdir{|>}{%
!/4.5pt/@{|}*:(1,-.2)@{>}*:(1,+.2)@_{>}}

 \def\fatbar{\mathop{[\!]}}
 
 \newcommand{\goi}{{\it goi}}
\newcommand{\bel}{\sqsubseteq}
\newcommand{\leb}{\sqsupseteq}
\newcommand{\Tensor}{\bigotimes} 
\newcommand{\tensor}{\otimes} 
\newcommand{\Hilb}{{\mathcal H}}
\newcommand{\Q}{{\mathcal Q}}

\newcommand{\Para}[1]{\rm Para^{}({#1})}

\newcommand{\A}{{\mathcal A}}
\newcommand{\B}{{\mathcal B}}

\newcommand{\sncirc}{\oast} 

\def\endex{{\hspace*{\fill}\hbox{$\Box$}}}

\newcommand{\opmove}{{\boxminus}}
\newcommand{\plmove}{{\boxplus}}

\def\profto{\!\!\!\xymatrix@C-.75pc{\ar[r]|-{\! +\!} &}\!\!\! }
\newcommand{\spano}[5]{{
\xymatrix
{ 
    & {#3}\ar[dl]_{#2}\ar[dr]^-{#4} &\\
      {#1} && {#5}
  }}}
\def\all{\forall}

\newdir{pb}{:(1,-1)@^{|-}}
\def\pb#1{\save[]+<16 pt,0 pt>:a(#1)\ar@{pb{}}[]\restore}

\newcommand{\co}{\mathbin{{\it co}}}
\newcommand{\vvbar}{{\mathbin{\parallel}}}
\newcommand{\stcirc}{{{\odot}}}
\newcommand{\scirc}{{{\odot}}}

\newcommand{\longcov}[1]{{\stackrel{#1}{\mathrel-\joinrel\relbar\joinrel\subset\,}}}

\renewcommand{\max}{\it top}
\newcommand{\imc}{\rightarrowtriangle}
\newcommand{\setdif}{\setminus}
\newcommand{\sig}{\sigma}
\newcommand{\eps}{\, \epsilon}

\newcommand{\fsubseteq}{\subseteq_{\rm fin}}
\newcommand{\CC}{{\rm C\!\!C}}
\newcommand{\cc}{\ c\!c\,}
\newcommand{\pol}{{\it pol}}

\newcommand{\F}{{\mathcal F}}
\renewcommand{\G}{{\mathcal G}}

\def\Con{{\rm Con}}
\newcommand{\Fam}{{\mathcal F}}
\newcommand{\parrow}{\rightharpoonup}

\newcommand{\id}{{\rm id}}
\newcommand{\set}[2]{{\{  #1\  | \  #2 \} }}
\newcommand{\setof}[1]{{\{ #1 \} }}
\newcommand{\eqdef}{\coloneqq}
\newcommand{\iso}{\cong}

\newcommand{\ie}{{\it i.e.}}
\newcommand{\eg}{{\it e.g.}}
\newcommand{\cf}{{\it cf.}}

\newcommand{\viz}{{\it viz.}}

\def\mxxth{\mathsurround=0pt}
\dimendef\dimenxx=0
\def\openup{\afterassignment\xxpenup\dimenxx=}
\def\xxpenup{\advance\lineskip\dimenxx
  \advance\baselineskip\dimenxx \advance\lineskiplimit\dimenxx}
\def\eqalign#1{\,\vcenter{\openup1\jot \mxxth
  \ialign{\strut\hfil$\displaystyle{##}$&$\displaystyle{{}##}$\hfil
     \crcr#1\crcr}}\,}

\newif\ifdtxxp
\def\displxxy{\global\dtxxptrue \openup1\jot \mxxth
  \everycr{\noalign{\ifdtxxp \global\dtxxpfalse
      \vskip-\lineskiplimit \vskip\normallineskiplimit
      \else \penalty\interdisplaylinepenalty \fi}}}
\def\displaylines#1{\displxxy
  \halign{\hbox to\displaywidth{$\hfil\displaystyle##\hfil$}\crcr
      #1\crcr}}

\newskip\mycntring \mycntring=0pt plus 1000pt minus 1000pt

\def\leqalignno#1{\displxxy \tabskip=\mycntring
  \halign to\displaywidth{\hfil$\displaystyle{##}$\tabskip=0pt
      &$\displaystyle{{}##}$\hfil\tabskip=\mycntring
      &\kern-\displaywidth\rlap{$##$}\tabskip=\displaywidth\crcr
      #1\crcr}}

\acmConference[POPL 2023]{Principles of Programming Languages}{15-21 January}{San Antonio, Texas, USA}
\acmYear{2023}
\acmISBN{} 
\acmDOI{} 
\startPage{1}

\newcommand{\fconf}[1]{{\mathcal C}(#1)}
\newcommand{\conf}[1]{\:\!{\mathcal C}(#1)^o}
\newcommand{\iconf}[1]{\:\!{\mathcal C}(#1)}
 
\newdir{C}{%
!/4.5pt/@{( }*:(1,-.2)@{}*:(1,+.2)@_{}}

\newdir{|>}{%
!/4.5pt/@{|}*:(1,-.2)@{>}*:(1,+.2)@_{>}}

\begin{document}

\title{
Making Concurrency Functional}


\author{Glynn Winskel}
\affiliation{%
  \institution{Edinburgh Research Centre, Central Software Institute, Huawei;\\ }
   \institution{University of Strathclyde, Glasgow}
  \country{United Kingdom}}
\renewcommand{\shortauthors}{Winskel}

\begin{abstract}
The article bridges between two 
major paradigms in computation, the {\em functional}, at basis computation from input to output,  and the {\em interactive}, where computation reacts to its environment while underway.  
 Central to  any compositional theory of interaction is the dichotomy between a system and its environment.
 Concurrent games and strategies address the dichotomy in fine detail,  very locally, in a distributed fashion, through 
 distinctions between Player moves (events of the system) and Opponent moves (those of the environment). 
A functional approach has to handle the dichotomy 
more 
ingeniously,
via its blunter distinction between input and output.  This has led to  a variety of functional approaches, specialised to particular interactive demands.  
Through concurrent games we can 
see what separates and connects the differing paradigms, and 
show how: 

\noindent
$\bullet$ to 
lift functions to strategies
; 
how to turn functional dependency to causal dependency. 

\noindent
$\bullet$ 
several paradigms of functional programming and logic arise naturally as full subcategories of concurrent games,\! 
including stable domain theory; nondeterministic dataflow; geometry of interaction; 
the dialectica interpretation; lenses and optics, and 
their extensions to containers in dependent lenses and optics. 

\noindent
$\bullet$ 
the enrichments  of strategies (
\eg~to probabilistic, quantum 
or real-number computation)  
specialise to the 
functional cases.
  \end{abstract}



\maketitle

\section{Introduction}

The view of computation as functions is at the very 
foundation of computer science: the Church-Turing thesis expresses the coincidence of different notions of computable function; programming with higher-order functions is now taken for granted.  

In contrast the view of computation as interaction is more recent and less settled, and often obscured 
by adherence to one syntax or another, perhaps each with its own mechanism of interaction.  Instead our approach is maths-driven. 
Its tools are those of distributed/concurrent games and strategies~\cite{lics11}, a causal model which allows for highly distributed interaction. 
Concurrent games and strategies are built on the mathematical foundations of categories of models for interaction~\cite{handbook}, chiefly on the central model of event structures~\cite{evstrs}.\footnote{A core language for concurrent strategies derives from the mathematical structure, although we shall only glimpse it here in Section~\ref{sec:corelang}: it  
 is {\em higher-order} and
an interesting hybrid of {\em dataflow}, \cf~TensorFlow~\cite{tensorflow},   {\em concurrent process calculi}, \cf~CSP, CCS and {\em Session Types}~\cite{hoare,milner,session-types}.}
 
 Whereas the basic mechanism of interaction of functions is clear---ultimately by function composition---a functional approach can struggle with finding quite the right way to approach computation which isn't simply from input to output.  The literature includes approaches via lenses, optics, combs,  containers, dependent lenses, 
open games and learners~\cite{oles,pierce,optics,Chiribella2008,containers,opengames,openlearners}. 
The difficulties are compounded by enrichments 
 to, say, probabilistic, quantum or real-number computation. 
  
In functional approaches new patterns of interaction are often achieved
by extending the usual input/output of functions with extra parameters to permit exchanges with the environment while computation is underway; the environment may comprise another similar parameterised function.  But the types of
 functions tend only to give a static, rather rigid, partial picture of the dynamics of interaction.  
This handicaps the expression of and search for more complicated patterns of interaction within functional languages: for instance, 
patterns of interaction that may change over time, perhaps with one pattern of interaction replacing 
another, 
or perhaps being chosen nondeterministically or probabilistically.  And if we are to allow very general interactions how are we to avoid functional loops which may not  be sensible for the functions of interest?

By adopting a model  which addresses interaction from the outset, we can better understand and explore the space of possible interactions, functional or otherwise. 
 Concurrent games and strategies provide 
 a way to describe and orchestrate temporal patterns of interaction between functions, their fine-grained dependencies and dynamic linkage---Sections~\ref{sec:part2} and~\ref{sec:enrichment}. 
 They support 
 enrichments to strategies for probabilistic, quantum and real-number computation. 

This article bridges between the two paradigms of computation, the functional and the interactive. 
In broad terms it shows: 
\begin{itemize}
\item
How to convert a general class of functions to concurrent strategies---Section~\ref{sec:part1}; this helps in the programming of strategies via functional techniques and is of potential further use in describing sub(bi)categories of strategies through functions.  
 \item
How in many cases we can describe concurrent strategies as interacting patterns of  functions; it reveals many paradigms in functional programming arise as full subcategories associated with special cases of concurrent games; in these cases composition of strategies can be described via simpler function composition---Section~\ref{sec:part2}.
\item
How 
concurrent strategies enriched in a symmetric monoidal category $\mathcal M$ determine interacting patterns of ``functions'' (maps in $\mathcal M$) and how these compose through the composition of strategies; in this sense a sub(bi)category of strategies determines its own functional paradigm.  This can be used to systematise the way we explore interaction between functions---Section~\ref{sec:enrichment}.\end{itemize}

Amplifying the second point above, it was a surprise to the author 
how neatly and automatically many functional paradigms arise  simply by specialising to full subcategories of concurrent games. 
For example: 
\begin{itemize}
\item
We shall see how by restricting to deterministic strategies between concurrent games where all moves are Player moves we rediscover {\em stable functions} and Berry's stable domain theory, of which Girard's qualitative domains and coherence spaces are special cases. For such restricted games, general, possibly nondeterministic, strategies correspond to {\em stable spans,} a model discovered and rediscovered in 
compositional accounts of nondeterministic dataflow. 
\item
Only marginally more complicated than those purely Player games are games which consist of two parallel components, one a purely Player game and the other with purely Opponent moves.  Strategies between such games yield models for {\em Geometry of Interaction} built on stable functions and stable spans~\cite{GirardGoI1,nfgoi,joyal-street-verity,AbramskyHaghverdiScott}. 
\item
Adjoining winning conditions and  imperfect information to these games, so Opponent can see the moves of Player but not the converse, we recover a {\em dialectica category}~\cite{valeria}, so G\"odel's dialectica interpretation~\cite{dialecticaInt}, from deterministic strategies.  We obtain from G\"odel's work an interpretation of proofs in first-order arithmetic as winning strategies.  Dialectica categories, studied by Valeria de Paiva in her Cambridge PhD~\cite{valeria}, mark an early occurrence of {\em lenses} used in functional programming, where they were invented independently to make  composable local changes on data-structures~\cite{oles,pierce}.  
 \item
 The newer paradigm of optics appears in characterising arbitrary, not just deterministic, strategies between dialectica games and when we move to more general container games, associated with container types~\cite{containers}.  Deterministic strategies between container games amount to {\em dependent lenses} and nondeterministic strategies to a form of {\em dependent optics}.  The definition of dependent optic is derived as a characterisation of general strategies between container games; it appears to be new~\cite{hedges-depoptics}. 
\end{itemize}

 {\sl After the basics on event structures, the tools of stable families, and concurrent strategies, the new contribution comes  in three parts which can roughly be described as:
how to describe strategies by functions;
how to describe functions and functional paradigms by strategies; and, how enriched strategies describe interacting patterns of functions. The first, rather technical part, Section~\ref{sec:part1},  introduces a powerful method for lifting a very broad class of functions to strategies, turning functional into causal dependency. It makes essential use of stable families and the Scott order 
  intrinsic to a concurrent game.  The second part, Section~\ref{sec:part2},  
  concerns how causal dependency 
  determines functional dependency, and shows how many 
paradigms discovered in 
making functions interactive arise as subcategories of concurrent 
games.  
The final, much shorter, third part, Section~\ref{sec:enrichment},  shows how to enrich strategies in a symmetric monoidal category.  
An enriched  strategy imposes a dynamic pattern of interaction  between arrows in the monoidal category; the pattern of interaction has the form of an event structure.    Such patterns of interaction compose well and won't contain loops of functional dependency because they are determined by strategies. }

\section{Event structures}
An {\em  event structure}~\cite{evstrs} comprises $(E, \leq, \Con)$, consisting of a set $E$ of {\em events}  
which are
partially ordered by $\leq$, the {\em causal dependency
relation},
and a  nonempty {\em consistency} relation $\Con$ consisting of finite subsets of $E$.  The relation
$e'\leq e$ expresses that event $e$ causally depends on the previous occurrence of event $e'$; the consistency relation, those events which may occur together.   We insist that the partial order is {\em finitary}, \ie
\begin{itemize}
\item $[e]\eqdef \set{e'}{e'\leq e}\hbox{ is finite for all } e\in E$\,,
\end{itemize}
and that consistency satisfies 
\begin{itemize}
\item 
$\setof{e}\in\Con \hbox{  for all } e\in E$\,,
\item 
$Y\subseteq X\in\Con \hbox{ implies }Y\in \Con,\ \hbox{ and}$
\item
$
X\in\Con \ \&\  e\leq e'\in X \hbox{ implies } 
X\cup\setof{e}\in\Con$\,.
\end{itemize}
There is an accompanying notion of state or history. A {\em configuration} is a, possibly infinite, subset 
$x\subseteq E$ which is:  
\begin{itemize} 
\item
{\em consistent,} $X\subseteq x 
\ \&\ X \hbox{ is finite}  \hbox{ implies } X\in\Con$\,; 
and
\item
{\em down-closed, }
$ e'\leq e\in x  \hbox{ implies } e' \in x$\,.
 \end{itemize}

Two events $e, e'$ are  called {\em concurrent} if  the set $\setof{e,e'}$ is in $\Con$ and neither event is causally dependent on the other; then we write $e\co e'$.  In games the relation of {\em immediate} dependency $e\imc e'$, meaning $e$ and $e'$ are distinct with $e\leq e'$ and no event in between,  plays a very important role.
  We write $[X]$ for the down-closure of a subset of events $X$.  
Write $\iconf E$ for the configurations of $E$ and $\conf E$ for its {\em finite configurations}.   (Sometimes we shall need to distinguish the precise event structure to which a relation is associated and write, for instance, $\leq_E$,  $\imc_E$ or $\co_E$.)

Let  $E$ and $E'$ be event structures.
A {\em  map} of event structures 
$f:E\to E'$ 
 is a partial function on events
$f:E\parrow E'$ such that 
for all   $x\in\iconf E$
its direct image $f x\in\iconf{E'}$  and 
$$
\hbox{if } e, e' \in x 
\hbox{ and } f(e) =f(e') \hbox{ (with both defined)}, \hbox{ then } e=e'\,.
$$

Maps of event structures compose as partial functions. 
Notice that  for a {\em total} map $f$, \ie~when the function $f$ is total,  the condition on maps 
says it is {\em  locally injective}, in the sense that w.r.t.~any configuration $x$ of the domain the restriction of $f$ to a function from $x$   is injective; the restriction of total $f$ to a function from $x$ to $f x$ is thus bijective.  

Although a map $f:E\to E'$ of event structures does not generally preserve causal dependency, it does reflect   
causal dependency  locally:  whenever $e, e'\in x$, a configuration of $E$, and $f(e)$ and $f(e')$ are both defined with $f(e')\leq f(e)$, then $e'\leq e$.  Consequently,
$f$  preserves the concurrency relation: if $e\co e'$ in $E$  
then $f(e)\co f(e')$, when defined.

A total map of event structures is {\em rigid} when it preserves causal dependency.  Rigid  maps induce  discrete fibrations:

\begin{proposition}
A total map $f:E\to E'$ of event structures is rigid iff  for all   $x\in\fconf E$ and $y\in\fconf{E'}$,
$$
y\subseteq f x \implies \exists z\in\fconf{E}.\ z\subseteq x \hbox{ and } f z = y\ .
$$
The configuration $z$ is necessarily unique by 
local injectivity
.  
\end{proposition}

 \section{Stable families
 }
 In an 
 event structure, defined above, an event $e$ has a unique causal history, the prime configuration  $[e]$.  Constructions directly on such event structures can be unwieldy, as often an event is more immediately associated with several mutually inconsistent causal histories. In this case the broader model of stable families is apt, especially so, as any stable family yields an event structure~\cite{icalp82,evstrs}.
 
 A subset $X$ of a family of sets $\F$ 
 is  {\em compatible} if there is an element of $\F$ which includes all elements of $X$; we say  $X$ is {\em finitely compatible} if every finite subset of $X$ is compatible. \\
  A  {\em stable family} 
is a non-empty family of sets $\F$ which is\\
$\bullet$  {\em  Complete:}    
 $\all Z\subseteq\F.  \   \hbox{if } Z\hbox{ is finitely compatible, } 
 \bigcup Z \in \F$\,;\\
$\bullet$ 
{\em  Stable:}
$
\all Z\subseteq\F.\ Z\not= \emptyset  \ \& \    Z\hbox{ is compatible } 
\implies 
\bigcap Z\in \F$;\\
$\bullet$ 
{\em  Finitary:}
$
\all x \in\F, e\in x\exists x_0\in\F.\ x_0 \hbox{ is finite }\ \&\ e\in x_0\subseteq x$; \\
$\bullet$ {\em  Coincidence-free:}   For all
$x \in \F$,  $ e, e'
\in x$ with $e \not=  e'$,
$$
 \exists x_0 \in \F. \   
 x_0  \subseteq  x \ \& \  (e\in  x_0  \iff e' \notin  x_0)\, .$$ 
We call elements of $\F$ its {\em configurations}, $\bigcup\F$ its {\em events} and write $\F^o$ for its {\em finite configurations}.   

A {\em map} $f:\F\to \G$  between stable families $\F$ and $\G$ is a partial function $f$ from the events of $\F$ to those of $\G$ such that 
for all   $x\in\F$
its direct image $f x\in\G$  and 
if  $e, e' \in x$ and $f(e) =f(e')$ then  
$e=e'$. The choice of map ensures an inclusion functor from the category of event structures to that of stable families.  The inclusion functor has a right adjoint $\Pr$ giving a coreflection (an adjunction with unit an isomorphism).  The construction $\Pr(\F)$ essentially replaces  the original events of a stable family $\F$ by the minimal, prime configurations at which they occur
.  
Let $x$ be a configuration of a stable family $\Fam$. 
Define the {\em prime} configuration of $e$ in $x$ by
$$[ e ]_x \eqdef \hbox{$\bigcap$} \set{y\in \Fam}{  e\in y  \ \& \  y\subseteq x}\ .
$$
By  coincidence-freeness, the  function 
  $\max:\iconf{\Pr({\F})}\to \F$ which takes  
a prime configuration $[e]_x$ to 
$e$ is well-defined;   it is  the 
counit  of the adjunction~\cite{icalp82,evstrs}. 

\begin{theorem}\label{thm:Pr}
Let $\Fam$ be a stable family.   Then,
$\Pr(\Fam) \eqdef (P, \Con, \leq)$ is an event structure where
$$
\eqalign{
&P= \set{[ e ]_x }{e\in x\ \&\ x\in \Fam}\ ,\cr
&
Z \in \Con  \hbox{ iff }  Z \subseteq P \ \&\  \hbox{$\bigcup$} Z \in\Fam\,,  \hbox{ and}\cr
&
p\leq p' \hbox{ iff }  p, p'\in P\ \&\ p\subseteq p'\,.}
$$
There is an order isomorphism $\theta: (\iconf{\Pr(\Fam)},\subseteq) \iso (\Fam, \subseteq)$
where $\theta(y)\eqdef\max\, y = 
\bigcup y$ for $y\in \iconf{\Pr({\F})}$
; its  mutual inverse is $\varphi$ where $\varphi(x) = \set{[e]_x}{e\in x}$ for  $x\in\Fam$.
\end{theorem}
 
 The partial orders represented by configurations under inclusion 
 are the same whether for event structures or stable families. They are 
 G\'erard Berry's {\em dI-domains}~\cite{berry,icalp82,evstrs}.

\subsection{Hiding---the defined part of a map}
Let $(E,\leq, \Con)$ be an event structure.  Let $V\subseteq E$ be a subset of `visible' events.
Define   
the {\em projection} 
on $V$, by 
$
E{\mathbin\downarrow} V\eqdef
(V, \leq_V, \Con_V)
$, 
where
$v \leq_V v' \hbox{ iff } v\leq v' \ \&\ v,v'\in V$ and $X\in\Con_V \hbox{ iff }  X\in\Con\ \&\ X\subseteq V$. 
The operation {\em hides} all events outside $V$.  It is associated with a {\em partial-total factorization system}. 
Consider a partial map of event structures $f:E\pto E'$.  Let 
$$V\eqdef \set{e\in E}{ f(e) \hbox{ is defined}}\,.$$
Then $f$ clearly factors into the composition 
$$
\xymatrix{
E\ar@^{->}[r]^{f_0}& E{\mathbin\downarrow} V \ar[r]^{f_1}& E'\\}
$$
of $f_0$, a partial map of event structures taking $e\in E$ to itself if $e\in V$ and undefined otherwise, and $f_1$, a total map of event structures acting like $f$ on $V$. 
Note that any $x\in \iconf{E{\mathbin\downarrow} V}$ is the image under $f_0$ of a {\em minimum} configuration, \viz~$[x]_E\in\iconf E$.
We call $f_0$ a {\em projection} and $f_1$ the {\em defined part} of the 
map 
$f$.    
 
\subsection{Pullbacks}
The coreflection from event structures to stable families is a considerable aid in constructing limits in the former from limits in the latter. 
The {\em pullback} of total maps of event structures is essential in composing strategies. We can define it via the pullback of stable families, obtained as
a  stable family of {\em secured bijections}. Let $\sig:S\to B$ and $\tau:T\to B$ be total maps of event structures. There is a composite bijection 
$$
\theta: x\iso \sig x = \tau y \iso y\,,
$$
between $x\in\iconf S$ and 
 $y\in \iconf T$ such that $\sig x = \tau y$; because $\sig$ and $\tau$ are total they induce bijections between configurations and their image. The bijection is {\em secured}  when
the transitive relation generated on $\theta$ by  $(s, t) \leq (s', t')$ if $s\leq_S s'$ or $t\leq_T t'$ is a finitary partial  order.

\begin{theorem}\label{thm:securedbijns}
Let $\sig:S\to B$, 
$\tau:T\to B$ be total maps of event structures.  The family $\mathcal R$ of secured bijections between $x\in\iconf S$ and 
 $y\in \iconf T$ such that $\sig x = \tau y$ is a  stable family.  
 The functions $\pi_1: \Pr({\mathcal R})\to S$, 
 $\pi_2:\Pr({\mathcal R})\to T$, taking a secured bijection with top to, respectively, the left and right components of its top, are maps of event structures.
$ \Pr({\mathcal R})$ with $\pi_1$, 
$\pi_2$ is the pullback of $\sig$, 
$\tau$ in the category of event structures.
 \end{theorem}

\begin{notation}{\rm W.r.t.~$\sig:S\to B$ and $\tau:T\to B$,  define $x\wedge y$ to be the configuration of their pullback which corresponds via this isomorphism to a secured bijection between $x\in\iconf S$ and $y\in \iconf T$, necessarily with $\sig x = \tau y$; any configuration of the pullback takes the form  $x\wedge y$  for unique $x$ and $y$. } 
 \end{notation} 

\section{Concurrent games and strategies}
The driving idea is to replace the traditional role of game trees by that of event structures. 
Both games and strategies will be represented by
  {\em event structures with polarity}, which  comprise $(A,\pol_A)$ where $A$ is an event
structure and a polarity function $\pol_A:A\to \{+,-,0\}$ ascribing a
polarity + (Player) or $-$ (Opponent) or 0 (neutral) to its events. The events
correspond to (occurrences of) moves. It will be technically useful to
allow events of neutral polarity; they arise, for example, in a play
between a strategy and a counterstrategy. Maps are those 
of event structures which preserve polarity. A {\em game }is represented by an event structure with polarities restricted to + or $-$, with no neutral events. 

\begin{definition}\label{def:scottorder} {\rm  In an event structure with polarity, with
configurations $x$ and $y$, write $x \subseteq^-y$ to mean inclusion in which all the
intervening events $y\setdif x$ are Opponent moves. Write $x  \subseteq^+ y$ for inclusion in
which the intervening events are neutral or Player moves. 
For a subset of events $X$ we write $X^+$ and $X^-$ for its restriction to Player and Opponent moves, respectively. 
The {\em Scott order}
  will play a central role: between $x,y\in\iconf A$, where $A$  is a game, define
  $$
   y\bel_A x \iff  \exists z\in \iconf A.\  y\supseteq^- z \ \&\ z\subseteq^+ x\,
   $$
  ---it is not hard to show that $z$ is unique and equal to $x\cap y$. 
The Scott order is also characterised by  
$$ y\bel_A x \iff  y^-\supseteq x^-\ \&\  y^+\subseteq x^+\,,
 $$
 which makes clear why it is a partial order.
 The Scott order is 
so named because it reduces to Scott's order on functions in special cases and plays a central role in relating games to Scott domains and ``generalised domain theory''~\cite{hyland-gendomthy}.
} \end{definition}

There are two
fundamentally important operations on games. One is that of
forming the {\em dual game}. On a game $A$ this amounts to reversing the polarities
of events to produce the dual $A^\perp$. The other operation, a {\em simple parallel
composition} $A\vvbar B$, is achieved on games $A$ and $B$ by simply juxtaposing them,
ensuring a finite subset of events is consistent if its overlaps with the
two games are individually consistent; any configuration $x$ of $A\vvbar B$ decomposes into
 $x_A\vvbar x_B$ where $x_A$ and $x_B$  are configurations of $A$ and $B$ respectively.
  
A {\em strategy} {\em in} a game $A$ is a total map $\sig: S \to A$ of \esswp~such that 
 \begin{itemize}
 \item[{(i)}] if $\sigma x  \subseteq^- y$, for $x\in\iconf S, y\in\iconf A$,  
 there is a unique $x'\in\iconf S$ with $x\subseteq x'$ and $\sig x' = y$;  
 \item[{(ii)}] if 
$s\imc_S s'$ and ($\pol(s) = +$ or $\pol(s') = -$ )\,, then 
$ \sigma(s)\imc_A \sigma(s')$.
\end{itemize}
The conditions 
 prevent Player from constraining Opponent’s behaviour 
 beyond the constraints of the game. 
Condition (i)  is  {\em receptivity}, ensuring that the strategy is open to all moves of Opponent permitted by the game. Condition (ii), called {\em innocence} in~\cite{FP}, ensures that the only additional immediate causal dependencies a strategy can enforce beyond those of the game are those  in which a Player move awaits moves of Opponent.  A map $f:\sig \Rightarrow \sig'$ of
strategies $\sig:S\to A$ and $\sig' :S' \to A$ is a map $f:S\to S'$ 
such that $\sig=\sig'f$; this determines when 
strategies are isomorphic.  

Following~\cite{conway, joyal}, a {\em strategy  from} a
game $A$
{\em to} a game $B$ is a strategy in the game $A^\perp \vvbar B$. 
Given a
 strategy from $B$ to a game $C$, so in $B^\perp\vvbar C$, we compose the two strategies essentially by playing them against each other in the common game $B$, where if one strategy makes a Player move the other sees it as a move of Opponent. 
The conditions of receptivity and innocence 
precisely 
ensure that the copycat strategy behaves as 
identity w.r.t.~composition, detailed below~\cite{lics11}. 


\subsection{Copycat}
 Let $A$ be a game. The copycat strategy $\cc_A:\CC_A\to A^\perp\vvbar A$ is an instance of a
strategy from $A$ to $A$. The event
structure $\CC_A$ is based on the idea that Player moves in one
component of the game $A^\perp\vvbar A$ always copy 
corresponding moves of
Opponent in the other component. For $c \in A^\perp\vvbar A$ we use $\bar c$ to mean the corresponding copy of $c$, of opposite polarity, in the
alternative component. The event structure $\CC_A$ comprises   $A^\perp\vvbar A$  with extra causal dependencies 
$\bar c\leq c$ for all  
events $c$ with $\pol_{A^\perp\vvbar A}(c) = +$; with the original causal dependency they generate a partial order; 
a finite subset 
is consistent in $\CC_A$
iff its down-closure w.r.t.~
$\leq$ is consistent in $A^\perp\vvbar A$. 
The map $\cc_A$ acts as  the identity function. 
 In characterising the configurations of $\CC_A$ we recall the Scott order of Defn
 ~\ref{def:scottorder}.
  \begin{lemma}\label{lem:CCbel} Let $A$ be a game.  Let
  $x\in\iconf{A^\perp}$ and $y\in\iconf A$.  Then  
 $$x\vvbar y \in\iconf{\CC_A} \ \hbox{ iff }\ y \bel_A  x\,.$$
   \end{lemma}
\subsection{Composition}
Two strategies $\sig:S\to A^\perp\vvbar B$ and $\tau:T\to B^\perp \vvbar C$
compose via pullback and hiding, summarised below.
$$\small\xymatrix@R=12pt@C=10pt{
 &\ar[dl]_{\pi_1
 }T\sncirc S \ar@^{-->}[rr]
\ar@{..>}[dd]|{\tau\sncirc\sig}
\pb{270}\ar[dr]^{\pi_2
 }&  & T\scirc S \ar@{-->}[dd]^{\tau\scirc \sig}\\
S\vvbar C\ar[dr]_{\sig\vvbar C}&&\ar[dl]^{A\vvbar \tau}A\vvbar T\\
 &A\vvbar B\vvbar C\ar@^{->}[rr]&
 &A\vvbar C}$$
Ignoring polarities, by forming the pullback of $\sig\vvbar C$ and $A\vvbar \tau$ we obtain the synchronisation of
complementary moves of $S$ and $T$ over the common game $B$;  subject to the causal constraints of $S$ and $T$, the effect is to instantiate the Opponent moves of $T$ in $B^\perp$ by the corresponding Player moves of $S$ in $B$, and {\it vice versa}.  Reinstating polarities we obtain the {\em interaction}  of $\sig$ and $\tau$  
$$
\tau\sncirc \sig: T\sncirc S \to A^\perp \vvbar B^0 \vvbar C\,,
$$
where we assign neutral polarities to all moves in or over $B$. Neutral moves over the common part $B^0$ remain unhidden. The map $A^\perp \vvbar B^0\vvbar C\pto A^\perp\vvbar C$  is undefined on $B^0$ and otherwise mimics the identity. Pre-composing this map with $\tau\sncirc \sig$ we obtain a partial map $T\sncirc S\pto A^\perp\vvbar C$; it is undefined on precisely the neutral events of $T\sncirc S$.  The defined parts of its partial-total factorization yields 
 $$\tau\scirc \sig: T\scirc S\to A^\perp\vvbar C\,$$ 
 ---this is the {\em composition} of $\sig $ and $\tau$. 
 
 \begin{notation}{\rm For $x\in\iconf S$ and $y\in\iconf T$, 
let $\sig x = x_A\vvbar x_B$ and $\tau y = y_B\vvbar y_C$ where $x_A\in \iconf A$, $x_B, y_B\in\iconf B$, $y_C\in\iconf C$.  
Define
$
y\sncirc x = (x\vvbar y_C) \wedge (x_A\vvbar y)
$.
This is a partial operation only  defined if the $\wedge$-expression is. 
It is defined and glues configurations $x$ and $y$ together at their common overlap over $B$ provided $x_B=y_B$ and a finitary partial order of causal dependency results. 
Any 
 configuration of $T\sncirc S$ has the form 
$
y\sncirc x
$,
for unique $x\in\iconf S, y\in\iconf T$.   
 } \end{notation}

\subsection{A bicategory of strategies}
 We obtain a bicategory $\Strat$ for which the objects are games,  the arrows $\sig: A\profto B$  are strategies $\sig:S\to A^\perp\vvbar B$;
 with 2-cells $f:\sig \Rightarrow \sig'$ maps of strategies. The vertical composition of 2-cells is the usual composition of maps.  Horizontal composition 
is 
the composition of strategies $\scirc$ (which extends to a functor 
via the universality of pullback and partial-total factorisation).  We can restrict the 2-cells to be rigid maps and still obtain a bicategory. 
The bicategory of strategies is compact-closed, 
though with the addition 
of winning conditions---Section~\ref{sec:winningcond}---this weakens to $*$-autonomous.   

A strategy  $\sig:S\to A$ is {\em deterministic} if $S$ is deterministic, \viz 
$$
\all X\fsubseteq S.\ [X]^-\in \Con_S \implies X\in\Con_S\,,
$$
where
$[X]^- \eqdef \set{s'\in S}{\exists s\in X.\ \pol_S(s')=- \ \&\ s'\leq s}$.
So, a strategy is deterministic if consistent behaviour of Opponent is answered by consistent behaviour of Player.
Copycat $\cc_A$ is deterministic iff the game $A$ is {\em race-free}, \ie~if $x\subseteq^- y$ and  
$x\subseteq^+ z$ in $\iconf A$
then $y\cup z\in \iconf A$. 
    The bicategory of strategies restricts to a bicategory of deterministic strategies between race-free games.

    There are several ways to reformulate strategies.  Deterministic strategies coincide with the {\em receptive} ingenuous strategies of Melli\`es and Mimram based on asynchronous transition systems~\cite{Asgames,DBLP:journals/fac/Winskel12}. Via the Scott order, we can see strategies  as a refinement of  profunctors:  a strategy in a game $A$  induces a discrete fibration, so presheaf, on $(\conf A, \bel_A)$, a construction which extends to strategies between games~\cite{fossacs13}.      
 
 
 \subsection{Winning conditions}\label{sec:winningcond}
 
{\em Winning conditions} of a game $A$ specify 
a subset  $W$ of  its  configurations, an outcome in which 
is a win for Player.  Informally, a
 strategy (for Player) is  {\em winning}    if it   always prescribes moves for  Player to  end up in a winning configuration, no matter what the activity or inactivity of Opponent. 

 Formally, a game with winning conditions $(A, W_A)$ comprises a concurrent game $A$ with winning conditions $W_A\subseteq \iconf A$. 
 A strategy  $\sig:S\to A$   is {\em winning}  if  $\sig x$ is in $W_A$ for all   +-maximal configurations $x$ of $S$; in general, a configuration is +-maximal if 
 no additional Player, or neutral, moves can occur from it. That $\sig$ is winning can be shown  equivalent to:  all plays of $\sig$ against any counterstrategy of Opponent result in a win for Player~\cite{lics12,ecsym-notes}.

As the dual of a game with winning conditions $(A, W_A)$ we again reverse the roles of Player and Opponent, 
and take its  winning conditions to be the  
set-complement  of $W_A$, \ie~$(A,W_A)^\perp = (A^\perp, \iconf A \setdif W_A)$. 

In a 
parallel composition  
of games with winning conditions, 
we deem
a configuration $x$ of ${A\vvbar B}$ 
winning if  its component  $x_A$ is winning in $A$ or its  component  $x_B$ is winning in $B$: 
$(A, W_A) \vvbar (B, W_B) \eqdef (A\vvbar B, W)$
where 
$W= \set{x\in\iconf{A\vvbar B}}{x_A\in W_A \hbox{ or } x_B \in W_B}$.  

With these extensions, we take a winning strategy from a game $(A,W_A)$ to a game $(B,W_B)$, 
to be a winning strategy in the game $A^\perp\vvbar B$ ---its winning conditions form the set $$\set{x\in\iconf{A^\perp\vvbar B}}{x_A\in W_A \Rightarrow x_B \in W_B}\,.$$  When games are race-free, copycat will  be a winning strategy.
The composition of winning strategies is winning~\cite{lics12,ecsym-notes}.  In the proof the following lemma is 
 critical:  
\begin{lemma}\label{lem:+max} Let $\sig:S\to A^\perp \vvbar B$ and 
$\tau:T\to B^\perp\vvbar C$ be strategies.  Suppose  
$y\sncirc x\in \iconf{T\sncirc S}$ where 
$x\in\iconf S$ and $y\in\iconf T$.  Then,
$y\sncirc x$
 is +-maximal iff both  
$x$ and $y$ are +-maximal.  
\end{lemma}
One can extend winning conditions to payoff functions~\cite{payoff} or to allow draws, where neither player wins~\cite{dexter}.

\subsection{Imperfect information}
 In  a game of  {\em imperfect information}  
 some moves are masked, or inaccessible,  
 and strategies with dependencies on unseen moves are ruled out. One can 
 extend   games with imperfect information  in a way that respects the operations of concurrent games and strategies~\cite{dexter}. 
 Each move of a game is assigned a level in a global order of access levels;  moves of the game or its strategies can only causally depend on moves at equal or lower levels.  
 
 In more detail, a fixed preorder of {\em access levels}
  $(\Lambda, \preceq)$ is pre-supposed.   
  A   {\em $\Lambda$-game} comprises a game $A$  with a {\em level function} $l:A\to \Lambda$ such that if
$a\leq_A a' $ then  
$ l(a)\preceq l(a')
$ 
 for all moves $a, a'$ in $A$.  A  {\em $\Lambda$-strategy} in the $\Lambda$-game is a   strategy $\sig:S\to A$ for which 
if $
s\leq_S s' $
then 
$l\sig(s) \preceq l\sig(s')
$
  for all $s, s'$ in $S$.  
 The access levels of moves in a game are left undisturbed in forming the dual and parallel composition of $\Lambda$-games. 
 As before, a $\Lambda$-strategy from a  $\Lambda$-game $A$ to a  $\Lambda$-game $B$ is a  $\Lambda$-strategy in the game $A^\perp\vvbar B$.  It can be shown that $\Lambda$-strategies compose~\cite{dexter}.

\subsection{A language for strategies}\label{sec:corelang}
We recall briefly the language for strategies introduced in~\cite{stratasproc}.  
Games  $A, B, C, \cdots$ play the role of types.  Operations on games include forming the dual $A^\perp$, 
simple parallel composition $A\vvbar B$, a
sum $\Sigma_{i\in I} A_i$ as well as recursively-defined games. 

Terms, denoting strategies, 
have typing judgements
\begin{center}$
x_1:A_1, \cdots, x_m:A_m \vdash\  t \ \dashv y_1:B_1, \cdots, y_n:B_n\ ,
$
\end{center}
where all the variables are distinct,
interpreted as a strategy from  
$\vec A=
 A_1 \vvbar \cdots \vvbar A_m$   to  
 $\vec B= B_1 \vvbar \cdots \vvbar B_n$.
 We can picture the term $t$ as a box with input and output wires for the  variables $\vec x$ and $\vec y$:
 {\small \begin{center}
 \setlength{\unitlength}{2947sp}%
\begingroup\makeatletter\ifx\SetFigFont\undefined%
\gdef\SetFigFont#1#2#3#4#5{%
  \fontfamily{#3}\fontseries{#4}\fontshape{#5}%
  \selectfont}%
\fi\endgroup%
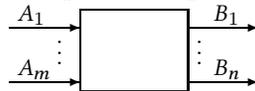
\begin{picture}(2124,699)(889,-598)
\thinlines
{\put(1501,-586){\framebox(900,675){}}
}%
{\put(901,-61){\vector( 1, 0){600}}
}%
{\put(2401,-511){\vector( 1, 0){600}}
}%
{\put(901,-511){\vector( 1, 0){600}}
}%
{\put(2401,-61){\vector( 1, 0){600}}
}%
\put(976,0){\makebox(0,0)[lb]{\smash{{\SetFigFont{9}{14.4}{\rmdefault}{\mddefault}{\updefault}{$A_1$}%
}}}}
\put(976,-451){\makebox(0,0)[lb]{\smash{{\SetFigFont{9}{14.4}{\rmdefault}{\mddefault}{\updefault}{$A_m$}%
}}}}
\put(2626,0){\makebox(0,0)[lb]{\smash{{\SetFigFont{9}{14.4}{\rmdefault}{\mddefault}{\updefault}{$B_1$}%
}}}}
\put(2626,-451){\makebox(0,0)[lb]{\smash{{\SetFigFont{9}{14.4}{\rmdefault}{\mddefault}{\updefault}{$B_n$}%
}}}}
\put(2476,-361){\makebox(0,0)[lb]{\smash{{\SetFigFont{12}{14.4}{\rmdefault}{\mddefault}{\updefault}{$\vdots$}%
}}}}
\put(1301,-361){\makebox(0,0)[lb]{\smash{{\SetFigFont{12}{14.4}{\rmdefault}{\mddefault}{\updefault}{\vdots}%
}}}}
\end{picture}%
\end{center} }
 
 The term  
$t$ denotes a strategy $\sig: S\to \vec A^\perp \vvbar \vec B$. It does so by
 describing witnesses,  configurations of $S$, to a
relation between  configurations $\vec x$ of $\vec A$ and $\vec y$ of
$\vec B$.  
For example, the term 
$$x:A \vdash y\bel_A x \dashv y:A$$ 
denotes the copycat strategy on a game $A$; it describes configurations of copycat, $\CC_A$,  
as 
witnesses, \viz~ 
those 
configurations $x\vvbar y$ of $\CC_A$ 
for which $y\bel_A x$ in the Scott order. 
There are other operations, such as sum $\fatbar$ and pullback $\wedge$ on strategies of the same type.

 \noindent
 Duality 
  is caught by the rules
$$
\pfrule{\Gamma, x:A \vdash t \dashv  \Delta}{ \Gamma \vdash t \dashv x:{A}^\perp, \Delta}
 \qquad
\pfrule{\Gamma \vdash t \dashv x:A, \Delta}{ \Gamma, x: {A}^\perp \vdash t \dashv \Delta}
$$
and composition of   strategies by 
$$
\pfrule{{\Gamma \vdash t \dashv \Delta \qquad \Delta \vdash u \dashv \Eta} }{\Gamma \vdash 
\stcomp\Delta t u
 \dashv \Eta}
$$
which, in the picture of strategies as boxes, joins the output wires of one strategy to input  wires of the other.   Simple parallel composition of strategies arises 
when $\Delta$ is empty.

\section{From functions to strategies}\label{sec:part1}
The language for strategies in~\cite{stratasproc} included a judgement
$$
x:A \vdash g (y) \bel_{C} f (x) \dashv y:B
$$
for building strategies out of expressions $f(x)$ and $g(y)$ denoting ``affine functions.''   It breaks down into a composition 
$$
x:A \vdash 
\stcomp{z:C }{ g (y)\bel_C  z}{  z\bel_{C} f (x) }
\dashv y:B\,.
$$ 
Here we present a considerably broader class of affine-stable functions $f$  and ``co-affine-stable'' functions $g$ with which to define strategies in this manner. 
It hinges on the Scott order to 
convert functional dependency to
causal dependency, in the sense captured by Theorem~\ref{thm:affstmaptostrats} below. 
 

\subsection{Affine-stable maps and their strategies}\label{sec:affstab}

\begin{definition}\label{def:affine-stable}{\rm 
An {\em affine-stable map} between games from $A$ to $B$, is 
a  function $f: \iconf A \to \iconf B$ which is\\
$\bullet$ 
{\em polarity-respecting:} for $x, y\in\iconf A$, 
$$
x\subseteq^- y\Rightarrow f(x)\subseteq^- f(y)
\ \hbox{ and } \   
x\subseteq^+ y\Rightarrow f(x)\subseteq^+ f(y)\,;
$$
$\bullet$ 
{\em +-continuous:} for $x\in\iconf A$, 
$$
b\in f(x) \ \&\ \pol_B(b) = + \Rightarrow \exists x_0\in\conf A.\  x_0\subseteq x\ \&\ b\in f(x_0)\,;
$$
$\bullet$ 
 {\em $-$-image finite:} for all finite configurations 
$
x\in\conf A$    the set  $ f(x)^-$   is finite; \\
$\bullet$ 
 {\em affine:} for all compatible  families $\set{x_i}{i\in I}$ in $ \iconf A$, 
$$
\hbox{$\bigcup$}_{i\in I} f (x_i) \subseteq^+ f( \hbox{$\bigcup$}_{i\in I} x_i)\,
$$
---when $I$ is empty this  amounts to $\emptyset \subseteq^+ f(\emptyset)$; and \\
$\bullet$ 
 {\em stable:}  for all 
 compatible  families $\set{x_i}{i\in I}\neq \emptyset$ in $ \iconf A$, 
$$
 f( \hbox{$\bigcap$}_{i\in I} x_i) \subseteq^-  \hbox{$\bigcap$}_{i\in I} f ( x_i)\,.
$$
}\end{definition}

When all the moves of games $A$ and $B$ are those of Player, the definition reduces to that of 
stable function.
If 
all moves are those of Opponent, it becomes that of demand maps---see Section~\ref{sec:stabspans}~\cite{Visions}.  Affine-stable maps include maps of event structures with polarity, including partial maps between games, and the affine maps of~\cite{stratasproc}.   They are 
the most general maps out of which we can construct a corresponding strategy, in a  way we now describe.

\begin{theorem}\label{thm:affstmaptostrats}
Let $f: \iconf A \to \iconf B$ be an affine-stable map between games $A$ and $B$.    
Then 
$$
\Fam \eqdef \set{
x\vvbar y\in \iconf{A^\perp\vvbar B}}{y\bel_B f(x)}
$$
is a 
stable family.  
The map $\max:\Pr(\Fam)\to A^\perp\vvbar B$ is a strategy $f_!:A\profto B$.  The strategy $f_!$ is deterministic if  $A$ and $B$ are race-free and $f$ reflects $-$-compatibility, \ie~$x\subseteq^- x_1$ and  $x\subseteq^-  x_2$ in\,   $\iconf A$
 and
 $f x_1 \cup f x_2\in \iconf B$ implies  
 $x_1\cup x_2\in \iconf A$.  
\end{theorem}

The theorem above explains how to convert functional dependency, expressed as $y \bel_B f(x)$, to causal dependency between moves $\Pr(\Fam)$, obtained as primes of the stable family $\Fam$. 
The expression of functional dependency as causal dependency is quite subtle; the direction of causal dependency hinges critically on the polarities of events.

For $f$  an affine-stable map from $A$ to $B$ we can write $f_!$ as 
$$
x:A \vdash y \bel_{B} f(x) \dashv y:B\,.
$$
Through suitable 
$f$ we can create strategies from structural maps,  injections and projections as strategies,  for 
conditional and case statements and, generally, much of 
the causal wiring that is often explained informally in 
diagrammatic reasoning. 
If $\sig$ is a strategy in $A$ then $f_!\scirc \sig$ is its
``pushforward'' to a strategy in $B$.
Some 
basic examples:

\begin{example}\label{ex:projection}{\rm ({\it Projectors})  
Let $f:A\vvbar B \to B$ be the function undefined on game $A$ but acting as identity on game $B$.  Let $\sig$ be a strategy in the game $A\vvbar B$.  The  strategy $f_!\scirc \sig$ is its projection to a strategy in $B$. 
}\endex\end{example}

\begin{example}\label{ex:duplication}{\rm ({\it Duplicators})  
Let $A$ be a game.  Consider the function $d_A: x\mapsto x\vvbar x$ from $\iconf A$ to $\iconf{A\vvbar A}$.  It is easily checked to be affine-stable.  Hence there is a {\em duplicator} strategy $\delta_A={d_A}_!: A\profto A\vvbar A$.  (The strategy $\delta_A$ is not natural in $A$ 
 as $\vvbar$ is not a product, except in subcategories.) 
}\endex\end{example}

\begin{example}\label{ex:detector}{\rm {\it (Detectors)}
Let $A$ be a game.  Let $X\in\Con_A$ with $X\subseteq A^+$.  Let $\plmove$ be a single ``detector'' event, of +ve polarity.  Let 
$$
d_X:\iconf A \to \iconf\plmove
$$
be the function such that
$$
d_X(x) =
\begin{cases}
\plmove &  \hbox{ if } X\subseteq x\,,\\
\emptyset &  \hbox{ otherwise. }
\end{cases}
$$
The function $d_X$ is affine-stable. There is a {\em detector} strategy
$$
{d_X}_!: A \profto \plmove\,.
$$
The strategy 
simply adjoins extra causal dependencies $a\imc \plmove$ from $a\in X$.  It detects the presence of $X$.  In a similar way, one can extend detectors to detect the occurrence of one of a family $\langle X_i\rangle_{i\in I}$ of $X_i\in\Con_A$ provided 
$
X_i\cup X_j \in \Con_A \implies i=j
$
for $i,j\in I$.  
}\endex\end{example}

\begin{example}\label{ex:inhibitor}{\rm ({\it Blockers})
Let $A$ be a game and $Y\subseteq A^-$. Let 
$$
h_Y : \iconf A \to \iconf\opmove
$$
be the function which acts so
$$
h_Y(x) = 
\begin{cases}\opmove &\hbox{ if } x\cap Y\neq \emptyset\,,\\
\emptyset & \hbox{ otherwise.}
\end{cases}
$$
$h_Y$ is a map of event structures so affine-stable.  
The {\em blocker} strategy ${h_Y}_!$ 
adjoins causal dependencies $\opmove \imc a$ from $\opmove$ to each 
$a\in Y$.  The absence of move $\opmove$ blocks 
all moves $Y$.  
}\endex\end{example}
 

  
\begin{theorem}\label{thm:functor}
The operation $(\_)_!$ is a (pseudo) functor from the category of affine-stable maps  to concurrent strategies $\Strat$.  
\end{theorem}

\subsection{co-Affine-stability}
We examine the dual, or co-notion, to affine-stability.  
An affine-stable map $f$ from $A^\perp$ to $B^\perp$ yields a strategy  $f_!: A^\perp\profto B^\perp$, so by duality a strategy 
$f^*:B\profto A$.  We obtain 
the dual to Theorems~\ref{thm:affstmaptostrats},\ref{thm:functor}  as a corollary: 
\begin{corollary} Let $g:\iconf A\to \iconf B$ be such that 
 $g:\iconf{A^\perp}\to \iconf{B^\perp}$ is affine-stable, then
$$
{\mathcal G} \eqdef \set{y\vvbar x\in \iconf{B^\perp\vvbar A}}{g(x)\bel_B y}
$$
is a 
stable family.  
The map $\max:\Pr({\mathcal G})\to B^\perp\vvbar A$ is a strategy $g^*:B\profto A$.  The strategy $g^*$ is deterministic if  $A$ is race-free and $g$ reflects $+$-compatibility. The operation 
$(\_)^*$ is a contravariant (pseudo) functor from the category of affine-stable maps  to $\Strat$. 
\end{corollary}

For $g$  an  affine-stable map from $A^\perp$ to $B^\perp$ we can write $g^*$ as 
$$
y:B \vdash g (y) \bel_{B} x \dashv x:A\,.
$$
For a strategy $\sig$ in game $B$ the operation $g^*\scirc \sig$ yields the strategy in $A$ got as the pullback of $\sig$ along $g$. In particular, if $A$ prefixed game $B$  by some initial move,  $g^*\scirc \sig$ would be a prefix operation on strategies.  
 
\subsection{An adjunction}\label{sec:adjn}
An affine-stable map $f$ from $A$ to $B$ is not generally 
an affine-stable map from $A^\perp$ to $B^\perp$.
 The 
 next definition, of an {\em additive-stable} map  $f$ from $A$ to $B$, bluntens affine-stability to ensure $f$ is  also a additive-stable map from $A^\perp$ to $B^\perp$; and hence is associated with both a strategy
$
f_!:A\profto B
$
and a converse strategy
$
f^*:B\profto A
$.  
Together they form an adjunction.

\begin{definition}{\rm 
A {\em additive-stable map} between \esswp, from $A$ to $B$, is 
a  function $f: \iconf A  \to \iconf B$ which is\\
$\bullet$ 
{\em polarity-respecting:} for $x, y\in\iconf A$, 
$$
x\subseteq^- y \Rightarrow f(x)\subseteq^- f(y)
\  \hbox{ and } \   
x\subseteq^+ y\Rightarrow f(x)\subseteq^+ f(y)\,;
$$
$\bullet$ 
 {\em  image finite:} if  
$
x\in\conf A$    then  $ f(x)\in\conf B$; \\
$\bullet$ 
 {\em additive:} for all  compatible  families $\set{x_i}{i\in I}$ in $ \iconf A$, 
$$
 \hbox{$\bigcup$}_{i\in I} f (x_i)= f(\hbox{$\bigcup$}_{i\in I} x_i)\,; 
$$
$\bullet$ 
  {\em stable:} for all compatible  families $\set{x_i}{i\in I}\neq\emptyset$ in $ \iconf A$, 
$$
  f(\hbox{$\bigcap$}_{i\in I} x_i) = \hbox{$\bigcap$}_{i\in I} f ( x_i)\,.
$$
}\end{definition}
The usual maps  of games are additive-stable, including those which are partial, 
as are Girard’s linear maps. 
Additive-stability is indifferent to a switch of polarities:

\begin{proposition}
An additive-stable map $f$ from $A$ to $B$ is an additive-stable map $f$ from $A^\perp$ to $B^\perp$ and {\it vice versa}.
\end{proposition}

Given an additive-stable map $f$ from $A$ to $B$ we obtain a strategy $f_!:A\profto B$ and, via $f$ from $A^\perp$ to $B^\perp$, 
$f^*:B\profto A$.   

\begin{theorem}\label{thm:adjn} Let $f$ be an additive-stable map from $A$ to $B$ between  \esswp. 
The strategies $f_!$ and $f^*$ form an adjunction $f_!\dashv f^*$ in the bicategory 
$\Strat$. 
\end{theorem}
 This says $\Strat$ forms a {\em pseudo double category}~\cite{2DCat,hugo-thesis}.
In Section~\ref{sec:GoIadjns} we apply Theorem~\ref{thm:adjn} to relate deterministic strategies in general, to those of Geometry of Interaction. 
The adjunction in 
$\Strat$ of  Theorem~\ref{thm:adjn} yields a traditional adjunction:  

\begin{corollary}
 Let $f$ be an additive-stable map from game $A$ to game $B$.  Let $\Strat_A$ be the 
 category of strategies in the game $A$, and $\Strat_B$ that in $B$.  Then there are functors $f_!\scirc (\_):\Strat_A\to \Strat_B$ and $f^*\scirc (\_):\Strat_B\to \Strat_A$ with 
$f_!\scirc (\_)$ left adjoint to $f^*\scirc (\_)$.  
\end{corollary}

\section{From strategies to functions}\label{sec:part2}
We recover familiar notions of games from those based on event structures. 
A game is {\em tree-like} when any two events are either inconsistent or causally dependent.    When such a game is
race-free,   at any finite configuration, the next  possible moves, if there are any, belong purely to Player, or purely to Opponent.
Then, 
at each position where Player may move, a  deterministic strategy either chooses a unique move or to stay put.  In contrast to many presentations of games, in a concurrent strategy Player isn't forced to make a move, though that can be encouraged through suitable winning conditions.  
A counterstrategy, as a strategy in the dual game, picks moves for Opponent at their configurations.  
The interaction $\tau\sncirc\sig$ of a deterministic strategy $\sig$ with a deterministic counterstrategy $\tau$  determines a finite or infinite branch in the tree of configurations, which in the presence of winning conditions will be 
a win for one of the 
players.   

On tree-like games we recover familiar notions. 
More surprising is that by exploiting the 
richer structure of concurrent games we can recover other familiar paradigms, not traditionally tied to games, or if so only somewhat informally.
 We start by 
 rediscovering Berry's stable domain theory, of which Jean-Yves Girard's qualitative domains and coherence spaces are special cases.  
 The other examples, from dataflow, logic and functional programming, 
 concern ways
of handling  interaction within a functional approach.   
We shall
restrict to race-free games, so guaranteeing that deterministic strategies have an identity w.r.t.~composition, given by copycat. 


\subsection{Stable functions} 

Consider games in which all 
moves are Player moves. 
Consider a strategy $\sig$ from one such purely Player game $A$ to another $B$.  This is a map $\sig: S\to A^\perp\vvbar B$ which is receptive and innocent.  
Notice that  in $A^\perp\vvbar B$  all the Opponent moves are in $A^\perp$ and all the   Player moves are in $B$.  By receptivity any configuration of $A$ can be input.  The only new immediate causal connections, beyond those in $A^\perp$ and  $B$,  that can be introduced in a strategy are those from 
Opponent moves of $A^\perp$ to a Player move in $B$.  Beyond the causal dependencies of the games, a  strategy $\sig$ can only make a Player move in $B$ causally depend on  a finite subset of 
moves in $A^\perp$.  

When $\sig$ is deterministic, all conflicts are inherited from  conflicts between Opponent moves.  Then the strategy $\sig$ gives rise to a stable function from  the 
configurations of $A$ to the 
configurations of $B$.  Conversely, such a stable function $f$ yields a deterministic strategy $f_!:A\profto B$,  by Theorem~\ref{thm:affstmaptostrats}.  

\begin{theorem}
The category ${\bf dI}$ of dI-domains and stable functions, enriched by the stable order, is equivalent to the bicategory of deterministic strategies between purely Player games 
with rigid 2-cells. 
(The bicategory of deterministic strategies between purely Player games  
with {\em all} 2-cells is equivalent to the category of dI-domains enriched by the Scott---or pointwise---order.)
\end{theorem}

The category of dI-domains and stable functions is well-known to be cartesian-closed; its function space and product are 
realised by constructions $[A\to B]$ and   $A\vvbar B$ on event structures. 
When the games are further restricted to have trivial 
causal dependency we recover Girard's {\em qualitative domains}  and, with
conflict determined in a binary fashion, his {\em coherence spaces}.  
Girard's models for polymorphism there generalise to dI-domains, with
dependent types
$\Pi_{x:A} B(x)$ and $\Sigma_{x:A} B(x)$ on event structures~\cite{evstrs,dIPoly, CGW}---see Appendices~\ref{sec:Stfnspace}, \ref{app:deptypes}.

\subsection{Stable spans} \label{sec:stabspans}
When between games in which all the moves are Player moves, a general, 
nondeterministic, strategy corresponds to a {\em stable span}, a form of many-valued  stable function which has been discovered, and rediscovered, in giving semantics to higher-order processes and especially nondeterministic dataflow~\cite{mikkel,Visions,SEW}; the trace of strategies, derived from their compact closure, specialises to the  feedback operation of dataflow.  
Recall a {\em stable span} comprises 
$$
 \spano A\dem {E}\out{B\,,}
 $$
 with  event structure $E$ relating input given by an event structure $A$ and output by an event structure $B$. 
   The map
$\out: E\to B$ is a {\em rigid} map.  The map $\dem:E\to A$, associated to input, is of a different character.  It is a {\em demand} map, \ie, a function from $\iconf E$ to $\iconf A$ which preserves unions and finite configurations;  $\dem(x)$ is the minimum input for $x$ to occur and is the union of the demands of its events.  The occurrence of an event $e$ in $E$ demands  minimum input $\dem([e])$ and  is observed as  the output event $\out(e)$.  
Spans  from 
$A$ to 
$B$ are related by the usual 2-cells
, here  (necessarily) rigid maps $r$ making the diagram below commute:
$$
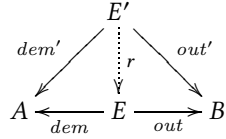
\xymatrix
{
 &\ar[ld]_{\dem'}E'\ar@{..>}[d]^r\ar[rd]^{\out'}&\\
 A&\ar[l]^\dem E\ar[r]_\out&B
 }
 $$
 
Stable spans compose via the usual pullback construction of spans
, as both demand and output maps extend to functions between configurations.
A stable span $E$ corresponds to a (special) profunctor $$\tilde E(x,y) =\set{w\in\conf E}{\dem(w) \subseteq x \ \&\ \out\,w= y}\,,$$
between the partial-order  categories $\conf A$ and  $\conf B$ ---a correspondence that respects composition.   
Recalling the view of profunctors as Kleisli maps w.r.t.~the presheaf construction~\cite{Kleisli}, we 
borrow from Moggi~\cite{Moggi} and 
describe the composition of stable spans
$
F:A\profto B$,  
$G:B\profto C$ 
as 
$$
G\stcirc F (x)
= 
\plet y \larrow F(x) \pin  G(y)  \,
$$
---which, via the correspondence with profunctors, stands for the coend
$
\int^{y\in\conf B} \tilde F(x,y) \times\tilde G(y,\_)$. In using let-notation we can take account of the shape of the configuration $y$ in the definition of $G$, 
in effect 
an informal pattern matching.  

Stable spans are monoidal closed~\cite{mikkel,lics02}: 
 w.r.t.~an event structure $A$, the functor $(\_\vvbar A)$ has a right adjoint, the function space $[A\multimap\_]$. The construction $[A\multimap B]$  is recalled in Appendix~\ref{app:stspans}
 along with a more general dependent product  $\sPi_{x:A}B(x)$ for stable spans:  the type of stable spans which on input $x:A$ yield output  $y:B(x)$ nondeterministically.   Stable spans are trace monoidal closed; their trace is described in~\cite{SEW}.

 Let $A$ and $B$ be purely Player games.  A strategy $\sig:S\to A^\perp\vvbar B$ gives rise to a stable span
 $$
 \xymatrix{
  A&\ar[l]_\dem{E}\ar[r]^\out&{B\,,}
  }
 $$
where $E= S^+$, and  
$\out$ gives the image of its events in $B$ and $\dem$ those events in $A$ on which they casually depend.  
Conversely given a stable span, as above, we obtain a strategy as the composition $\out_! \scirc \dem^*
$, by the results of Section~\ref{sec:part1}; as regarding $E$ and $A$ as purely Player games, both $\out$ and $\dem:\iconf{E^\perp}\to \iconf{A^\perp}$ are affine-stable. In the strategy $\out_! \scirc \dem^*:S\to A^\perp\vvbar B$ so obtained, $S$ comprises the disjoint union of $A$ and $E$  with the additional causal dependencies of $e\in E$ on $a\in A^\perp$ prescribed by 
 $\dem$. 

\begin{theorem}
The bicategory ${\bf Stab}$ of stable spans is equivalent to the bicategory of strategies between purely Player games with rigid  
2-cells. 
\end{theorem}
 
We show that in a similar way, we obtain geometry of interaction, dialectica categories, containers, lenses, open games and learners, optics and dependent optics by moving to slightly more complicated subcategories of  
games, sometimes with winning conditions and  imperfect information.

\subsection{Geometry of Interaction} \label{sec:GoI}
Let's now consider slightly more complex games.  A {\em GoI game} comprises a parallel composition $A:=A_1\vvbar A_2$  of  a purely Player game $A_1$ with a purely Opponent game  $A_2$.  Consider a strategy $\sig$ from a GoI game $A:= A_1\vvbar A_2$ to a GoI game $B:= B_1\vvbar B_2$.   Rearranging the parallel compositions, 
$$
 A^\perp\vvbar B
= 
A_1^\perp \vvbar A_2^\perp \vvbar B_1\vvbar B_2
\iso
  (A_1 \vvbar B_2^\perp)^\perp \vvbar (A_2^\perp\vvbar B_1)\,.
$$
So $\sig$, as a strategy in $A^\perp\vvbar B$,
 corresponds to 
a strategy from the purely Player game $A_1 \vvbar B_2^\perp$ to the purely Player game $A_2^\perp\vvbar B_1$.  We are back to the simple situation considered in the previous section, 
of strategies between purely Player games.

 Strategies between GoI games, from $A$ to $B$, correspond to stable spans from ${A_1 \vvbar B_2^\perp}$ to ${A_2^\perp\vvbar B_1}$.  The maps are familiar from models of geometry of interaction built as free compact-closed categories from  traced monoidal categories~\cite{joyal-street-verity,AbramskyHaghverdiScott},
though here lifted to the {\em bicategory} ${\bf Stab}$ of stable spans.
 
\begin{theorem}
 The bicategory of strategies on GoI games with rigid 2-cells  is equivalent to 
the free compact-closed bicategory built 
on the trace monoidal bicategory ${\bf Stab}$.
\end{theorem} 
 
 When deterministic, strategies from GoI game $A$ to GoI game $B$ 
 correspond to a stable function from $\iconf{A_1 \vvbar B_2^\perp}$ to $\iconf{A_2^\perp\vvbar B_1}$.
 Note that a configuration of a parallel composition of games splits into a pair of configurations: 
$$
{
\iconf{A_1 \vvbar B_2^\perp} 
\iso  \iconf{A_1}\times\,\iconf{B_2}, \ 
\iconf{A_2^\perp\vvbar B_1} 
\iso  \iconf{A_2}\times\,\iconf{B_1}.
}$$
Thus 
 deterministic strategies from  $A$ to $B$ correspond to  stable functions
 $$
 S=\langle g, f\rangle :  \iconf{A_1}\times\iconf{B_2} \to  \iconf{A_2}\times  \iconf{B_1}\,,
 $$
 associated with a pair of stable functions 
 $g:  \iconf{A_1}\times\iconf{B_2} \to \iconf{A_2}$   and  $f:  \iconf{A_1}\times\, \iconf{B_2} \to \iconf{B_1}$,    
summarised diagrammatically by:  
 $$
\xymatrix{
A_1\ar@/^1.1pc/[d]\ar@/_1.0pc/[r]_S&B_1\ar@{}[l]|f\\
 A_2&B_2\ar@/^1.1pc/[u]\ar@/_1.0pc/[l] \ar@{}[l]|g
 }
 $$
 Such maps are obtained by Abramsky and Jagadeesan's GoI construction, 
 here starting from {\em stable} domain theory~\cite{nfgoi}.  
 
 The composition of deterministic strategies between GoI games, $\sig$ from $A$ to $B$ and $\tau$ from $B$ to $C$ coincides with the composition of GoI  given  by ``tracing out'' $B_1$ and $B_2$.  Precisely, 
 supposing $\sig$ corresponds to the stable function
 $$
 S:  \iconf{A_1}\times\,\iconf{B_2}  \to   \iconf{A_2}\times\,\iconf{B_1}\,
 $$
 and $\tau$ to the stable function
 $$
 T:  \iconf{B_1}\times\,\iconf{C_2}  \to   \iconf{B_2}\times\,\iconf{C_1}\,,
 $$
 we see a loop in the functional dependency at $B$:
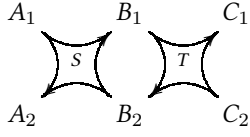
$$
\xymatrix{
A_1\ar@/^1.1pc/[d]\ar@/_1.0pc/[r]_S&{B_1}\ar@/^1.1pc/[d]\ar@/_1.0pc/[r]_T&C_1& \\
 A_2&B_2\ar@/^1.1pc/[u]\ar@/_1.0pc/[l]&C_2\ar@/^1.1pc/[u]\ar@/_1.0pc/[l]& 
 }
 $$
Accordingly, the composition $\tau\scirc \sig$ corresponds to the stable function taking $(x_1,z_2)\in  \iconf{A_1}\times\,\iconf{C_2}$ to $(x_2, z_1) \in  \iconf{A_2}\times\,\iconf{C_1}$  in the 
  least solution to the equations
$$
(x_2, y_1) = S(x_1, y_2)
\ \hbox{ and }\ 
(y_2, z_1) = T(y_1, z_2)\,
$$
---given, as in Kahn 
networks, by taking a least fixed point.

\begin{theorem} 
 The bicategory of deterministic strategies on GoI games with rigid 2-cells  is equivalent to 
the free compact-closed category $ {\rm Int}({\bf dI})$ of~\cite{joyal-street-verity} and
 the  Geometry of Interaction category $\G({\bf dI})$ of~\cite{AbramskyHaghverdiScott} 
 built on the category $\bf dI$ of dI-domains and stable functions. 
  \end{theorem}

Geometry of Interaction started as an investigation of the nature of proofs of linear logic, understood as networks~\cite{GirardGoI1}.
It has subsequently been tied to 
optimal reduction in the $\lambda$-calculus~\cite{GoI-optred}, and inspired implementations via token machines on networks~\cite{mackie,ghica}; when the two components of a GoI game match 
events of exit and entry of a token at a link.  
 
It is straightforward to extend GoI games with winning conditions.
A winning condition on a GoI game $A=A_1\vvbar A_2
$ picks out a subset of the configurations $\iconf A$, so amounts to specifying a property $W_A(x_1,x_2)$ of pairs $(x_1,x_2)$  in $\iconf{A_1}\times  \iconf{A_2}$.  That a deterministic strategy from GoI game $A$ to GoI game $B= B_1\vvbar B_2$ is winning means 
$$
W_A(x, g(x,y)) \implies W_B(f(x,y), y)\,,
$$
for all $x\in \iconf{A_1}, y\in \iconf{B_2}$,   when expressed in terms of the pair of stable functions the strategy determines. In particular,
a deterministic winning strategy in the individual GoI game $B$, with winning conditions $W_B$, corresponds to a stable function $f:
\iconf{B_2} \to \iconf{B_1}$ such that $\all y\in\iconf{B_2}.\ W_B(f(y), y)$.  

With stable spans, unlike with dI-domains with stable functions,  the operation of parallel composition $\vvbar$ is no longer a product; 
stable spans are monoidal-closed and 
not a cartesian-closed.  
While 
general, not just deterministic, strategies $\sig:A\profto B$  between  GoI games are expressible 
as stable spans
$
A_1\vvbar B_2^\perp \profto A_2^\perp \vvbar B_1
$, 
their expression doesn't 
project to 
an equivalent pair of
separate components as with lenses. 

\subsubsection{The GoI adjunctions}\label{sec:GoIadjns}
%
For any game $A$ there is a map of \esswp
$$
f_A: A \to A^+ \vvbar A^-\,,
$$
where $A^+$ is the projection of $A$ to its $+$ve events and $A^-$ is the projection to its $-$ve events:  the map $f_A$ acts as the identity function on events;  
it sends a  configurations $x\in\iconf A$ to $f_A x = x^+\vvbar x^-$.   It determines an adjunction
$
f_! \dashv f^*$ from $A $ to $A^+ \vvbar A^-$.
Because the game $A$ is race-free,  both ${f_A}_!$ and ${f_A}^*$ are deterministic strategies.
This provides a lax functor from  deterministic strategies in general, to those between GoI games.
Let $\sig:A\profto B$ be a deterministic 
strategy between 
games $A$ and $B$.  Defining
$\goi(\sig) = {f_B}_! \scirc \,\sig\, \scirc {f_A}^*$
we obtain a deterministic 
strategy $$\goi(\sig): A^+\vvbar A^- \profto B^+\vvbar B^-\,.$$  Then, the strategy $\goi(\sig)$ corresponds  to a stable function from 
$A^+\vvbar B^-$ to  $A^- \vvbar B^+$, so to a GoI map. 
The operation $\goi$ only forms a lax functor however: 
for $\sig:A\profto B$ and $\tau:B\profto C$, there is, in general, a nontrivial 2-cell $ \goi(\tau\scirc \sig) \Rightarrow  \goi(\tau)\scirc \goi(\sig)$. This puts pay to $\goi$ being right adjoint to the inclusion functor in a pseudo adjunction from the category of GoI games to deterministic strategies. 
But, there is a {\em lax} pseudo adjunction, of potential use in 
abstract interpretation.
 
   
 \subsection{Dialectica games}
 
 Dialectica categories were devised in the late 1980's by Valeria de Paiva in her Cambridge PhD work with Martin Hyland~\cite{valeria}.  The motivation then was 
 to provide a model of linear logic underlying Kurt G\"odel's {\em dialectica interpretation} of first-order logic~\cite{dialecticaInt}.  They have come to prominence again recently because of a renewed interest in their maps in a variety of contexts, in formalisations of reverse differentiation and back propagation, open games and learners,  and as an early occurrence of maps as lenses.   The dialectica interpretation underpins most proof-mining techniques~\cite{kreisel,kohlenbach}. 
 
 We obtain a particular dialectica category, based on Berry's stable functions, as a full subcategory of deterministic strategies on dialectica games.  Dialectica games are obtained as GoI games of imperfect information, intuitively by not allowing Player to see the moves of Opponent. 
 
A  {\em dialectica game} is   
a GoI game $A= A_1\vvbar A_2$ with winning conditions, and with imperfect information given as follows.
The imperfect information is determined by particularly simple order of access levels: $\exi \prec \foral$.  All Player moves, those in $A_1$, are assigned to $\exi$ and all Opponent moves, those in $A_2$, are assigned to $\foral$. 
It is helpful to think of the access levels $\exi$ and $\foral$ as representing two rooms separated by a one-way mirror allowing anyone in room  $\foral$ to see through to room $\exi$.  In a dialectica game, Player is in room $\exi$  and Opponent in room $\foral$.  Whereas Opponent can see the moves of Player, and in a counterstrategy make their moves dependent  on those of Player, the moves of Player are made blindly, in that they cannot depend on Opponent's moves.   

Although we are mainly interested in strategies {\em between} dialectica games it is worth pausing to think about strategies {\em in} a single dialectica game  $A= A_1\vvbar A_2$ with winning conditions $W_A$.   Because Player moves cannot causally depend on Opponent moves, a
 deterministic strategy in $A$ corresponds to a configuration $x\in\iconf{A_\exi}$; that it is winning means $\forall y\in \iconf{A_\foral}.\ W_A(x,y)$.  So to have a winning strategy for the dialectica game means
$$
\exists x\in \iconf{A_\exi}\forall y\in \iconf{A_\foral}.\ W_A(x,y)\,.
$$

Consider now a deterministic winning strategy $\sig$ from a dialectica game $A=A_1\vvbar A_2$ with winning conditions $W_A$ to another 
$B=B_1\vvbar B_2$ with winning conditions $W_B$.  Ignoring  access levels,   $\sig$ is also a deterministic strategy between GoI games, so corresponds to a pair of stable functions
$$f:  \iconf{A_1}\times\iconf{B_2} \to \iconf{B_1}\hbox{  and } g:  \iconf{A_1}\times\iconf{B_2} \to \iconf{A_2}\,.$$
But moves in $B_2$ have access level $\foral$, moves of $B_1$ access level $\exi$;  
a causal dependency in the strategy $\sig$ of a move in $B_1$ on a move in $B_2$ would violate the access order  $\exi \prec \foral$.  That no move in $B_1$ can causally depend on a move in $B_2$ is reflected in the functional independence of $f$ on its second argument.  As a deterministic strategy between dialectica categories, $\sig$ corresponds to a pair of stable functions
$$f:  \iconf{A_1} \to \iconf{B_1} \hbox{  and } g:  \iconf{A_1}\times \iconf{B_2} \to  \iconf{A_2}\,,$$
which we can picture as:
$$
\xymatrix@R=10pt
{
A_1\ar@{-}@/^2.2pc/[d]^g\ar[r]^f&B_1\\
 A_2&B_2\ar@/_0.3pc/[l]  
 }
 $$
 That $\sig$ is winning means, for all $x\in \iconf{A_1}, y\in \iconf{B_2}$, 
$$
W_A(x, g(x,y)) \implies W_B(f(x), y)\,.
$$
 Pairs of functions $f,g$ satisfying this winning condition are precisely
 the maps of de Paiva's construction of a dialectica category from Berry's stable functions.  

Such pairs of functions are the {\em lenses} of functional programming where they were invented to make  composable local changes on data-structures~\cite{oles,pierce}.    
We recover their at-first puzzling composition from the composition of strategies.  Let 
$\sig$ be a deterministic strategy from dialectica game $A$ to dialectica game $B$; and $\tau$ a deterministic strategy from $B$ to another dialectica game $C$.  Assume $\sig$ corresponds to a pair of stable functions $f$ and $g$, as above, and analogously that   $\tau$ corresponds to stable functions $f'$ and $g'$.  Then,  the composition of strategies $\tau\scirc\sig$ corresponds to the composition of lenses:
 with first component  $f'\circ f$
and second component taking $x\in\iconf{A_1}$ and $y\in \iconf{C_2}$ to 
$g(x, g'(f(x), y))$.

\begin{theorem} 
 The bicategory of deterministic strategies on dialectica games with rigid 2-cells  is equivalent to 
the dialectica category of~\cite{valeria} 
 built on dI-domains and stable functions. 
  \end{theorem}
 
\subsubsection*{Girard's variant} 
In the first half of de Paiva's thesis she concentrates on the construction of dialectica categories.  In the second half, she follows up on a suggestion of Girard to explore a variant.  This too is easily understood in the context of concurrent games: imitate the work of this section, with GoI games extended with imperfect information, but now with access levels modified to the {\em discrete} order on $\exi, \foral$.  Then the causal dependencies of strategies are further reduced and deterministic strategies from $A= A_1\vvbar A_2$  to $B= B_1\vvbar B_2$ correspond to pairs of stable functions
$$
f: \iconf{A_1} \to \iconf{B_1} \hbox{  and  }
g: \iconf{B_2} \to \iconf{A_2}\,.
$$

\subsubsection*{Combs}
Discussions of causality in science, and quantum information in particular, are often concerned with what causal dependencies are feasible;  
then structures  similar  to orders of access levels are used to capture {\em one-way signalling}, as in dialectica games, and {\em non-signalling}, as in Girard's variant. In this vein, through another variation of games with imperfect information, we obtain
the generalisation of lenses to {\em combs}, used in quantum architecture and  information~\cite{Chiribella2008,aleks}.  Combs provide a common method for imposing higher-order structure on quantum circuits or string diagrams.  

Combs arise as strategies between {\em comb games} which, at least formally, are an obvious generalisation of dialectica games; their name comes from their graphical representation
as structures that look like (hair) combs, with  teeth representing successive transformations from input to output. An {\em $n$-comb game}, for a natural number $n$, is an $n$-fold parallel composition $A_1\vvbar A_2 \vvbar \cdots\vvbar A_n$ of purely Player or purely Opponent games $A_i$ of alternating polarity; 
it is a game of imperfect information associated with access levels $1\prec 2\prec \cdots \prec n$ with moves of component $A_i$ having access level $i$.  Dialectica games are 2-comb games with winning conditions.  

\subsubsection*{Open games and learners}
{\em Open games and learners}~\cite{opengames,openlearners} have recently been presented as
{\em parameterised} 
lenses or optics, in the case of open games with some concept of  equilibrium or winning condition~\cite{capucci2021foundations}. As an example, we obtain a form of open game between dialectica games $A$ and $B$ as a strategy  
$
A\vvbar P \profto B
$, 
where  $P$ is a dialectica game 
of which the configurations 
 specify {\em strategy profiles}.  A variation based on optimal strategies between dialectica games with payoff, following~\cite{payoff}, 
 introduces Nash equilibria and takes us 
 into 
 game-theory territory, 
 and to a testing ground for open games and the notions being developed there. 
 
\subsection{Optics}
 Now 
 we show that
general, possibly nondeterministic, strategies between dialectica games are precisely 
{\em optics}~\cite{optics,optics2} based on stable spans~\cite{SEW,mikkel,Visions}.   Recall 
that a dialectica game comprises $A_1\vvbar A_2$ where $A_1$ is a purely Player game, all events of which have access level $\exi$ and $A_2$ is a purely Opponent game with all events of access level $\foral$, w.r.t.~access order $\Lambda$ specifying $1\prec 2$.  
We ignore winning conditions. 

Let $A$ and $B$ be dialectica games. Let $Q$ be a  purely Player $\Lambda$-game.  
Recall that nondeterministic strategies between purely Player games correspond to stable spans.  
Consider 
strategies
$$F: A_1 \profto B_1\vvbar Q \hbox{ and  }
G: Q \vvbar B_2^\perp \profto A_2^\perp\,.$$  
Then the strategies $F$ and $G$ are between purely Player games, so  correspond to 
stable spans---
Appendix~\ref{app:stspans}. As any causal dependencies of $F$ or $G$ 
respect $\Lambda$, they 
are $\Lambda$-strategies.  

Hence the composition
$$\xymatrix{
A_1\vvbar B_2^\perp \ar[r]|-{\! +\!}^{F\vvbar B_2^\perp}&
B_1\vvbar Q\vvbar B_2^\perp  \ar[r]|-{\! +\!}^{B_1\vvbar G}& 
B_1\vvbar A_2^\perp
}
$$
is also a $\Lambda$-strategy and, being between purely Player games, corresponds to a stable span.   The composition, rearranges to a strategy 
$$
\sig: A_1\vvbar A_2 \profto B_1\vvbar B_2\,,
$$
which is a $\Lambda$-strategy, 
so to a strategy between the original dialectica games $A$ and $B$.  
We call this strategy $\Optic(F,G)$ and call $(F,G)$ its {\em presentation} from $A$ to $B$ with {\em residual} $Q$. 
The terminology is apt, as we'll show
strategies obtained in this way 
 coincide with  {\em optics} as usually defined. 
 Presentations can be represented diagrammatically:  
$$
\small\xymatrix@R=5pt@C=10pt{
A_1 \ar[rrr]^F&\ar@/^0.5pc/[dr]&&B_1\\
&&Q\ar@{-}@/^0.5pc/[dl]\\
A_2&&&B_2\,,\ar[lll]^G
}
$$
 illustrating how  $F$ and $G$ are ``coupled'' via the 
 residual 
 $Q$.   
 
 As usually defined, an optic is an equivalence class of presentations.  Let $(F,G)$ and $(F', G')$ be  presentations  from $A$ to $B$ with   residuals  $Q$ and $Q'$ respectively.  The equivalence relation $\sim$ on presentations is that generated by taking $(F,G)\sim(F', G')$ if, for some $f: Q\profto Q'$, the following triangles commute
 $$
 \xymatrix
 {
 A_1\ar[r]|+^{F'}\ar[dr]|\times_{F}& B_1\vvbar Q' \!\!\!& Q'\vvbar B_2\ar[r]|+^{G'} & B_1\\
 &B_1\vvbar Q \ar[u]|+_{B_1\vvbar f} & Q\vvbar B_2\,.\ar[u]|+^{f\vvbar B_2}\ar[ru]|\times_{G}  \\
 } \eqno(\sim\hbox{def})
 $$

 Presentations of optics compose. 
 Let $A$, $B$ and $C$ be dialectica games.  Given a presentation $(F,G)$ from $A$ to $B$ with residual $Q$ and another  $(F',G')$ from $B$ to $C$ with residual $P$ we obtain a presentation   
 from $A$ to $C$
with residual $P\vvbar Q$ guided by the diagram  
$$
\small\xymatrix@R=5pt@C=10pt{
A_1 \ar[rrr]^F&\ar@/^0.5pc/[dr]&&B_1
\ar[rrr]^{F'}&\ar@/^0.5pc/[dr]
&&C_1
\\
&&Q\ar@{-}@/^0.5pc/[dl]&&&P\ar@{-}@/^0.5pc/[dl]
\\
A_2&&&B_2 \ar[lll]^G
&&&C_2\,, \ar[lll]^{G'}
}
$$
precisely, 
as $((F'\vvbar Q)\scirc F, \, G\scirc (Q\vvbar G')\scirc (s_{PQ}\vvbar C_2))$, where $s_{PQ}$ expresses the symmetry $P\vvbar Q\iso Q\vvbar P$
. 

Composition preserves $\sim$ and has the evident identity presentation, with residual the empty game.  
It follows 
 that $\Optic$ is functorial and that 
if $(F,G)\sim(F', G')$ then $$\Optic(F,G)\iso\Optic(F', G')\,.$$   

To show 
any strategy between container games is an optic, we exploit  the monoidal-closure of stable spans
---see Appendix~\ref{app:stspans}.  
A presentation
$(F,G)$ 
 is $\sim$-equivalent to a {\em canonical presentation} $(F', G')$ with residual  $Q' = [B_2^\perp \multimap A_2^\perp]$ and  $G'$ as {\em application} $\app$: in ($\sim$def), take $f = \curry G
$ 
and $F' = (B_1\vvbar f) \scirc F$
.  

Now, strategies $\sig:  A\profto  B$,   between dialectica games $A$ and $B$, correspond to canonical presentations. To see this,
ignoring the access levels for the moment, a general strategy 
$$\sig: A_1\vvbar A_2 \profto B_1\vvbar B_2$$
corresponds to a strategy between purely Player games  
$$\sig_1: A_1\vvbar B_2^\perp \profto B_1\vvbar A_2^\perp\,,$$
so to a stable span. 
From the monoidal-closure of stable spans we can curry 
$\sig_1$, to obtain a corresponding strategy  
$$\sig_2: A_1\profto  [B_2^\perp \multimap (B_1\vvbar A_2^\perp)]\,$$
with the property
$$
\sig_1 \iso \app_{B_1\vvbar A_2^\perp}\scirc (\sig_2\vvbar B_2^\perp)\,.
$$
Recalling the access levels, 
no event of $B_1$ can causally depend on an event of $B_2$, 
ensuring that $\sig_2$ corresponds to 
$$\sig^+: A_1\profto  B_1\vvbar [B_2^\perp \multimap  A_2^\perp]\,$$
where
$$
 \sig_1 \iso \app_{B_1\vvbar A_2^\perp}\scirc (\sig_2\vvbar B_2^\perp)\iso  (B_1\vvbar \app_{A_2^\perp}) \scirc  (\sig^+\vvbar B_2^\perp)\,.
$$
It follows that $(\sig^+, \app_{A_2^\perp})$ is a canonical presentation for which 
$$\sig\iso \Optic(\sig^+, \app_{A_2^\perp})\,,$$
giving a correspondence between strategies $\sig:A\profto B$ between dialectica games 
  and canonical presentations $(\sig^+, \app_{A_2^\perp})$. 
   
   Via canonical presentations we obtain a bicategory 
   of optics.  Its objects are dialectica games.  Its maps are stable spans 
    $
    A_1 
    \profto B_1\vvbar [B_2^\perp \multimap  A_2^\perp]
    \,,$
    with the associated 
    2-cells,  from dialectica game $A$ to dialectica game $B$. 
   
\begin{theorem} 
 The bicategories of strategies on dialectica games with rigid 2-cells 
and that 
of 
optics built on stable spans are equivalent. 
  \end{theorem}

\subsection{Containers}\label{sec:containers}

A {\em container game} is a  
game of imperfect information $A$ w.r.t.~access levels $\exi \prec \foral$; each Player move of $A$ is sent to $\exi$ and each Opponent move to $\foral$.  
So in $A$ the only 
 causal dependencies between moves of different polarity 
 are $\plmove \leq \opmove$. 

The configurations of a container game $A$ have a dependent-type structure.  Opponent moves can causally depend on Player moves, but not conversely.  Let $A_1$ denote the subgame comprising the initial substructure of purely Player moves of $A$.  
A configuration $x\in\iconf{A_1}$  determines a subgame $A_2(x)$ comprising the substructure of $A$ based on all those Opponent moves for which all the Player moves on which they depend appear in $x$.
A configuration of  $A$ breaks down uniquely into a union $x\cup y$, so a pair $(x,y)$, where $x\in \iconf{A_1}$ and $y\in \iconf{A_2(x)}$.  We can see the configurations of a container game $A$ as  forming a dependent sum $\Sigma_{x:{A_1}} \,  {{A_2(x)}}$. In this way a container game represents  a container type, familiar from functional programming~\cite{containers};
configurations $x$ of $A_1$ are its ``shapes,''  indexing ``positions'' $y\in A_2(x)$.\footnote{
  We won't treat symmetry in concurrent games at all here, but 
  it is important in many applications. 
With the addition of symmetry
, configurations form a nontrivial category, not merely a partial order based on inclusion~\cite{lics14}.  
  } 
 
    We can of course extend a container game $A$  with winning conditions which we identify with a property $W_A$ of the dependent sum $\Sigma_{x:{A_1}} \,  {{A_2(x)}}$. 
A deterministic winning strategy in the container game corresponds to a configuration $x\in \iconf{A_1}$ such that $\forall y \in \iconf{{A_2(x)}}.\ W_A(x,y)$.

Strategies between container games respect 
$\preceq$ on access levels.  
A deterministic strategy $\sig$ from a container game $A$ to a container game $B$ corresponds to a map of container types, also called a {\em dependent lens}, 
having  type
$$
\Sigma_{f:[{A_1}\to {B_1}]} \Pi_{x:{A_1}}\  
        [ {B_2(f (x)) } \to {A_2(x)}]\,; \eqno(*)
$$
so $\sig$ corresponds to a pair of stable functions 
$$f:[{A_1} \to {B_1}] \ \hbox{  and  } \ g:\Pi_{x:{A_1}}  \, [{B_2(f(x))} \to {A_2(x)}]\,,$$  
where we  are using 
the function space, dependent sum and product of stable functions
---see Appendix~\ref{app:deptypes}.  
With  winning conditions $W_A$ and $W_B$
, the strategy from 
$A$ to $B$  would be winning iff, 
for all $x\in\iconf{A_1}$, $y\in \iconf{B_2(f(x))}$,
$$
W_A(x, g_x(y)) \implies W_B(f(x), y)\,.
$$
The correspondence respects composition.  Container types built on dI-domains and stable functions arise as a full subcategory of deterministic concurrent games.

\begin{theorem} 
 The bicategory of deterministic strategies on container games with rigid 2-cells is equivalent to a full subcategory of
containers of dI-domains and stable functions~\cite{containers}.
  \end{theorem}

\subsection{Dependent optics}
What about general, nondeterministic, strategies between container games?  
A way to motivate their characterisation  is to observe the isomorphism of the type of a dependent lens $(*)$ above with 
$$
\Pi_{x:{A_1}}\Sigma_{y:{B_1}} \  
[{B_2(y)} \to A_2(x) ]\,. 
$$
It is this nonstandard way to present the type of lenses that generalises to the monoidal-closed bicategory of stable spans, once we move to the dependent product $\sPi$ of  stable spans
---
Appendix~\ref{app:deptypes}. 
 
Ignoring winning conditions, a general strategy 
between container games  
corresponds to a new form of 
optic.   A {\em dependent optic} between  container games, from 
 $A$ to 
 $B$, is a 
 stable span of type 
$$
\DOP[A, B]  = \sPi_{x:{A_1}}\Sigma_{y:{B_1}} \   [{B_2(y)} \multimap {A_2(x)}]\,,
$$
so a rigid map into $\DOP[A,B]$.  
A 2-cell $f:F\Rightarrow F'$ between dependent optics $F, F':\DOP[A, B]$  is a 2-cell of stable spans.  
Composition of dependent optics is the stable span
$$
\comp: \DOP[B, C]\, \vvbar \, \DOP[A, B] \profto \DOP[A,C]
$$
described by
$$\eqalign{
G \comp F
\eqdef 
\lambda x:{A_1}. \ 
\plet 
&
(y, F') \larrow F(x) \pin \cr
\plet 
&
(z, G') \larrow G(y)  \pin (z, F'\scirc G')\,,
}
$$
where $F'\stcirc G': [{C_2(z)}  \multimap {A_2(x)}]$  is the composition of stable spans  $G':[{C_2(z)} \multimap {B_2(y)}]$ and $F':[{B_2(y)} \multimap {A_2(x)}]$.   The 
identity 
optic of container game $A$ acts on
$x:A_1$ to return the identity 
at the $x$-component of 
$\Sigma_{x:A_1}[A_2(x) \multimap A_2(x)]$.  

The equivalence of strategies between container games 
 with dependent optics, hinges on recasting $\DOP[A,B]$ as a strategy $\posn: A\profto B$ 
between container games  $A$ and $B$.
Any strategy 
between container games is of course a strategy  
where we forget the access levels. 
We can express that 
a strategy $\sig: A \profto B$ respects the access levels, so is truly a strategy between container games, precisely through the presence of a rigid 2-cell 
$$
\xymatrix@R=14pt@C=18pt{
A\ar@/^1pc/[rr]|-{\! +\!}^\sig
\ar@/_1pc/[rr]|-{\! +\!}_{\posn}&\Downarrow
r\!\!\!\!\!&B\,. \\
}$$  
The 2-cell $r$
 is unique, making the strategy $\posn$  terminal amongst strategies $\sig$ between  container games, from $A$ to $B$. By restricting $r$ to Player moves we obtain the dependent optic $\sig^+: \DOP[A,B]$ which corresponds to $\sig$. 
  
 \begin{theorem} \label{thm:depop}
 The bicategory of strategies between container games, with rigid 2-cells,  is equivalent to the bicategory of dependent optics. 
 \end{theorem}  
\begin{proof}(Sketch)
 
 The proof relies on position functions of strategies~\cite{ecsym-notes}---see \S4.5.1 of {\it op. cit.}.  
 The {\em position function} $d$ of a strategy $\sig:S\to A^\perp\vvbar B$  is given by $d(x) = \sig [x]_S$ for $x\in \conf{S^+}$; the event structure $S^+$ is the projection of $S$ to its Player moves. A position function $d$ is characterised as a union-preserving function $d:\conf{S^+}\to \conf A$ which  restricts to a map $f:S^+\to A^+$ of event structures that on $s\in S^+$ gives the unique event of Player amongst the $\leq_A$-maximal events in $d([s]_S)$. 
 
We describe  the strategy $\posn
$ associated with $\DOP[A,B]$ via its position function. 

Recall that $\DOP[A,B]$ is built using $\Pr$ out of primes of the stable family 
\[ \sPi_{x\in\iconf{A_1}}\Sigma_{y\in\iconf{B_1}} \   [\iconf{B_2(y)} \multimap \iconf{A_2(x)}]\,,\]
whose events take the form
$(x, b)$ or  $(x, (y,a))$ where $x\in\conf{A_1}$, $y\in\conf{B_2}$, $b\in B_1$ and $a\in A_2$.  The position function $d_o$ takes a prime with top element $(x, b)$ to $x\vvbar [b]_B$ and one with top element $(x, (y,a))$ to $x\vvbar (y\cup [a]_A)$.  
 
Showing $\posn$ is terminal amongst strategies $\sig:S\to A^\perp\vvbar B$ between container games, 
rests on there being a unique rigid map $\sig^+$ such that
$$
\xymatrix@R=14pt@C=18pt{
\ar[dr]_dS^+\ar[r]^{\sig^+}& \DOP[A,B]\ar[d]^{d_o}\\
&A^\perp\vvbar B}
$$
commutes, where 
 $d:S^+\to A^\perp\vvbar B$ is the position function of $\sig$. The map 
   $\sig^+ = \Pr(\sig_0)$, where $\sig_0$ is the map of stable families $$\sig_0:\iconf{S^+} \to \Pi_{x\in\iconf{A_1}}\Sigma_{y\in\iconf{B_1}} \   [\iconf{B_2(y)} \multimap \iconf{A_2(x)}]\,,$$
defined as follows.  
For notational simplicity, assume $A$ and $B$ have disjoint sets of events so we can regard the events of $A\vvbar B$ as $A\cup B$. 
Let $s\in S^+$.  Either $\sig(s)\in B_1$ or $\sig(s)\in A_2$.  Accordingly,  define  
$$
\sig_0(s) = 
\begin{cases}
(x,b) &\hbox{  if } \sig(s) = b\in B_1\ \&\ x= \sig[s]_S^-\,; \\
 (x,(y,a)) & \hbox{ if  } \sig(s) = a\in A_2\  \&\ x= \sig[s]_S^-\cap A_1\\
 & \&\ y= \sig[s]_S^-\cap B_2\,.
\end{cases}
$$ 
The fact that 
 $\dop[B,C] \scirc  \dop[A,B] \iso   \dop[A,C]$
 is key in showing the correspondence of $\sig$ with $\sig^+$ respects composition.
 \end{proof}

The results on optics for container games specialise  to those for dialectica games.

\section{Enrichment 
}\label{sec:enrichment}
Games and strategies support enrichments, to:  
probabilistic  strategies, also with continuous distributions~\cite{probstrats,mfps2018}; quantum strategies~\cite{popl19}; 
and strategies on the reals~\cite{aurore}. 
The enrichments specialise to the cases above. 
 Work on enriched concurrent strategies transfers to situations of interest in functional programming, domain theory and geometry of interaction. 
  In explaining how, we can take advantage of a general method for enriching strategies. 

The 
enrichments named above were developed individually and are not always the final story.  For instance,
the assignment of quantum operators to configurations of strategies in~\cite{popl19} is not functorial w.r.t.~inclusion on configurations, a defect when it comes to understanding how the 
operator of a configuration is built up. The authors' 
remedy 
also achieves all the enrichments just named, now uniformly by the same construction.

The construction is 
w.r.t.~a symmetric monoidal category $(\mathcal M, \otimes, {\rm I})$.  For example, $\mathcal M$ can be
the monoid  $([0,1], \cdot, 1)$ comprising the unit interval under multiplication (for probabilistic strategies); 
measurable spaces with Markov kernels (for probabilistic strategies with continuous distributions); 
CPM, finite-dimensional Hilbert spaces with completely positive maps (for quantum strategies); 
or Euclidean spaces with smooth maps, to support (reverse) differentiation. 

We first extend $\mathcal M$ to allow interaction beyond that from argument to result.  The  {\em parameterised category}   $\Para{\mathcal M}$  has the same objects, 
now with maps
$(P, f, Q):X\to Y$ consisting of $f:X\tensor P \to Q\tensor Y$ in $\mathcal M$; the parameters $P$  and $Q$ allow input and output with 
the environment. 
Composition accumulates parameters: $(R,g,S)\circ (P, f, Q) \eqdef (P\tensor R,\  (Q\tensor g)\circ (f\tensor R), \  Q\tensor S) $.
Then, 
 \begin{itemize}
\item[{\bf (1) }]
moves $a$ of a game $A$ are assigned objects $\Hilb(a)$ in $\mathcal M$, extended to $X\in\Con_A$ by $\Hilb(X) \eqdef \Tensor_{a\in X} \Hilb(a)$. 
(Neutral moves, appearing in interaction,  are assigned the tensor unit $\rm I$.)
%
\item[{\bf (2) }] an $\mathcal M$-enriched strategy $\sig:S\to A$ is accompanied by  a functor $\Q:(\conf S, \subseteq) \to \Para{\mathcal M}$.  To an interval $x\subseteq x'$ in
$\conf S$ this assigns a parameterised map 
$\Q(x\subseteq x')$  from $\Q(x)$ to $\Q(x')$ with input parameters  $\Hilb(\sig{(x'\setdif x)}^-)$ and output parameters  $\Hilb(\sig{(x'\setdif x)}^+)$. 
 \end{itemize}
The assignment in {\bf (2)} describes how the internal state is transformed in moving from $x$ to $x'$ under interaction with the environment 
through events $x'\setdif x$.   
 The assignment in {\bf (2)} is assumed
{\em oblivious}, \ie~$\Q(x\subseteq^- x')$ is always an isomorphism in $\mathcal M$, 
expressing that all the input from $x'\setdif x$ is adjoined to the internal state $\Q(x)$ to produce a new internal state $\Q(x')=\Q(x)\tensor \Hilb(\sig{(x'\setdif x)})$.  This is needed to ensure that the enriched version of copycat  acts as identity w.r.t.~composition.  

 In the  quantum and probabilistic cases, observation is contextual, reflected in the presence of an extra {\em drop condition}, a form of inclusion-exclusion principle~\cite{probstrats,popl19}; it requires 
 $\mathcal M$  be enriched over,  at least, cancellative commutative monoids. 

Moves, their positions, 
dependencies and polarities, 
orchestrate the functional dependency and dynamic 
linkage in composing enriched strategies.   Consider an enriched strategy  $\sig:S\to A^\perp\vvbar B$. In the enrichment  the interval  $x\subseteq x'$  of  $S$ is assigned a parameterised map $\Q(x\subseteq^- x')$ 
pictured below, in which the input parameters $P_A\otimes P_B$  and output parameters $Q_A\otimes Q_B$ have been factored into those over $A$ and those over $B$:
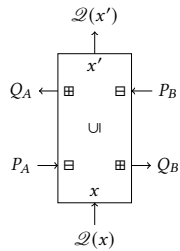
\begin{center}
\scalebox{0.7}{\begin{tikzpicture}[scale=0.7]
\node (qa) at (0,1) {$Q_A$};
\node (pa) at (0,-1) {$P_A$};
\node (qb) at (4,-1) {$Q_B$};
\node (pb) at (4,1) {$P_B$};
\node (qxp) at (2,3) {$\Q(x')$};
\node (qx) at (2,-3) {$\Q(x)$};
\node (plus1) at (1.3,1) {$\boxplus$};
\node (minus1) at (1.3,-1) {$\boxminus$};
\node (plus2) at (2.7,-1) {$\boxplus$};
\node (minus2) at (2.7,1) {$\boxminus$};
\node (xp) at (2,1.7) {$x'$};
\node (x) at (2,-1.7) {$x$};
\node (subset) at (2,0) {\rotatebox[origin=c]{90}{$\subseteq$}};
\draw (1,2) rectangle (3,-2);
\draw[->] (1,1) -- (qa);
\draw[->] (3,-1) -- (qb);
\draw[->] (pa) -- (1,-1);
\draw[->] (pb) -- (3,1);
\draw[->] (qx) -- (2,-2);
\draw[->] (2,2) -- (qxp);
\end{tikzpicture}
}
\end{center}
Consider now the interaction of enriched strategies $\sig:S\to A^\perp\vvbar B$  and $\tau:T\to B^\perp\vvbar C$. An interval $y\sncirc x \subseteq y'\sncirc x'$ in the interaction $T\sncirc S$ breaks down into two intervals $x \subseteq x'$ of $S$ and  $y \subseteq y'$ of $T$.  Their assignments compose together as shown to give the assignment to 
$y\sncirc x \subseteq y'\sncirc x'$:
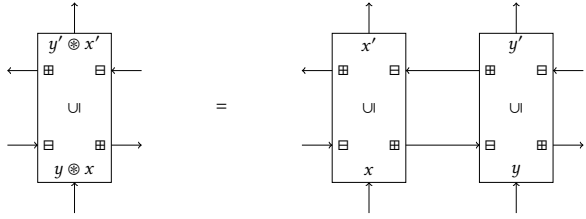
\begin{center}
\scalebox{0.7}{
\begin{tikzpicture}[scale=0.7]
\node (qa) at (-2,1) {$ $};
\node (pa) at (-2,-1) {$ $};
\node (qb) at (2,-1) {$ $};
\node (pb) at (2,1) {$ $};
\node (qxp) at (0,3) {$ $};
\node (qx) at (0,-3) {$ $};
\node (plus1) at (-0.7,1) {$\boxplus$};
\node (minus1) at (-0.7,-1) {$\boxminus$};
\node (plus2) at (0.7,-1) {$\boxplus$};
\node (minus2) at (0.7,1) {$\boxminus$};
\node (xp) at (0,1.7) {$y' \circledast x'$};
\node (x) at (0,-1.7) {$y \circledast x$};
\node (subset) at (0,0) {\rotatebox[origin=c]{90}{$\subseteq$}};
\draw (-1,2) rectangle (1,-2);
\draw[->] (-1,1) -- (qa);
\draw[->] (1,-1) -- (qb);
\draw[->] (pa) -- (-1,-1);
\draw[->] (pb) -- (1,1);
\draw[->] (qx) -- (0,-2);
\draw[->] (0,2) -- (qxp);
\node(equal) at (4,0) {$=$};
\node (2qa) at (6,1) {};
\node (2pa) at (6,-1) {};
\node (2qb) at (10,-1) {};
\node (2pb) at (10,1) {};
\node (2qxp) at (8,3) {};
\node (2qx) at (8,-3) {};
\node (2plus1) at (7.3,1) {$\boxplus$};
\node (2minus1) at (7.3,-1) {$\boxminus$};
\node (2plus2) at (8.7,-1) {$\boxplus$};
\node (2minus2) at (8.7,1) {$\boxminus$};
\node (2xp) at (8,1.7) {$x'$};
\node (2x) at (8,-1.7) {$x$};
\node (2subset) at (8,0) {\rotatebox[origin=c]{90}{$\subseteq$}};
\draw (7,2) rectangle (9,-2);
\draw[->] (7,1) -- (2qa);
\draw[->] (2pa) -- (7,-1);
\draw[->] (2qx) -- (8,-2);
\draw[->] (8,2) -- (2qxp);
\node (2nqa) at (10,1) {};
\node (2npa) at (10,-1) {};
\node (2nqb) at (14,-1) {};
\node (2npb) at (14,1) {};
\node (2nqxp) at (12,3) {};
\node (2nqx) at (12,-3) {};
\node (2nplus1) at (11.3,1) {$\boxplus$};
\node (2nminus1) at (11.3,-1) {$\boxminus$};
\node (2nplus2) at (12.7,-1) {$\boxplus$};
\node (2nminus2) at (12.7,1) {$\boxminus$};
\node (2nxp) at (12,1.7) {$y'$};
\node (2nx) at (12,-1.7) {$y$};
\node (2nsubset) at (12,0) {\rotatebox[origin=c]{90}{$\subseteq$}};
\draw (11,2) rectangle (13,-2);
\draw[->] (13,-1) -- (2nqb);
\draw[->] (2npb) -- (13,1);
\draw[->] (2nqx) -- (12,-2);
\draw[->] (12,2) -- (2nqxp);
\draw[->] (9,-1) -- (11,-1);
\draw[->] (11,1) -- (9,1);
\end{tikzpicture}}
\end{center}
For this composition to be well-defined we need that it involves no functional loops.  But this is assured through the absence of causal loops in the interaction. 
Suppose $z\subseteq z'$ is an interval of ${T\scirc S}$.  The event structure $T\scirc S$ is the projection of $T\sncirc S$.  Take $y\sncirc x = [z]_{T\sncirc S}$ and $y'\sncirc x' = [z']_{T\sncirc S}$ the down-closures of $z$ and $z'$ in $T\sncirc S$.  By definition, the interval 
$z\subseteq z'$ of ${T\scirc S}$ is assigned the same parameterised map as $y\sncirc x  \subseteq y'\sncirc x'$ in ${T\sncirc S}$. 

Enrichments achieved in this way specialise automatically to sub(bi)categories, and  the functional cases  we have considered, without needing extra demands 
on the category  $\mathcal M$. 
Some of the specialisations are known.  For example, stable spans when enriched by probability, via the monoid  $([0,1], \cdot, 1)$, 
become Markov kernels, and this enrichment extends to the various forms of optics we have uncovered.  
Others deserve further exploration.  Enrichment w.r.t.~CPM, yielding quantum strategies, specialises to nondeterministic strategies between GoI games.   
This provides an enrichment of Geometry of Interaction with quantum effects, and a likely candidate with which to give a semantics for the more operational, multi-token machine  treatment of~\cite{quantumGoI}.
 Concurrent games and strategies enrich with  Euclidean spaces and (partial) smooth maps and via them connect with  forwards and reverse differentiation. An easier subcase to explore first is the enrichment of stable functions; here  one already encounters many of the issues of differential programming.

 Enriched strategies provide a general framework in which to explore the interaction patterns of ``functions'' (maps in $\mathcal M$) and realisations of approaches to causal inference through string diagrams~\cite{diagcausality,aleks}. 
 
\section{Conclusion}

Functional paradigms help tame the wild world of interactive computation. 
On the other hand, discovering the simplifying paradigms has often required considerable ingenuity, for example, by G\"odel in his Dialectica Interpretation, or Girard in Geometry of Interaction. 

  The challenges to a functional approach are even more acute with enrichments, say to  probabilistic, quantum or real number computation.  The traditional categories of mathematics do not often support all the features required by computation.  
They often don't have function spaces or 
support recursion. Their extension to computational features has often to be 
 dealt with separately, and ingeniously,
 for example, 
by replacing Borel spaces by quasi Borel spaces to support recursion and higher-order with probability~\cite{Heunen_2017} or the category CPM of completely positive maps by a completion for quantum lambda calculi~\cite{Selingeretal}. Concurrent games and strategies provide enough computational infrastructure that traditional symmetric monoidal categories suffice (the unit interval for probability, Markov kernels for general distributions, CPM for quantum, or smooth maps for differentiation).

As a
model of interaction, concurrent games and strategies are more technically challenging and require a new, more local, way of thinking.  But,  as has been demonstrated here, they can provide a broad general context for interaction which can be specialised to functional paradigms, also in providing enrichments to probabilistic, quantum and real number computation, without requiring clever extensions to the traditional categories of mathematics.  

Concurrent games and strategies can also provide a rationale for new definitions. The form of 
dependent optic described here appears to be new.  It is {\em derived} as a characterisation of nondeterministic strategies between container games.  Contrast this with the incomplete search for a categorical axiomatics of dependent optics described in the blog post~\cite{hedges-depoptics}.  This is not a criticism of axiomatisations but does make the obvious point that they are 
best guided by concrete examples---of which concurrent games and strategies and their enrichments are a rich source.  (A broader characterisation of dependent optics would ensue if  games carried symmetry; then the configurations of a game would form a proper category rather than a partial order.)

There 
is work to do, 
specifically in extending the work here to games with symmetry~\cite{lics14}.  But a lot can be said for a single, expressive, 
intrinsically higher-order framework which readily adapts to enrichments. 

One challenge is that of connecting 
concurrent games and strategies 
with the theory of effects~\cite{Moggi,powerplotkin}, specifically with understanding effect handlers~\cite{effecthandlers} as concurrent strategies.  Though superficially rather different, 
effect handlers and concurrent strategies have very similar roles: both are concerned with orchestrating the future of a computation contingent on its past and its environment.   Through the work of this paper, the language of strategies outlined in Section~\ref{sec:corelang} is able to express complex functional dependencies---see Section~\ref{sec:enrichment}: how can it best be extended to existing theories of effects and effect handlers? As a beginning, the ``detectors'' of Example~\ref{ex:detector}  extend to a form of ``event handler.''
In the other direction, such an investigation should suggest ways to enhance effects and effect handlers to
support richer forms of parallel computation. 

\begin{acks}  
    I'm grateful for discussions with Bob Atkey, Matteo Capucci, Fredrik Nordvall Forsberg,
Bruno Gavranovic,
Neil Ghani,
Dan Ghica,
Samuel Ben Hamou, 
Jules Hedges, Martin Hyland, Clemens Kupke,
J\'er\'emy Ledent, Simon Mirwasser, Valeria de Paiva, Hugo Paquet, Gordon Plotkin and James Wood; 
I learnt of optics and combs from Jules and J\'er\'emy.  Samuel Ben Hamou, ENS Paris-Saclay, 
verified the early part of the dialectica-games section for his student internship.  Section~\ref{sec:enrichment} is directly inspired by 
joint work with Pierre Clairambault and  Marc de Visme.
\end{acks}   

\bibliographystyle{ACM-Reference-Format}
\bibliography{biblio}


\begin{thebibliography}{65}


\ifx \showCODEN    \undefined \def \showCODEN     #1{\unskip}     \fi
\ifx \showDOI      \undefined \def \showDOI       #1{#1}\fi
\ifx \showISBNx    \undefined \def \showISBNx     #1{\unskip}     \fi
\ifx \showISBNxiii \undefined \def \showISBNxiii  #1{\unskip}     \fi
\ifx \showISSN     \undefined \def \showISSN      #1{\unskip}     \fi
\ifx \showLCCN     \undefined \def \showLCCN      #1{\unskip}     \fi
\ifx \shownote     \undefined \def \shownote      #1{#1}          \fi
\ifx \showarticletitle \undefined \def \showarticletitle #1{#1}   \fi
\ifx \showURL      \undefined \def \showURL       {\relax}        \fi
\providecommand\bibfield[2]{#2}
\providecommand\bibinfo[2]{#2}
\providecommand\natexlab[1]{#1}
\providecommand\showeprint[2][]{arXiv:#2}

\bibitem[Abadi et~al\mbox{.}(2017)]%
        {tensorflow}
\bibfield{author}{\bibinfo{person}{Mart{\'{\i}}n Abadi},
  \bibinfo{person}{Michael Isard}, {and} \bibinfo{person}{Derek~Gordon
  Murray}.} \bibinfo{year}{2017}\natexlab{}.
\newblock \showarticletitle{A computational model for TensorFlow: an
  introduction}. In \bibinfo{booktitle}{\emph{Proceedings of the 1st {ACM}
  {SIGPLAN} International Workshop on Machine Learning and Programming
  Languages, MAPL@PLDI 2017}}. \bibinfo{publisher}{{ACM}},
  \bibinfo{pages}{1--7}.
\newblock
\urldef\tempurl%
\url{https://doi.org/10.1145/3088525.3088527}
\showDOI{\tempurl}


\bibitem[Abbott et~al\mbox{.}(2005)]%
        {containers}
\bibfield{author}{\bibinfo{person}{Michael~Gordon Abbott},
  \bibinfo{person}{Thorsten Altenkirch}, {and} \bibinfo{person}{Neil Ghani}.}
  \bibinfo{year}{2005}\natexlab{}.
\newblock \showarticletitle{Containers: Constructing strictly positive types}.
\newblock \bibinfo{journal}{\emph{Theor. Comput. Sci.}} \bibinfo{volume}{342},
  \bibinfo{number}{1} (\bibinfo{year}{2005}), \bibinfo{pages}{3--27}.
\newblock
\urldef\tempurl%
\url{https://doi.org/10.1016/j.tcs.2005.06.002}
\showDOI{\tempurl}


\bibitem[Abramsky et~al\mbox{.}(2002)]%
        {AbramskyHaghverdiScott}
\bibfield{author}{\bibinfo{person}{Samson Abramsky}, \bibinfo{person}{Esfandiar
  Haghverdi}, {and} \bibinfo{person}{Philip~J. Scott}.}
  \bibinfo{year}{2002}\natexlab{}.
\newblock \showarticletitle{Geometry of Interaction and Linear Combinatory
  Algebras}.
\newblock \bibinfo{journal}{\emph{Math. Struct. Comput. Sci.}}
  \bibinfo{volume}{12}, \bibinfo{number}{5} (\bibinfo{year}{2002}),
  \bibinfo{pages}{625--665}.
\newblock
\urldef\tempurl%
\url{https://doi.org/10.1017/S0960129502003730}
\showDOI{\tempurl}


\bibitem[Abramsky and Jagadeesan(1994)]%
        {nfgoi}
\bibfield{author}{\bibinfo{person}{Samson Abramsky} {and}
  \bibinfo{person}{Radha Jagadeesan}.} \bibinfo{year}{1994}\natexlab{}.
\newblock \showarticletitle{New Foundations for the Geometry of Interaction}.
\newblock \bibinfo{journal}{\emph{Inf. Comput.}} \bibinfo{volume}{111},
  \bibinfo{number}{1} (\bibinfo{year}{1994}), \bibinfo{pages}{53--119}.
\newblock
\urldef\tempurl%
\url{https://doi.org/10.1006/inco.1994.1041}
\showDOI{\tempurl}


\bibitem[Alcolei(2019)]%
        {aurore}
\bibfield{author}{\bibinfo{person}{Aurore Alcolei}.}
  \bibinfo{year}{2019}\natexlab{}.
\newblock \emph{\bibinfo{title}{Enriched concurrent games : witnesses for
  proofs and resource analysis. (Jeux concurrents enrichis : t{\'{e}}moins pour
  les preuves et les ressources)}}.
\newblock \bibinfo{thesistype}{Ph.\,D. Dissertation}.
  \bibinfo{school}{University of Lyon, France}.
\newblock
\urldef\tempurl%
\url{https://tel.archives-ouvertes.fr/tel-02448974}
\showURL{%
\tempurl}


\bibitem[Andr\'e~Joyal and Verity(1996)]%
        {joyal-street-verity}
\bibfield{author}{\bibinfo{person}{Ross~Street Andr\'e~Joyal} {and}
  \bibinfo{person}{Dominic Verity}.} \bibinfo{year}{1996}\natexlab{}.
\newblock \showarticletitle{Traced monoidal categories}.
\newblock \bibinfo{journal}{\emph{Math. Proc. Camb. Phil. Soc. (1996), 119}}
  (\bibinfo{year}{1996}), \bibinfo{pages}{447--468}.
\newblock


\bibitem[Avigad and Feferman(1999)]%
        {dialecticaInt}
\bibfield{author}{\bibinfo{person}{Jeremy Avigad} {and}
  \bibinfo{person}{Solomon Feferman}.} \bibinfo{year}{1999}\natexlab{}.
\newblock \showarticletitle{Gödel's functional ("Dialectica") interpretation}.
\newblock  (\bibinfo{year}{1999}), \bibinfo{pages}{337–405}.
\newblock


\bibitem[Berry(1978)]%
        {berry}
\bibfield{author}{\bibinfo{person}{G{\'e}rard Berry}.}
  \bibinfo{year}{1978}\natexlab{}.
\newblock \showarticletitle{Stable Models of Typed lambda-Calculi}. In
  \bibinfo{booktitle}{\emph{ICALP}} \emph{(\bibinfo{series}{Lecture Notes in
  Computer Science}, Vol.~\bibinfo{volume}{62})}.
  \bibinfo{publisher}{Springer}, \bibinfo{pages}{72--89}.
\newblock


\bibitem[Brookes et~al\mbox{.}(1984)]%
        {hoare}
\bibfield{author}{\bibinfo{person}{S. Brookes}, \bibinfo{person}{C.~A.~R.
  Hoare}, {and} \bibinfo{person}{A.~W. Roscoe}.}
  \bibinfo{year}{1984}\natexlab{}.
\newblock \showarticletitle{A Theory of Communicating Sequential Processes}.
\newblock \bibinfo{journal}{\emph{J. ACM}}  \bibinfo{volume}{31}
  (\bibinfo{year}{1984}), \bibinfo{pages}{560--599}.
\newblock


\bibitem[Capucci et~al\mbox{.}(2021)]%
        {capucci2021foundations}
\bibfield{author}{\bibinfo{person}{Matteo Capucci}, \bibinfo{person}{Bruno
  Gavranović}, \bibinfo{person}{Jules Hedges}, {and}
  \bibinfo{person}{Eigil~Fjeldgren Rischel}.} \bibinfo{year}{2021}\natexlab{}.
\newblock \bibinfo{title}{Towards foundations of categorical cybernetics}.
\newblock
\newblock
\showeprint[arxiv]{2105.06332}~[math.CT]


\bibitem[Castellan et~al\mbox{.}(2014a)]%
        {lics14}
\bibfield{author}{\bibinfo{person}{Simon Castellan}, \bibinfo{person}{Pierre
  Clairambault}, {and} \bibinfo{person}{Glynn Winskel}.}
  \bibinfo{year}{2014}\natexlab{a}.
\newblock \showarticletitle{Symmetry in concurrent games}. In
  \bibinfo{booktitle}{\emph{Joint Meeting of the Twenty-Third {EACSL} Annual
  Conference on Computer Science Logic {(CSL)} and the Twenty-Ninth Annual
  {ACM/IEEE} Symposium on Logic in Computer Science (LICS), {CSL-LICS} '14,
  Vienna, Austria, July 14 - 18, 2014}}. \bibinfo{publisher}{{ACM}}.
\newblock


\bibitem[Castellan et~al\mbox{.}(2014b)]%
        {stratasproc}
\bibfield{author}{\bibinfo{person}{Simon Castellan}, \bibinfo{person}{Jonathan
  Hayman}, \bibinfo{person}{Marc Lasson}, {and} \bibinfo{person}{Glynn
  Winskel}.} \bibinfo{year}{2014}\natexlab{b}.
\newblock \showarticletitle{Strategies as concurrent processes}.
\newblock \bibinfo{journal}{\emph{Electr. Notes Theor. Comput. Sci.}}
  \bibinfo{volume}{308} (\bibinfo{year}{2014}), \bibinfo{pages}{87--107}.
\newblock


\bibitem[Castellan and Yoshida(2019)]%
        {session-types}
\bibfield{author}{\bibinfo{person}{Simon Castellan} {and}
  \bibinfo{person}{Nobuko Yoshida}.} \bibinfo{year}{2019}\natexlab{}.
\newblock \showarticletitle{Two sides of the same coin: session types and game
  semantics: a synchronous side and an asynchronous side}.
\newblock \bibinfo{journal}{\emph{Proc. {ACM} Program. Lang.}}
  \bibinfo{volume}{3}, \bibinfo{number}{{POPL}} (\bibinfo{year}{2019}),
  \bibinfo{pages}{27:1--27:29}.
\newblock
\urldef\tempurl%
\url{https://doi.org/10.1145/3290340}
\showDOI{\tempurl}


\bibitem[Chiribella et~al\mbox{.}(2008)]%
        {Chiribella2008}
\bibfield{author}{\bibinfo{person}{G. Chiribella}, \bibinfo{person}{G.~M.
  D’Ariano}, {and} \bibinfo{person}{P. Perinotti}.}
  \bibinfo{year}{2008}\natexlab{}.
\newblock \showarticletitle{Quantum Circuit Architecture}.
\newblock \bibinfo{journal}{\emph{Physical Review Letters}}
  \bibinfo{volume}{101}, \bibinfo{number}{6} (\bibinfo{date}{Aug}
  \bibinfo{year}{2008}).
\newblock
\showISSN{1079-7114}
\urldef\tempurl%
\url{https://doi.org/10.1103/physrevlett.101.060401}
\showDOI{\tempurl}


\bibitem[Clairambault et~al\mbox{.}(2019)]%
        {popl19}
\bibfield{author}{\bibinfo{person}{Pierre Clairambault}, \bibinfo{person}{Marc
  de Visme}, {and} \bibinfo{person}{Glynn Winskel}.}
  \bibinfo{year}{2019}\natexlab{}.
\newblock \showarticletitle{Game semantics for quantum programming}.
\newblock \bibinfo{journal}{\emph{Proc. {ACM} Program. Lang.}}
  \bibinfo{volume}{3}, \bibinfo{number}{{POPL}} (\bibinfo{year}{2019}),
  \bibinfo{pages}{32:1--32:29}.
\newblock
\urldef\tempurl%
\url{https://doi.org/10.1145/3290345}
\showDOI{\tempurl}


\bibitem[Clairambault et~al\mbox{.}(2012)]%
        {lics12}
\bibfield{author}{\bibinfo{person}{Pierre Clairambault},
  \bibinfo{person}{Julian Gutierrez}, {and} \bibinfo{person}{Glynn Winskel}.}
  \bibinfo{year}{2012}\natexlab{}.
\newblock \showarticletitle{The Winning Ways of Concurrent Games}. In
  \bibinfo{booktitle}{\emph{LICS 2012: 235-244}}.
\newblock


\bibitem[Clairambault and Winskel(2013)]%
        {payoff}
\bibfield{author}{\bibinfo{person}{Pierre Clairambault} {and}
  \bibinfo{person}{Glynn Winskel}.} \bibinfo{year}{2013}\natexlab{}.
\newblock \showarticletitle{On Concurrent Games with Payoff}.
\newblock \bibinfo{journal}{\emph{Electr. Notes Theor. Comput. Sci. 298:
  71-92}} (\bibinfo{year}{2013}).
\newblock


\bibitem[Conway(2000)]%
        {conway}
\bibfield{author}{\bibinfo{person}{John Conway}.}
  \bibinfo{year}{2000}\natexlab{}.
\newblock \bibinfo{booktitle}{\emph{On Numbers and Games}}.
\newblock \bibinfo{publisher}{Wellesley, MA: A K Peters}.
\newblock


\bibitem[Coquand et~al\mbox{.}(1987)]%
        {dIPoly}
\bibfield{author}{\bibinfo{person}{Thierry Coquand}, \bibinfo{person}{Carl~A.
  Gunter}, {and} \bibinfo{person}{Glynn Winskel}.}
  \bibinfo{year}{1987}\natexlab{}.
\newblock \showarticletitle{DI-Domains as a Model of Polymorphism}. In
  \bibinfo{booktitle}{\emph{Mathematical Foundations of Programming Language
  Semantics, 3rd Workshop, Tulane University, New Orleans, Louisiana, USA,
  April 8-10, 1987, Proceedings}} \emph{(\bibinfo{series}{Lecture Notes in
  Computer Science}, Vol.~\bibinfo{volume}{298})},
  \bibfield{editor}{\bibinfo{person}{Michael~G. Main}, \bibinfo{person}{Austin
  Melton}, \bibinfo{person}{Michael~W. Mislove}, {and}
  \bibinfo{person}{David~A. Schmidt}} (Eds.). \bibinfo{publisher}{Springer},
  \bibinfo{pages}{344--363}.
\newblock
\urldef\tempurl%
\url{https://doi.org/10.1007/3-540-19020-1\_18}
\showDOI{\tempurl}


\bibitem[Coquand et~al\mbox{.}(1989)]%
        {CGW}
\bibfield{author}{\bibinfo{person}{Thierry Coquand}, \bibinfo{person}{Carl~A.
  Gunter}, {and} \bibinfo{person}{Glynn Winskel}.}
  \bibinfo{year}{1989}\natexlab{}.
\newblock \showarticletitle{Domain Theoretic Models of Polymorphism}.
\newblock \bibinfo{journal}{\emph{Inf. Comput.}} \bibinfo{volume}{81},
  \bibinfo{number}{2} (\bibinfo{year}{1989}), \bibinfo{pages}{123--167}.
\newblock
\urldef\tempurl%
\url{https://doi.org/10.1016/0890-5401(89)90068-0}
\showDOI{\tempurl}


\bibitem[de~Paiva(1988)]%
        {valeria}
\bibfield{author}{\bibinfo{person}{Valeria de Paiva}.}
  \bibinfo{year}{1988}\natexlab{}.
\newblock \bibinfo{booktitle}{\emph{The Dialectica categories}}.
\newblock \bibinfo{publisher}{PhD Thesis, University of Cambridge}.
\newblock


\bibitem[Faggian and Piccolo(2009)]%
        {FP}
\bibfield{author}{\bibinfo{person}{Claudia Faggian} {and}
  \bibinfo{person}{Mauro Piccolo}.} \bibinfo{year}{2009}\natexlab{}.
\newblock \showarticletitle{Partial orders, event structures and linear
  strategies}. In \bibinfo{booktitle}{\emph{TLCA '09}}
  \emph{(\bibinfo{series}{LNCS}, Vol.~\bibinfo{volume}{5608})}.
  \bibinfo{publisher}{Springer}.
\newblock


\bibitem[Feferman(1996)]%
        {kreisel}
\bibfield{author}{\bibinfo{person}{Solomon Feferman}.}
  \bibinfo{year}{1996}\natexlab{}.
\newblock \showarticletitle{Kreisel’s ‘Unwinding Program’, Kreiseliana:
  about and around Georg Kreisel}.
\newblock  (\bibinfo{year}{1996}), \bibinfo{pages}{247–273}.
\newblock


\bibitem[Fiore et~al\mbox{.}(2018)]%
        {Kleisli}
\bibfield{author}{\bibinfo{person}{Marcelo Fiore}, \bibinfo{person}{Nicola
  Gambino}, \bibinfo{person}{Martin Hyland}, {and} \bibinfo{person}{Glynn
  Winskel}.} \bibinfo{year}{2018}\natexlab{}.
\newblock \showarticletitle{Relative pseudomonads, Kelisli bicategories and
  substitution monoidal structures}.
\newblock \bibinfo{journal}{\emph{Selecta Mathematica - New Series}}
  \bibinfo{volume}{24}, \bibinfo{number}{3} (\bibinfo{year}{2018}),
  \bibinfo{pages}{2791--2830}.
\newblock


\bibitem[Fong et~al\mbox{.}(2019)]%
        {openlearners}
\bibfield{author}{\bibinfo{person}{Brendan Fong}, \bibinfo{person}{David~I.
  Spivak}, {and} \bibinfo{person}{R{\'{e}}my Tuy{\'{e}}ras}.}
  \bibinfo{year}{2019}\natexlab{}.
\newblock \showarticletitle{Backprop as Functor: {A} compositional perspective
  on supervised learning}. In \bibinfo{booktitle}{\emph{34th Annual {ACM/IEEE}
  Symposium on Logic in Computer Science, {LICS} 2019, Vancouver, BC, Canada,
  June 24-27, 2019}}. \bibinfo{publisher}{{IEEE}}, \bibinfo{pages}{1--13}.
\newblock
\urldef\tempurl%
\url{https://doi.org/10.1109/LICS.2019.8785665}
\showDOI{\tempurl}


\bibitem[Foster et~al\mbox{.}(2007)]%
        {pierce}
\bibfield{author}{\bibinfo{person}{J.~Nathan Foster},
  \bibinfo{person}{Michael~B. Greenwald}, \bibinfo{person}{Jonathan~T. Moore},
  \bibinfo{person}{Benjamin~C. Pierce}, {and} \bibinfo{person}{Alan Schmitt}.}
  \bibinfo{year}{2007}\natexlab{}.
\newblock \showarticletitle{Combinators for bidirectional tree transformations:
  {A} linguistic approach to the view-update problem}.
\newblock \bibinfo{journal}{\emph{{ACM} Trans. Program. Lang. Syst.}}
  \bibinfo{volume}{29}, \bibinfo{number}{3} (\bibinfo{year}{2007}),
  \bibinfo{pages}{17}.
\newblock
\urldef\tempurl%
\url{https://doi.org/10.1145/1232420.1232424}
\showDOI{\tempurl}


\bibitem[Ghani et~al\mbox{.}(2018)]%
        {opengames}
\bibfield{author}{\bibinfo{person}{Neil Ghani}, \bibinfo{person}{Jules Hedges},
  \bibinfo{person}{Viktor Winschel}, {and} \bibinfo{person}{Philipp Zahn}.}
  \bibinfo{year}{2018}\natexlab{}.
\newblock \showarticletitle{Compositional Game Theory}. In
  \bibinfo{booktitle}{\emph{Proceedings of the 33rd Annual {ACM/IEEE} Symposium
  on Logic in Computer Science, {LICS} 2018, Oxford, UK, July 09-12, 2018}},
  \bibfield{editor}{\bibinfo{person}{Anuj Dawar} {and} \bibinfo{person}{Erich
  Gr{\"{a}}del}} (Eds.). \bibinfo{publisher}{{ACM}}, \bibinfo{pages}{472--481}.
\newblock
\urldef\tempurl%
\url{https://doi.org/10.1145/3209108.3209165}
\showDOI{\tempurl}


\bibitem[Girard(1989)]%
        {GirardGoI1}
\bibfield{author}{\bibinfo{person}{Jean-Yves Girard}.}
  \bibinfo{year}{1989}\natexlab{}.
\newblock \showarticletitle{Geometry of interaction I: Interpretation of System
  F}. In \bibinfo{booktitle}{\emph{Logic Colloquium ’88}},
  \bibfield{editor}{\bibinfo{person}{C.~Bonotto}, \bibinfo{person}{S~R.~Ferro},
  \bibinfo{person}{Valentini}, {and} \bibinfo{person}{A.~Zanardo}} (Eds.).
  \bibinfo{publisher}{North-Holland}, \bibinfo{pages}{221–260}.
\newblock


\bibitem[Gonthier et~al\mbox{.}(1992)]%
        {GoI-optred}
\bibfield{author}{\bibinfo{person}{Georges Gonthier},
  \bibinfo{person}{Mart{\'{\i}}n Abadi}, {and} \bibinfo{person}{Jean{-}Jacques
  L{\'{e}}vy}.} \bibinfo{year}{1992}\natexlab{}.
\newblock \showarticletitle{The Geometry of Optimal Lambda Reduction}. In
  \bibinfo{booktitle}{\emph{Conference Record of the Nineteenth Annual {ACM}
  {SIGPLAN-SIGACT} Symposium on Principles of Programming Languages,
  Albuquerque, New Mexico, USA, January 19-22, 1992}},
  \bibfield{editor}{\bibinfo{person}{Ravi Sethi}} (Ed.).
  \bibinfo{publisher}{{ACM} Press}, \bibinfo{pages}{15--26}.
\newblock
\urldef\tempurl%
\url{https://doi.org/10.1145/143165.143172}
\showDOI{\tempurl}


\bibitem[Hedges(2020)]%
        {hedges-depoptics}
\bibfield{author}{\bibinfo{person}{Jules Hedges}.}
  \bibinfo{year}{2020}\natexlab{}.
\newblock \bibinfo{title}{Towards dependent optics (Blog Post)}.
\newblock
\newblock
\urldef\tempurl%
\url{https://julesh.com/2020/06/10/towards-dependent-optics/}
\showURL{%
\tempurl}


\bibitem[Heunen et~al\mbox{.}(2017)]%
        {Heunen_2017}
\bibfield{author}{\bibinfo{person}{Chris Heunen}, \bibinfo{person}{Ohad
  Kammar}, \bibinfo{person}{Sam Staton}, {and} \bibinfo{person}{Hongseok
  Yang}.} \bibinfo{year}{2017}\natexlab{}.
\newblock \showarticletitle{A convenient category for higher-order probability
  theory}. In \bibinfo{booktitle}{\emph{2017 32nd Annual {ACM}/{IEEE} Symposium
  on Logic in Computer Science ({LICS})}}. \bibinfo{publisher}{{IEEE}}.
\newblock
\urldef\tempurl%
\url{https://doi.org/10.1109/lics.2017.8005137}
\showDOI{\tempurl}


\bibitem[Hyland(2010)]%
        {hyland-gendomthy}
\bibfield{author}{\bibinfo{person}{Martin Hyland}.}
  \bibinfo{year}{2010}\natexlab{}.
\newblock \showarticletitle{Some reasons for generalising domain theory}.
\newblock \bibinfo{journal}{\emph{Mathematical Structures in Computer Science}}
  \bibinfo{volume}{20}, \bibinfo{number}{2} (\bibinfo{year}{2010}),
  \bibinfo{pages}{239--265}.
\newblock


\bibitem[Jacobs et~al\mbox{.}(2019)]%
        {diagcausality}
\bibfield{author}{\bibinfo{person}{Bart Jacobs}, \bibinfo{person}{Aleks
  Kissinger}, {and} \bibinfo{person}{Fabio Zanasi}.}
  \bibinfo{year}{2019}\natexlab{}.
\newblock \showarticletitle{Causal Inference by String Diagram Surgery}. In
  \bibinfo{booktitle}{\emph{Foundations of Software Science and Computation
  Structures - 22nd International Conference, {FOSSACS} 2019, Held as Part of
  the European Joint Conferences on Theory and Practice of Software, {ETAPS}
  2019, Prague, Czech Republic, April 6-11, 2019, Proceedings}}
  \emph{(\bibinfo{series}{Lecture Notes in Computer Science},
  Vol.~\bibinfo{volume}{11425})}, \bibfield{editor}{\bibinfo{person}{Mikolaj
  Bojanczyk} {and} \bibinfo{person}{Alex Simpson}} (Eds.).
  \bibinfo{publisher}{Springer}, \bibinfo{pages}{313--329}.
\newblock
\urldef\tempurl%
\url{https://doi.org/10.1007/978-3-030-17127-8\_18}
\showDOI{\tempurl}


\bibitem[Johnson and Yau(2020)]%
        {2DCat}
\bibfield{author}{\bibinfo{person}{Niles Johnson} {and} \bibinfo{person}{Donald
  Yau}.} \bibinfo{year}{2020}\natexlab{}.
\newblock \bibinfo{title}{2-Dimensional Categories}.
\newblock
\newblock
\showeprint[arxiv]{2002.06055}~[math.CT]


\bibitem[Joyal(1997)]%
        {joyal}
\bibfield{author}{\bibinfo{person}{Andre Joyal}.}
  \bibinfo{year}{1997}\natexlab{}.
\newblock \showarticletitle{Remarques sur la th\'eorie des jeux \`a deux
  personnes}.
\newblock \bibinfo{journal}{\emph{Gazette des sciences math\'ematiques du
  Qu\'ebec, 1(4)}} (\bibinfo{year}{1997}).
\newblock


\bibitem[Kissinger and Uijlen(2017)]%
        {aleks}
\bibfield{author}{\bibinfo{person}{Aleks Kissinger} {and}
  \bibinfo{person}{Sander Uijlen}.} \bibinfo{year}{2017}\natexlab{}.
\newblock \showarticletitle{A categorical semantics for causal structure}. In
  \bibinfo{booktitle}{\emph{32nd Annual {ACM/IEEE} Symposium on Logic in
  Computer Science, {LICS} 2017, Reykjavik, Iceland, June 20-23, 2017}}.
  \bibinfo{publisher}{{IEEE} Computer Society}, \bibinfo{pages}{1--12}.
\newblock
\urldef\tempurl%
\url{https://doi.org/10.1109/LICS.2017.8005095}
\showDOI{\tempurl}


\bibitem[Kohlenbach(2008)]%
        {kohlenbach}
\bibfield{author}{\bibinfo{person}{Ulrich Kohlenbach}.}
  \bibinfo{year}{2008}\natexlab{}.
\newblock \bibinfo{booktitle}{\emph{Applied Proof Theory: Proof Interpretations
  and their Use in Mathematics}}.
\newblock \bibinfo{publisher}{Springer Monographs in Mathematics}.
\newblock


\bibitem[Lago et~al\mbox{.}(2016)]%
        {quantumGoI}
\bibfield{author}{\bibinfo{person}{Ugo~Dal Lago}, \bibinfo{person}{Claudia
  Faggian}, \bibinfo{person}{Beno{\^{\i}}t Valiron}, {and}
  \bibinfo{person}{Akira Yoshimizu}.} \bibinfo{year}{2016}\natexlab{}.
\newblock \showarticletitle{The Geometry of Parallelism. Classical,
  Probabilistic, and Quantum Effects}.
\newblock \bibinfo{journal}{\emph{CoRR}}  \bibinfo{volume}{abs/1610.09629}
  (\bibinfo{year}{2016}).
\newblock
\showeprint[arXiv]{1610.09629}
\urldef\tempurl%
\url{http://arxiv.org/abs/1610.09629}
\showURL{%
\tempurl}


\bibitem[Mackie(1995)]%
        {mackie}
\bibfield{author}{\bibinfo{person}{Ian Mackie}.}
  \bibinfo{year}{1995}\natexlab{}.
\newblock \showarticletitle{The Geometry of Interaction Machine}. In
  \bibinfo{booktitle}{\emph{Conference Record of POPL'95: 22nd {ACM}
  {SIGPLAN-SIGACT} Symposium on Principles of Programming Languages, San
  Francisco, California, USA, January 23-25, 1995}},
  \bibfield{editor}{\bibinfo{person}{Ron~K. Cytron} {and}
  \bibinfo{person}{Peter Lee}} (Eds.). \bibinfo{publisher}{{ACM} Press},
  \bibinfo{pages}{198--208}.
\newblock
\urldef\tempurl%
\url{https://doi.org/10.1145/199448.199483}
\showDOI{\tempurl}


\bibitem[Melli\`es and Mimram(2007)]%
        {Asgames}
\bibfield{author}{\bibinfo{person}{Paul-Andr\'e Melli\`es} {and}
  \bibinfo{person}{Samuel Mimram}.} \bibinfo{year}{2007}\natexlab{}.
\newblock \showarticletitle{Asynchronous games : innocence without
  alternation}. In \bibinfo{booktitle}{\emph{CONCUR '07}}
  \emph{(\bibinfo{series}{LNCS}, Vol.~\bibinfo{volume}{4703})}.
  \bibinfo{publisher}{Springer}.
\newblock


\bibitem[Milner(1980)]%
        {milner}
\bibfield{author}{\bibinfo{person}{Robin Milner}.}
  \bibinfo{year}{1980}\natexlab{}.
\newblock \bibinfo{booktitle}{\emph{A Calculus of Communicating Systems}}.
  \bibinfo{series}{Lecture Notes in Computer Science},
  Vol.~\bibinfo{volume}{92}.
\newblock \bibinfo{publisher}{Springer}.
\newblock
\showISBNx{3-540-10235-3}
\urldef\tempurl%
\url{https://doi.org/10.1007/3-540-10235-3}
\showDOI{\tempurl}


\bibitem[Moggi(1989)]%
        {Moggi}
\bibfield{author}{\bibinfo{person}{Eugenio Moggi}.}
  \bibinfo{year}{1989}\natexlab{}.
\newblock \showarticletitle{Computational Lambda-Calculus and Monads}. In
  \bibinfo{booktitle}{\emph{Proceedings of the Fourth Annual Symposium on Logic
  in Computer Science {(LICS} '89), Pacific Grove, California, USA, June 5-8,
  1989}}. \bibinfo{publisher}{{IEEE} Computer Society},
  \bibinfo{pages}{14--23}.
\newblock
\urldef\tempurl%
\url{https://doi.org/10.1109/LICS.1989.39155}
\showDOI{\tempurl}


\bibitem[Muroya and Ghica(2019)]%
        {ghica}
\bibfield{author}{\bibinfo{person}{Koko Muroya} {and} \bibinfo{person}{Dan~R.
  Ghica}.} \bibinfo{year}{2019}\natexlab{}.
\newblock \showarticletitle{The Dynamic Geometry of Interaction Machine: {A}
  Token-Guided Graph Rewriter}.
\newblock \bibinfo{journal}{\emph{Log. Methods Comput. Sci.}}
  \bibinfo{volume}{15}, \bibinfo{number}{4} (\bibinfo{year}{2019}).
\newblock
\urldef\tempurl%
\url{https://doi.org/10.23638/LMCS-15(4:7)2019}
\showDOI{\tempurl}


\bibitem[Nygaard(2003)]%
        {mikkel}
\bibfield{author}{\bibinfo{person}{Mikkel Nygaard}.}
  \bibinfo{year}{2003}\natexlab{}.
\newblock \bibinfo{booktitle}{\emph{Domain theory for concurrency}}.
\newblock \bibinfo{publisher}{PhD Thesis, Aarhus University}.
\newblock
\urldef\tempurl%
\url{https://www.brics.dk/DS/03/13/BRICS-DS-03-13.pdf}
\showURL{%
\tempurl}


\bibitem[Nygaard and Winskel(2002)]%
        {lics02}
\bibfield{author}{\bibinfo{person}{Mikkel Nygaard} {and} \bibinfo{person}{Glynn
  Winskel}.} \bibinfo{year}{2002}\natexlab{}.
\newblock \showarticletitle{Linearity in Process Languages}. In
  \bibinfo{booktitle}{\emph{LICS'02}}. \bibinfo{publisher}{IEEE Computer
  Society}.
\newblock


\bibitem[Oles(1982)]%
        {oles}
\bibfield{author}{\bibinfo{person}{Frank~J. Oles}.}
  \bibinfo{year}{1982}\natexlab{}.
\newblock \bibinfo{booktitle}{\emph{A category theoretic approach to the
  semantics of programming languages}}.
\newblock \bibinfo{publisher}{PhD Thesis, University of Syracuse}.
\newblock


\bibitem[Pagani et~al\mbox{.}(2014)]%
        {Selingeretal}
\bibfield{author}{\bibinfo{person}{Michele Pagani}, \bibinfo{person}{Peter
  Selinger}, {and} \bibinfo{person}{Beno\^{\i}t Valiron}.}
  \bibinfo{year}{2014}\natexlab{}.
\newblock \showarticletitle{Applying Quantitative Semantics to Higher-Order
  Quantum Computing}.
\newblock \bibinfo{journal}{\emph{POPL '14}} (\bibinfo{year}{2014}),
  \bibinfo{pages}{647–658}.
\newblock
\urldef\tempurl%
\url{https://doi.org/10.1145/2578855.2535879}
\showDOI{\tempurl}


\bibitem[Paquet(2020)]%
        {hugo-thesis}
\bibfield{author}{\bibinfo{person}{Hugo Paquet}.}
  \bibinfo{year}{2020}\natexlab{}.
\newblock \emph{\bibinfo{title}{Probabilistic concurrent game semantics}}.
\newblock \bibinfo{thesistype}{Ph.\,D. Dissertation}. \bibinfo{school}{Computer
  Laboratory, University of Cambridge, UK}.
\newblock


\bibitem[Paquet and Winskel(2018)]%
        {mfps2018}
\bibfield{author}{\bibinfo{person}{Hugo Paquet} {and} \bibinfo{person}{Glynn
  Winskel}.} \bibinfo{year}{2018}\natexlab{}.
\newblock \showarticletitle{Continuous Probability Distributions in Concurrent
  Games}. In \bibinfo{booktitle}{\emph{Proceedings of the Thirty-Fourth
  Conference on the Mathematical Foundations of Programming Semantics, {MFPS}
  2018, Dalhousie University, Halifax, Canada, June 6-9, 2018}}
  \emph{(\bibinfo{series}{Electronic Notes in Theoretical Computer Science},
  Vol.~\bibinfo{volume}{341})}, \bibfield{editor}{\bibinfo{person}{Sam Staton}}
  (Ed.). \bibinfo{publisher}{Elsevier}, \bibinfo{pages}{321--344}.
\newblock
\urldef\tempurl%
\url{https://doi.org/10.1016/j.entcs.2018.11.016}
\showDOI{\tempurl}


\bibitem[Pickering et~al\mbox{.}(2017)]%
        {optics}
\bibfield{author}{\bibinfo{person}{Matthew Pickering}, \bibinfo{person}{Jeremy
  Gibbons}, {and} \bibinfo{person}{Nicolas Wu}.}
  \bibinfo{year}{2017}\natexlab{}.
\newblock \showarticletitle{Profunctor Optics: Modular Data Accessors}.
\newblock \bibinfo{journal}{\emph{Art Sci. Eng. Program.}} \bibinfo{volume}{1},
  \bibinfo{number}{2} (\bibinfo{year}{2017}), \bibinfo{pages}{7}.
\newblock
\urldef\tempurl%
\url{https://doi.org/10.22152/programming-journal.org/2017/1/7}
\showDOI{\tempurl}


\bibitem[Plotkin and Power(2004)]%
        {powerplotkin}
\bibfield{author}{\bibinfo{person}{Gordon~D. Plotkin} {and}
  \bibinfo{person}{A.~John Power}.} \bibinfo{year}{2004}\natexlab{}.
\newblock \showarticletitle{Computational Effects and Operations: An Overview}.
\newblock \bibinfo{journal}{\emph{Electron. Notes Theor. Comput. Sci.}}
  \bibinfo{volume}{73} (\bibinfo{year}{2004}), \bibinfo{pages}{149--163}.
\newblock
\urldef\tempurl%
\url{https://doi.org/10.1016/j.entcs.2004.08.008}
\showDOI{\tempurl}


\bibitem[Plotkin and Pretnar(2009)]%
        {effecthandlers}
\bibfield{author}{\bibinfo{person}{Gordon~D. Plotkin} {and}
  \bibinfo{person}{Matija Pretnar}.} \bibinfo{year}{2009}\natexlab{}.
\newblock \showarticletitle{Handlers of Algebraic Effects}. In
  \bibinfo{booktitle}{\emph{Programming Languages and Systems, 18th European
  Symposium on Programming, {ESOP} 2009, Held as Part of the Joint European
  Conferences on Theory and Practice of Software, {ETAPS} 2009, York, UK, March
  22-29, 2009. Proceedings}} \emph{(\bibinfo{series}{Lecture Notes in Computer
  Science}, Vol.~\bibinfo{volume}{5502})},
  \bibfield{editor}{\bibinfo{person}{Giuseppe Castagna}} (Ed.).
  \bibinfo{publisher}{Springer}, \bibinfo{pages}{80--94}.
\newblock
\urldef\tempurl%
\url{https://doi.org/10.1007/978-3-642-00590-9\_7}
\showDOI{\tempurl}


\bibitem[Rideau and Winskel(2011)]%
        {lics11}
\bibfield{author}{\bibinfo{person}{Silvain Rideau} {and} \bibinfo{person}{Glynn
  Winskel}.} \bibinfo{year}{2011}\natexlab{}.
\newblock \showarticletitle{Concurrent Strategies}. In
  \bibinfo{booktitle}{\emph{LICS 2011}}.
\newblock


\bibitem[Riley(2018)]%
        {optics2}
\bibfield{author}{\bibinfo{person}{Mitchell Riley}.}
  \bibinfo{year}{2018}\natexlab{}.
\newblock \bibinfo{title}{Categories of Optics}.
\newblock
\newblock
\showeprint[arxiv]{1809.00738}~[math.CT]


\bibitem[Saunders-Evans and Winskel(2007)]%
        {SEW}
\bibfield{author}{\bibinfo{person}{Lucy Saunders-Evans} {and}
  \bibinfo{person}{Glynn Winskel}.} \bibinfo{year}{2007}\natexlab{}.
\newblock \showarticletitle{Event Structure Spans for Nondeterministic
  Dataflow}.
\newblock \bibinfo{journal}{\emph{Electr. Notes Theor. Comput. Sci. 175(3):
  109-129}} (\bibinfo{year}{2007}).
\newblock


\bibitem[Winskel(1982)]%
        {icalp82}
\bibfield{author}{\bibinfo{person}{Glynn Winskel}.}
  \bibinfo{year}{1982}\natexlab{}.
\newblock \showarticletitle{Event Structure Semantics for \hbox{CCS} and
  Related Languages}. In \bibinfo{booktitle}{\emph{ICALP'82}}
  \emph{(\bibinfo{series}{LNCS}, Vol.~\bibinfo{volume}{140})}.
  \bibinfo{publisher}{Springer, A full version is available from Winskel's
  Cambridge homepage}.
\newblock


\bibitem[Winskel(1986)]%
        {evstrs}
\bibfield{author}{\bibinfo{person}{Glynn Winskel}.}
  \bibinfo{year}{1986}\natexlab{}.
\newblock \showarticletitle{Event Structures}. In
  \bibinfo{booktitle}{\emph{Advances in Petri Nets}}
  \emph{(\bibinfo{series}{LNCS}, Vol.~\bibinfo{volume}{255})}.
  \bibinfo{publisher}{Springer}, \bibinfo{pages}{325--392}.
\newblock


\bibitem[Winskel(2009)]%
        {primealg}
\bibfield{author}{\bibinfo{person}{Glynn Winskel}.}
  \bibinfo{year}{2009}\natexlab{}.
\newblock \showarticletitle{Prime algebraicity}.
\newblock \bibinfo{journal}{\emph{Theor. Comput. Sci.}} \bibinfo{volume}{410},
  \bibinfo{number}{41} (\bibinfo{year}{2009}), \bibinfo{pages}{4160--4168}.
\newblock
\urldef\tempurl%
\url{https://doi.org/10.1016/j.tcs.2009.06.015}
\showDOI{\tempurl}


\bibitem[Winskel(2011)]%
        {Visions}
\bibfield{author}{\bibinfo{person}{Glynn Winskel}.}
  \bibinfo{year}{2011}\natexlab{}.
\newblock \showarticletitle{Events, Causality and Symmetry}.
\newblock \bibinfo{journal}{\emph{Comput. J.}} \bibinfo{volume}{54},
  \bibinfo{number}{1} (\bibinfo{year}{2011}), \bibinfo{pages}{42--57}.
\newblock


\bibitem[Winskel(2012a)]%
        {DBLP:journals/fac/Winskel12}
\bibfield{author}{\bibinfo{person}{Glynn Winskel}.}
  \bibinfo{year}{2012}\natexlab{a}.
\newblock \showarticletitle{Deterministic concurrent strategies}.
\newblock \bibinfo{journal}{\emph{Formal Asp. Comput.}} \bibinfo{volume}{24},
  \bibinfo{number}{4-6} (\bibinfo{year}{2012}), \bibinfo{pages}{647--660}.
\newblock


\bibitem[Winskel(2012b)]%
        {dexter}
\bibfield{author}{\bibinfo{person}{Glynn Winskel}.}
  \bibinfo{year}{2012}\natexlab{b}.
\newblock \showarticletitle{Winning, Losing and Drawing in Concurrent Games
  with Perfect or Imperfect Information}. In
  \bibinfo{booktitle}{\emph{Festschrift for Dexter Kozen}}
  \emph{(\bibinfo{series}{LNCS}, Vol.~\bibinfo{volume}{7230})}.
  \bibinfo{publisher}{Springer}.
\newblock


\bibitem[Winskel(2013a)]%
        {probstrats}
\bibfield{author}{\bibinfo{person}{Glynn Winskel}.}
  \bibinfo{year}{2013}\natexlab{a}.
\newblock \showarticletitle{Distributed Probabilistic and Quantum Strategies}.
\newblock \bibinfo{journal}{\emph{Electr. Notes Theor. Comput. Sci. 298:
  403-425}} (\bibinfo{year}{2013}).
\newblock


\bibitem[Winskel(2013b)]%
        {fossacs13}
\bibfield{author}{\bibinfo{person}{Glynn Winskel}.}
  \bibinfo{year}{2013}\natexlab{b}.
\newblock \showarticletitle{Strategies as profunctors}. In
  \bibinfo{booktitle}{\emph{FOSSACS 2013}} \emph{(\bibinfo{series}{Lecture
  Notes in Computer Science})}. \bibinfo{publisher}{Springer}.
\newblock


\bibitem[Winskel(2017)]%
        {ecsym-notes}
\bibfield{author}{\bibinfo{person}{Glynn Winskel}.}
  \bibinfo{year}{2017}\natexlab{}.
\newblock \bibinfo{booktitle}{\emph{{ECSYM Notes:} Event Structures, Stable
  Families and Concurrent Games}}.
\newblock
\urldef\tempurl%
\url{http://www.cl.cam.ac.uk/~gw104/ecsym-notes.pdf}
\showURL{%
\tempurl}


\bibitem[Winskel and Nielsen(1995)]%
        {handbook}
\bibfield{author}{\bibinfo{person}{Glynn Winskel} {and} \bibinfo{person}{Mogens
  Nielsen}.} \bibinfo{year}{1995}\natexlab{}.
\newblock \bibinfo{booktitle}{\emph{Handbook of Logic in Computer Science 4}}.
\newblock \bibinfo{publisher}{OUP}, Chapter Models for Concurrency,
  \bibinfo{pages}{1--148}.
\newblock


\end{thebibliography}


\appendix

\section{dI-domains and stable functions}
That dI-domains are exactly the partial orders of configurations of an event structure was first published in~\cite{icalp82}---see the extended version or~\cite{primealg} for the proof.  A {\em stable} function between dI-domains is a Scott continuous function (\ie~preserves least upper bounds of directed sets) which  preserves greatest lower bounds of compatible pairs of elements. G\'erard Berry developed stable domain theory axiomatically, following operational guidelines~\cite{berry}.
For the reader's convenience, we include the constructions on event structures which realise the cartesian-closure of dI-domains and their dependent types.
 
\subsection{Stable function space}\label{sec:Stfnspace}

Berry's cartesian-closed category of dI-domains\footnote{Strictly speaking, Berry defined dI-domains to have a countable basis of finite elements.  Countability plays no role in the work here and we shall not impose it.}  and stable functions can be presented as an equivalent category of event structures~\cite{evstrs}.  We summarise the product and stable function space constructions on stable families and event structures. 

Let $\A$ and $\B$ be stable families with events $A$ and $B$ respectively.   
The product of their domains of configurations is easily realised as a simple parallel composition:  
$
\A\vvbar\B \eqdef \set{x\vvbar y}{x\in \A \ \&\ y\in \B}$.
 
We construct the stable function space of domains 
 as a stable family $[\A\to \B]$. 
The stable family $[\A\to \B]$ comprises those $
f\subseteq \A^o \times B$ for which, for all $x\in \A$,  
\begin{itemize}
\item
 $\set{b}{\exists x'\subseteq x.\ (x',b) \in f} \in \B$\  and 
\item
if $(x',b), (x'', b)\in f$ with $x', x''\subseteq x$  then $x'=x''$.
\end{itemize}

 \begin{theorem}
The construction $[\A\to \B]$ above is a stable family with 
$([\A\to \B], \subseteq)$  order isomorphic to  $[(\A,\subseteq) \to (\B,\subseteq)]$, the stable function space of stable functions,  ordered by the stable order, between dI-domains $(\A,\subseteq)$ and $(\B, \subseteq)$.   \end{theorem}

Given event structures $A$ and $B$, we define 
$$[A\to B] \eqdef \Pr([\iconf A \to \iconf B])\,.$$
The configurations of  $[A\to B]$ under inclusion are isomorphic to the stable function space of dI-domains $(\iconf A, \subseteq)$ and $(\iconf B, \subseteq)$.

\subsection{Dependent-type constructions
}\label{app:deptypes}

We base the constructions here on~\cite{dIPoly, CGW} (though with the simplification w.l.o.g.~of using the substructure relation $\tri$ between stable families~\cite{icalp82,evstrs} in place of rigid embeddings between dI-domains).  Recall the definition $\A\tri \B$, where   $\A$ and $\B$ are  stable families with events $A$ and $B$ respectively:
$$
\eqalign{
\A \tri \B \hbox{ iff }
&A\subseteq B\, \hbox{ and } \cr
&\all x.\ x\in \A \iff x\subseteq A \ \& \  x\in \B\,.
}
$$
The relation  $\A\tri \B$ specifies a rigid embedding from the dI-domain $(\A,\subseteq)$  to the dI-domain $(\B, \subseteq)$ with projection $y\mapsto y\cap A$ from 
 $\B$ to $\A$.\footnote{The relation $\tri$ does not have least upper bounds in general; 
 there can be distinct minimal upper bounds.}
 
 For event structures $A$ and $B$, we write $A \tri B$ when $\iconf A \tri \iconf B$.  On event structures,  $A \tri B$ is equivalent to the events of $A$ being included in those of $B$ with 
 $$
 \all a\in A.\ [a]_A = [a]_B \ \hbox{ and }\  \all X\subseteq A.\ X\in\Con_A \iff X\in\Con_B\,.
 $$

Let $\A$ be a stable family.
Let $\B(\_)
$ be a {\em stable functor} from the  
partial-order category 
$(\A, \subseteq)$ to the (large) partial-order category 
of stable families related by 
$\tri$.  We shall write $B(x)$ for the events of $\B(x)$, where $x\in\A$. 
That the functor is {\em stable} means it is continuous and preserves pullbacks, which in this case means it is 
a function which  preserves least upper bounds of directed sets and greatest lower bounds of compatible pairs.
Correspondingly, for an event structure $A$,  a functor from $x\in \iconf A$ to event structures $B(x)$ is {\em stable} when it is continuous and preserves pullbacks w.r.t.~$\tri$ on event structures. \\
 
\noindent
{\bf Dependent sum }\\
$\Sigma_{x\in \A}\, \B(x)$ is    the stable family 
$$
 \set{
x\vvbar y
}{
x\in\A \ \& \ y\in \B(x)}
\,. $$
\begin{proposition}
$\Sigma_{x\in \A}\, \B(x)$ is a stable family with configurations corresponding to pairs $(x, y)$, where $x\in\A$ and $y\in \B(x)$; the order of configurations corresponds to the coordinatewise order on pairs.  
\end{proposition}

We shall  describe a typical configuration of  $\Sigma_{x\in\A}\,  \B(x)$ as a pair $(x, y)$  where $x\in\A$ and $y\in \B(x)$.  It's often convenient to describe an operation on configurations of the dependent sum in terms of their decomposition into pairs. 

For an event  structure  $A$ and $B(x)$, stable  in $x\in\iconf A$,
$$\Sigma_{x:A}\, B(x)\eqdef \Pr(\Sigma_{x\in \iconf A}\, \iconf{B(x)})\,.
$$

By analysing the structure of the prime configurations of $\Sigma_{x\in \A}\, \B(x)$, we can see that 
the event structure $\Sigma_{x:A}.\, B(x)$ is isomorphic to the event structure comprising
\begin{itemize}
\item
{\em events}, consisting of the set of $a\in A$ in disjoint union with the set of pairs $(x,b)$, where $x\in\conf A$ is a smallest  configuration for which $b\in B(x)$; 
\item
{\em causal dependency},  that generated by the relations on  events 
$$
\eqalign{
&
a'\leq_A a\,, \cr
&
(x', b')\leq (x,b) \ \hbox{ if } \ x'\subseteq x \ \&\ b'\leq_B b\,, \hbox{ and }\cr
&
a\leq (x,b) \ \hbox{ if }\  a\in x\,;
}
$$
\item
{\em consistency}, a finite subset of events, 
$$
\eqalign{
&\set{a_i}{i\in I} \cup\set{(x_j, b_j)}{j\in J} \in \Con
\hbox{ iff } \cr
&\set{a_i}{i\in I} \in\Con_A\  \&\  \cr\
& \bigcup_{j\in J} x_j \in \conf A
\ \&\ 
\set{b_j}{j\in J} \in\Con_B 
\ \&\ \cr
&\all j, k\in J.\ 
b_j=b_k \implies x_j= x_k\,.
}
$$
\end{itemize}

\vskip 0.1in 

 \noindent
{\bf Dependent product}\\
The obvious projection from $\Sigma_{x\in \A}\, \B(x)$ to $\A$ is a simple form of Grothendieck fibration.  We obtain  $\Pi_{x\in\A}\,  \B(x)$  as a stable family whose configurations correspond to 
stable sections of the fibration, \ie~stable functions from $(\A, \subseteq)$ to $(\Sigma_{x\in \A}\, \B(x),\subseteq)$ which send $x\in\A$ 
to a configuration   $x\vvbar y$ where $y\in\B(x)$.  To this purpose, we can refashion the construction of the stable function space 
of Section~\ref{sec:Stfnspace} to restrict to stable functions which are sections.

$\Pi_{x\in\A}\,  \B(x)$ is the stable  
family  comprising those sets
$$
f \subseteq \set{(x,b)}{x\in \A^o \ \& \ b\in B(x)}$$ for which, for all $x\in \A$,
\begin{itemize}
\item
 $\set{b}{\exists x'\subseteq x.\ (x',b) \in f} \in \B(x)$ \ and 
\item
if $(x',b), (x'', b)\in f$ with $x', x''\subseteq x$  then $x'=x''$.
\end{itemize}
When $\B(x)$ is constantly $\B$, for all $x\in\A$, we observe that $$\Pi_{x\in\A}\,  \B(x)= [\A\to \B]\,.$$
 
\begin{theorem}
The configurations of $\Pi_{x\in\A}\,  \B(x)$ 
correspond to 
stable sections of $\Sigma_{x\in \A}\, \B(x)$; inclusion between configurations corresponds to the stable order on sections.  
\end{theorem}
 
Hence we can describe a typical configuration of  $\Pi_{x\in\A}\,  \B(x)$ as a stable section,  using $\lambda$-notation, as $\lambda x\in\A.\ f(x)$, provided $f(x)\in \B(x)$ is stable in $x\in\A$; we 
obtain its components by function application. 

For event  structures  $A$ and $B(x)$, stable in $x\in\iconf A$,
define 
$$
\Pi_{x:A}\, B(x) \eqdef 
\Pr(\Pi_{x\in\iconf A}\,  \iconf{B(x)})\,.
$$


\section{Stable spans}\label{app:stspans}
  

Stable spans are monoidal-closed---see~\cite{mikkel}\S7.5.  Their tensor is given by the simple parallel composition of event structures. We define the function space in slightly greater generality, between stable families.
    
Let $\A$ and $\B$ be stable families.  
We construct the function space 
of stable spans as a stable family.
The stable family $[\A\multimap  \B]$ comprises those
$
F\subseteq \A^o \times B$ for which 
\begin{itemize}
\item
$\bigcup\set{x}{\exists b.\ (x,b)\in F} \in \A$,
 \item
$\all x \in \A.\ \set{b}{\exists x'\subseteq x.\ (x',b)\in F} \in \B$ and 
\item
$\all (x, b), (x', b)\in F.\  x=x'$.
\end{itemize}
It can be checked that $[\A\multimap  \B]$ is a stable family. For event structures $A$ and $B$, define $[A\multimap B]\eqdef \Pr([{\iconf A \multimap \iconf B}])$.  The configurations of $[A\multimap  B]$ represent the possible paths the computation  of output in $B$ from input in $A$  can follow. 

 By broadening to nondeterministic computation 
 we can often regard types as special maps.  For example
 $[A\multimap B]$ becomes a stable span with the obvious demand and rigid map.  As such it is
 terminal within all stable spans from $A$ to $B$: for any span $S, d,r$ there is a unique 2-cell as shown
 $$\small\xymatrix@R=12pt@C=10pt{
 &\ar[ld]_dS\ar@{..>}[d]\ar[rd]^r&\\
 A&\ar[l][A\multimap B]\ar[r]&B\,.
 }
 $$
 Deterministic stable spans coincide with stable functions. 
 
\subsection{Dependent product for stable spans}
The dependent product for stable spans is a refashioning of the definition of their function space, to take account of the dependency of $\B(x)$ on $x\in \A$. The stable sections of the previous dependent product above are replaced by stable spans, so giving a form of nondeterministic dependent product.


$\sPi_{x\in\A}\,  \B(x)$ is the stable  
family  comprising those sets 
$$
F\subseteq \set{(x,b)}{x\in \A^o \ \& \ b\in B(x)}$$ for which 
\begin{itemize}
\item
$\bigcup\set{x}{\exists b.\ (x,b)\in F} \in \A$,
 \item
$\all x \in \A.\ \set{b}{\exists x'\subseteq x.\ (x',b)\in F} \in \B(x)$ and 
\item
$\all (x, b), (x' b)\in F.\  x=x'$.
\end{itemize}
When $\B(x)$ is constantly $\B$, for all $x\in\A$, we observe that $$\sPi_{x\in\A}\,  \B(x)= [\A\multimap \B]\,.$$

\begin{proposition}
The configurations of $\sPi_{x\in\A}\,  \B(x)$ correspond to stable sections $f: X_0 \to \Sigma_{x\in X_0} \B(x)$, where 
$X_0=\set{x\in \A}{x\subseteq x_0}$ for some $x_0\in\A$, and  for all $x\in X_0$, writing 
$f(x)= (x, f'(x))$ and $f(x_0)= (x_0, f'(x_0))$,  
if $f'(x) =f'(x_0)$ then $x=x_0$.  
\end{proposition}

For event  structures  $A$ and $B(x)$, stable in $x\in\iconf A$,
define 
$$
\sPi_{x:A}\, B(x) \eqdef 
\Pr(\sPi_{x\in\iconf A}\,  \iconf{B(x)})\,.
$$

 \section{Proofs for Section~\ref{sec:part1}}
\begin{theorem}\label{appthm:affstmaptostrats}
Let $f: \iconf A \to \iconf B$ be an affine-stable map between \esswp~$A$ and $B$.    
Then 
$$
\Fam \eqdef \set{
x\vvbar y\in \iconf{A^\perp\vvbar B}}{y\bel_B f(x)}
$$
is an infinitary stable family.  
The map $\max:\Pr(\Fam)\to A^\perp\vvbar B$ is a strategy $f_!:A\profto B$.  The strategy $f_!$ is deterministic if  $A$ and $B$ are race-free and $f$ reflects $-$-compatibility, \ie~$x\subseteq^- x_1$ and  $x\subseteq^-  x_2$ in\,   $\iconf A$
 and
 $f x_1 \cup f x_2\in \iconf B$ implies  
 $x_1\cup x_2\in \iconf A$.  
\end{theorem}
\begin{proof}
In the proof we make frequent use of the following observations.
 Let $B$ be an \eswp.
Let 
$y_i \bel_B y_i'$, for all $i\in I$.  Then, (with $I$ nonempty), 
$$
\bigcap_{i\in I} y_i \bel_B \bigcap_{i\in I} y'_i 
\,.$$
When   both $\set{y_i}{i\in I}$ and $\set{y'_i}{i\in I}$ are compatible in $\iconf B$, 
$$
\bigcup_{i\in I} y_i \bel_B \bigcup_{i\in I} y'_i\,. 
$$

We first show $\Fam$ is a stable family.

\noindent
{\em Completeness:}  Let $\set{x_i\vvbar y_i}{i\in I}$ be a finitely compatible subset in $\Fam$.  From compatibility, it follows that $\bigcup_{i\in I} x_i$ and $\bigcup_{i\in I} y_i$ are configurations.  By assumption $y_i\bel_B f(x_i)$, for all $i\in I$, so
$$
\bigcup_{i\in I} y_i \bel_B  \bigcup_{i\in I} f (x_i) \subseteq^+  f(\bigcup_{i\in I} x_i)\,.
$$
As the relation $\subseteq^+$ is included in $\bel_B$, by the latter's transitivity we obtain
$$
\bigcup_{i\in I} y_i \bel_B   f(\bigcup_{i\in I} x_i)\,, 
$$
so
$$
\bigcup_{i\in I} (x_i\vvbar y_i) = (\bigcup_{i\in I} x_i\vvbar \bigcup_{i\in I} y_i)  \in\Fam\,.
$$

\noindent
{\em Stability:}  Let $\set{x_i\vvbar y_i}{i\in I}$ be a nonempty compatible subset in $\Fam$. 
By assumption $y_i\bel_B f(x_i)$, for all $i\in I$, so
$$
\bigcap_{i\in I} y_i \bel_B  \bigcap_{i\in I} f (x_i) \supseteq^- f(\bigcap_{i\in I} x_i)\,
$$
---it follows from the assumptions that  $\set{x_i }{i\in I}$ is a nonempty compatible family in $\iconf A$, as is required to apply the stability of $f$.   As   $\supseteq^-$ is included in $\bel_B$, we deduce
$$
\bigcap_{i\in I} (x_i\vvbar y_i)  = (\bigcap_{i\in I} x_i\vvbar  \bigcap_{i\in I}  y_i) \in\Fam\,.
$$

\noindent
{\em Finiteness:} If $x\vvbar y$ in the family $\Fam$, then $x\in\iconf A$ and $y\in\iconf B$ with $y\bel_B f(x)$.  An   element in $x\vvbar y$ is either $(1,a)$ where $a\in x$ or $(2,b)$ where $b\in y$.  We analyse these two cases. 

\noindent
{\em Case $a\in x$.}  Observe the set 
$f([a])^-$ is finite by $-$-image finiteness.  It follows that $[f([a])^-]\in\conf B$ is a finite configuration of $B$ for which
$$
[f([a])^-] \subseteq^+ f[a]\,, \hbox{ so } [f([a])^-] \bel_B f[a]\,.
$$
As also $y\bel_B f(x)$ we have
$$
y\cap  [f([a])^-] \bel_B f(x) \cap f[a] = f[a]\,,
$$
whence  
$$
[a] \vvbar (y\cap [f([a])^-]) \in\Fam 
$$
creating a finite subconfiguration of $x\vvbar y$ containing $(1,a)$.   

\noindent
{\em Case $b\in y$.}  We prove a stronger result than is strictly needed for this part of the proof, in preparation for the proof of coincidence-freeness later.  
Letting $b\in y$, take
$$
x_0 \eqdef \bigcap\set{x'\in\iconf A }{[b]^+\subseteq f(x')\ \&\ x' \subseteq x}\,.
$$
By the stability of $f$,  
$$
f(x_0) \subseteq^- \bigcap \set{f(x')}{x'\in\iconf A \ \&\ [b]^+\subseteq f(x')\ \&\ x' \subseteq x}\,.
$$
Thus
$$
[b]^+\subseteq f(x_0)\,,
$$
and $x_0$ is the minimum subconfiguration of $x$ for which $[b]^+\subseteq f(x_0)$.   By +-continuity, $x_0$ is a finite configuration.  Also
$$
[f (x_0)^-] \subseteq^+ f(x_0)
$$
where the configuration $[f (x_0)^-]$ is also finite by $-$-image finiteness. 
We  observe that all the $\leq$-maximal events in $x_0$ are +ve:  supposing  otherwise, there is  a $\leq$-maximal $-$ve event in $x_0$ so  a configuration $x_0' \subsetneq^- x_0$; then, as $f$ preserves polarity,  $[b]^+\subseteq f(x_0)\subseteq^- f(x_0')$ so $[b]^+\subseteq f(x_0')$, contradicting the minimality of $x_0$.
Whatever the polarity of $b$ we obtain
$$
 [f (x_0)^-] \cup [b] \supseteq^-  [f (x_0)^-] \cup [[b]^+] \subseteq^+ f(x_0)\,,
$$
so
$$
 [f (x_0)^-] \cup [b] \bel_B  f(x_0)\,.
$$
We now show that $b\notin [f (x_0)^-]$ by cases on the polarity of $b$.

Suppose $\pol_b(b) = +$. In this case $[b] =  [[b]^+]$ and $x_0$ is the minimum subconfiguration of $x$ such that $b\in f(x_0)$.
If $x_0 =\emptyset$, by affinity, in the case of the empty family, we have $\emptyset \subseteq^+ f(\emptyset)$ which ensures $[f(x_0)^-]$ is empty, so does not contain $b$.  
Otherwise, the $\leq$-maximal events in $x_0$ are +ve and there is a subconfiguration $x_0' \subsetneq^+ x_0$.  As $f$ respects polarity, $f(x_0') \subseteq^+ f(x_0)$. Hence $f(x_0)^- \subseteq f(x_0')$ so $[f (x_0)^-]  \subseteq^+ f(x_0')$.  
From the minimality of $x_0$, we must have $b\notin f(x_0')$, 
so we also have $b\notin [f (x_0)^-]$, as required.

Suppose $\pol_B(b) = -$.
We show $b\notin f(x_0)$, from which  $b\notin [f (x_0)^-]$ follows directly.  Suppose otherwise that $b\in f(x_0)$. If $x_0$ is empty,   we have $\emptyset \subseteq^+ f(\emptyset) = f(x_0)$, contradicting the polarity of $b$.
When $x_0$ is nonempty,  
as the $\leq$-maximal events in $x_0$ are +ve, we must have a strictly smaller subconfiguration $x_0' \subsetneq^+ x_0$.  But then as $f$ respects polarity  $f(x_0') \subseteq^+ f(x_0)$. As $b$ is $-$ve, $b\in f(x_0')$ making $[b]^+ \subseteq f(x_0')$,which  contradicts the minimality of $x_0$.  This shows $b\notin f(x_0)$, as required to obtain $b\notin [f (x_0)^-]$.

To complete the proof of the finiteness property, observe that 
$y\bel_B f(x)$ with $[f (x_0)^-] \cup [b] \bel_B  f(x_0)$ entail
$$
y\cap  ([f (x_0)^-] \cup [b])  \bel_B f(x) \cap f(x_0) = f(x_0)\,.
$$
It follows that 
$$
x_0\vvbar (y\cap ([f (x_0)^-] \cup [b]))  \in \Fam\,,
$$
so yielding a finite subconfiguration of $x\vvbar y$ containing $(2,b)$. We note for later that
$x_0$ is the minimum subconfiguration of $x$ for which $[b]^+ \subseteq f(x_0)$  and  from this it follows  that 
$$
 b\notin [f (x_0)^-] \  \hbox{ with }\ 
[f (x_0)^-] \cup [b] \bel_B f(x_0) \,.
$$

\noindent
{\em Coincidence-free:}  Let $x\vvbar y\in \Fam$.  Consider two distinct events in $x\vvbar y$.  There are three cases: they belong to the same component $x$;  they belong to the same component $y$; or  they belong to different components.

If they both belong to the same $x$-component, from the argument above they are $(1,a_1)$ and $(1,a_2)$ and belong to the respective subconfigurations
$$
[a_1] \vvbar (y\cap [f([a_1])^-])   \  \hbox{ and } \ [a_2] \vvbar (y\cap [f([a_2])^-])  
$$
of $x\vvbar y$. If $a_1$ and $a_2$ are distinct, one of the subconfigurations must separate them in the sense of containing one but not the other.  

Assume they both belong to the same $y$-component,   one being $(2,b_1)$ and the other $(2,b_2)$, with $b_1, b_2\in y$.  From the 
proof of the finiteness part above, they belong to respective subconfigurations of $x\vvbar y$ of the form 
$$
x_1\vvbar (y\cap ([f (x_1)^-] \cup [b_1])) 
 \  \hbox{ and } \ 
x_2\vvbar (y\cap ([f (x_2)^-] \cup [b_2])) 
$$
where $x_1$ is the minimum subconfiguration of $x$ for which $[b_1]^+ \subseteq f(x_1)$ and  $x_2$ is the minimum subconfiguration of $x$ for which $[b_2]^+ \subseteq f(x_2)$.  Recall from earlier that 
$$
\eqalign{
& b_1\notin [f (x_1)^-] \  \hbox{ with }\ 
[f (x_1)^-] \cup [b_1] \bel_B f(x_1) \quad \hbox{  and }\cr
 &
b_2\notin [f (x_2)^-]  \  \hbox{ with }\ 
[f (x_2)^-] \cup [b_2] \bel_B f(x_2) \,.}
$$
Imagine the two subconfigurations  of $x\vvbar y$ above do not separate $(2,b_1)$ and  $(2,b_2)$, \ie 
$$
\eqalign{
&(2, b_2) \in x_1\vvbar (y\cap ([f (x_1)^-] \cup [b_1])) 
\quad \hbox{
and 
}\cr
&(2, b_1)\in x_2\vvbar (y\cap ([f (x_2)^-] \cup [b_2]))\,.}
$$
Then 
$$
\eqalign{
&b_2 \in [f (x_1)^-] \cup [b_1] \bel_B f(x_1)
\quad 
\hbox{ and }\cr
&
b_1\in[f (x_2)^-] \cup [b_2]\bel_B f(x_2)\,. 
}
$$
By the properties of $\bel_B$, 
we see that 
$
[b_2]^+ \subseteq f(x_1)
$
and 
$
[b_1]^+ \subseteq f(x_2)
$.
From the minimality properties of $x_1$ and $x_2$ we deduce that  $x_1= x_2$.  Writing $x_0 \eqdef x_1= x_2$ and 
recalling 
$
b_1, b_2 \notin [f (x_0)^-]
$
we obtain $b_1 \in [b_2]$ and $b_2 \in [b_1]$, so $b_1=b_2$.  
Hence distinct $(2,b_1)$ and $(2,b_2)$ are separated by the chosen subconfigurations of $x\vvbar y$.

Assume the  two distinct events in $x\vvbar y$ belong to different components, one being $(1,a)$, with $a\in x$, and the other $(2,b)$, with $b\in y$.  If $b\notin f([a])$ then one argues, as frequently above, that  $f([a])\bel_B f([a])$ together with $y\bel_B f(x)$ gives $y\cap f([a]) \bel_B  f([a])$ yielding $[a]\vvbar(y\cap  f([a]))$  a subconfiguration of $x\vvbar y$, which moreover contains $(1,a)$ but not $(2,b)$.  Thus suppose $b\in  f([a])$.  If  $b\in  f([a))$ then $[a)\vvbar(y\cap  f([a)))$ is a subconfiguration of $x\vvbar y$ which contains $(2,b)$ but not $(1,a)$.  The remaining case is when  $b\in  f([a])$ and  $b\notin  f([a))$.  Then
$[a)\longcov a [a]$ and 
$b\in  f([a])\setdif  f([a))$. 

If $\pol_A(a) = +$ then, as $f$ respects polarity, 
$$
f([a)) \subseteq^+ f([a]), \hbox{ so }  f([a)) \bel_B f([a])\,.
$$
 By the now familiar argument, this yields $[a]\vvbar (y\cap f[a))$ a subconfiguration of $x\vvbar y$ containing $(1,a)$ but not $(2,b)$.  

Similarly, if $\pol_A(a) = -$ then 
$$
f([a)) \subseteq^- f([a]), \hbox{ so }  f([a]) \bel_B f([a))\,,
$$
yielding a subconfiguration $[a)\vvbar (y\cap f[a])$ of $x\vvbar y$  which contains $(2,b)$ but not $(1,a)$.  

This completes  the proof of coincidence-freeness.
\\

We check the map $\max:\Pr(\Fam)\to A^\perp\vvbar B$ is a strategy. Observe
 that  
$$
x'\sqsupseteq_A x \ \& \  x\vvbar y\in \Fam  \ \&\  y\sqsupseteq_B y' \implies  x'\vvbar y'\in \Fam
$$
as the l.h.s.~clearly entails  
$$
y'\bel_B y\bel_B f(x)  \sqsubseteq_B f(x') \,, 
$$
so the r.h.s..  In particular, when $x\vvbar y\in \Fam$ and $(x'\vvbar y') \in \iconf{A^\perp\vvbar B}$,

if $(x\vvbar y)\subseteq^- (x'\vvbar y')$, then $(x'\vvbar y') \in\Fam$; and 

if $(x'\vvbar y')\subseteq^+ (x\vvbar y) $, then $(x'\vvbar y') \in\Fam$.

\noindent
Thus the composite map 
$$
\iconf{\Pr(\Fam)}\to \Fam \hookrightarrow \iconf{A^\perp\vvbar B}\,
$$
of stable families, where the first map is $\max$ and the second is an inclusion, satisfies the ``lifting'' 
conditions needed of  a strategy---see~\cite{ecsym-notes},
 ensuring that $\max:\Pr(\Fam)\to A^\perp\vvbar B$ is a strategy. \\

Assume now that $A$ and $B$ are race-free and that $f$ reflects $-$-compatibility.   As $A^\perp\vvbar B$ is now also race-free, to show $f_!$ a deterministic strategy it suffices to show that any two +ve event increments of a configuration in $\Fam$ are compatible in $\Fam$, \ie~ if $x\vvbar y \cov^+ x_1\vvbar y_1$ and $x\vvbar y \cov^+ x_2\vvbar y_2$ in $\Fam$, then $(x_1\cup x_2)\vvbar (y_1\cup y_2)\in\Fam$.  Consider cases. \\
 If the increments are $y\longcov{b_1} y_1$ and $y\longcov{b_2} y_2$, then $b_1$ and $b_2$ are +ve in $B$.  Because each $y_i\bel_B f(x)$, \ie~$y_i\supseteq^- z \subseteq^+ f(x)$ where $z=y \cap f(x)$, we see both   $b_1\in f(x)$ and $b_2\in f(x)$.  Hence $z\cup\setof{b_1, b_2}\in\iconf B$.  Because $B$ is race-free we obtain $y_1\cup y_2\in\iconf B$.  Checking $y_1\cup y_2\bel_B f(x)$, ensures $x\vvbar (y_1\cup y_2) \in \Fam$.\\
   If the increments are $x\longcov{a_1} x_1$ and $x\longcov{a_2} x_2$ then $a_1$ and $a_2$ are $-$ve in $A$ with $y\bel_B f(x_1)$ and 
$y\bel_B f(x_2)$.  It follows that each $f(x_i)\setdif f(x)$ consists of solely $-$ve events in $B$ and so are included in $y$.  This ensures the compatibility of $f(x_1)$ and $f(x_2)$.  That $(x_1\cup x_2)\vvbar y  \in \Fam$ now follows from $f$ reflecting $-$-compatibility and its  affinity. \\
The final case is when the increments are, w.l.o.g.~$x\longcov{a_1} x_1$ and $y\longcov{b_2} y_2$, when $a_1$ is $-$ve in $A$ and $b_2$ +ve in $B$.  Then $y\bel_B f(x_1)$ and $y_2\bel_B f(x)$, so $y_2\bel_B f(x_1)$, making $x_1\vvbar y_2\in\Fam$.  
\end{proof}

\subsection{A functor
}

Let $f:A\to B$ and $g:B\to C$ be affine stable maps. They determine stable families
$$
\eqalign{
{\F} = &\set{x\vvbar y}{ f(x)\leb_B y} \hbox{ and } \cr
{\G} = &\set{y\vvbar z}{ g(y)\leb_C z}\,,}
$$
respectively. 
Consider the stable family determined by the composition of functions $gf$, \viz
$$
\set{x\vvbar z
}{ gf(x)\leb_C z}\,.
$$
One can show straightforwardly that  
$$
\eqalign{
\set{x\vvbar z
}{ gf(x)\leb_C z}
=  &\set{x\vvbar z
}{ \exists y\in\iconf B.\ f(x)\leb_B y \ \&\  g(y)\leb_C z}\cr
&\set{x\vvbar z
}{ \exists y\in\iconf B.\ x\vvbar y\in {\F}  \ \&\ y\vvbar z\in{\G} }\cr
 = &\ {\G}\circ {\F}\,,
}
$$
where the last composition is essentially the composition of stable families as relations: for instance, regarding the stable family $\F$ 
as $$\set{(x,y)\in\iconf{A
}\times\iconf B}{f(x)\leb_B y}\,,$$ observing the isomorphism $\iconf{A
}\times\,\iconf B\iso \iconf{A^\perp\vvbar B}$.
We shall show that 
$$
\Pr({\G}) \scirc \Pr({\F}) \iso \Pr({\G}\circ {\F})\,,
$$
so reducing the composition of strategies of affine-stable maps to relational composition;  
by definition, it  follows directly that 
$$
g_! \scirc f_! \iso (gf)_! \,.
$$
For functoriality of $(\_)_!$ we also require preservation of identities.  However, the stable family determined by $\id_A:\iconf{A}\to \iconf{A}$ is, by definition, 
$$
\set{x\vvbar y}{x\leb_A y} =\iconf{\CC_A}\,,
$$
ensuring that ${\id_A}_! \iso \CC_A$.

\begin{lemma}\label{lem:partlrigid-y*x}
Let $\sig:A\profto B$ and $\tau:B\profto C$ be strategies.  Suppose $\tau_1$ is partial rigid (\ie, the component $\tau_1: T\to B$ preserves causal dependency when defined).  Letting $x\in\conf S$, $y\in\conf T$,
$$
y\sncirc x \hbox{  is defined } \hbox{  iff } \sig_2 x = \tau_1 y\,.
$$
\end{lemma}
\begin{proof}
Write $x_A = \sig_1 x$, $x_B = \sig_2 x$, $y_B = \tau_1 y$ and $y_C= \tau_2 y$.  
Recall $y\sncirc x $ is defined to be the bijection  
$$
x\vvbar y_C\iso x_A\vvbar x_B\vvbar x_C \iso x_A\vvbar y
$$
induced by $\sig$ and $\tau$ provided 
 $x_B = y_B$, \ie~$ \sig_2 x = \tau_1 y$, and the bijection is secured---see Theorem~\ref{thm:securedbijns}.
 To simplify notation we can present the bijection
 as $x\cup y$ in which we identify the two sets $x$ and $y$ at their parts $\sig^{-1} x_B$ and $\tau^{-1} y_B$ via the common image $x_B = y_B$.  
 
To obtain a contradiction, suppose that the bijection were not secured, that there were a causal loop in $x\cup y$, \ie~that there were a chain 
$$
u_1\imc u_2 \imc \cdots \imc u_n = u_1
$$
of events in $x\cup y$, with $n>1$, w.r.t.~causal dependency $\imc$ which is either  $\imc_{S}$ or $\imc_{T}$.  The events of  $x\sncirc y$ and so of the chain are either over $A$, $B$ or $C$.  As there are no causal loops in $S$ or $T$ the causal loop must contain events over each of $A$, $B$ and $C$.  W.l.o.g., we may assume $u_1$ is over $B$. 

Part of the chain is over $C$.  The whole chain has the form
$$
u_1\imc \cdots \imc u_{i-1} \imc_T u_i \imc_T \cdots \imc_T u_j \imc_T u_{j+1} \imc \cdots   \imc u_n = u_1
$$
where $u_{i-1}$ and $u_{j+1} $ are over $B$ and $u_i, \cdots, u_j$ are all over $C$.
Clearly $u_{i-1} <_T u_{j+1}$.   As $\tau_1$ is partial rigid, we obtain
$\tau(u_{i-1}) <_B \tau(u_{j+1})$.   With the identification of events over $B$ in $x$ and $y$, we have $\sig(u_{i-1}) <_B \sig(u_{j+1})$. As $\sig$  locally reflects causal dependency, we see that $u_{i-1} <_S u_{j+1}$.  
We now have a causal loop
$$
u_1\imc \cdots\imc  u_{i-1} <_S u_{j+1} \imc \cdots  \imc u_n = u_1
$$
from which the events $u_i, \cdots, u_j$ over $C$ have been excised. 
Continuing in this way we can remove all events over $C$ from the causal loop,  obtaining a causal loop in $S$ ---a contradiction. 
\end{proof}

Now to the isomorphism. 
First, a key observation, expressing that the strategy obtained from an affine-stable map doesn't disturb the causality of input: 
\begin{proposition}\label{prop:obs}
Let $g:B\to C$ be an affine-stable map which determines the  stable family
${\G} = \set{y\vvbar z}{ g(y)\leb_C z}$.  Let $y\vvbar z \in\G$.  Then,
$$
\all b, b'\in y.\ (1,b')\leq_{y\vvbar z} (1,b)  \iff b'\leq_B b\,.
$$
In the strategy $g_! =\max:\Pr({\G})\to B^\perp\vvbar C$, the component $(g_!)_1:\Pr({\G})\to B^\perp$ is partial rigid.
\end{proposition}
\begin{proof}
Recall $(1,b')\leq_{y\vvbar z} (1,b)$ iff every subconfiguration of $y\vvbar z$  in $\G$ which contains $(1,b)$ also contains $(1,b')$.  

Any subconfiguration of $y\vvbar z$ necessarily takes the form $y'\vvbar z'$ where $y'$ is a subconfiguration of $y$ in $B$ and $z'$ is a subconfiguration of $z$ in $C$ with $g(y')\leb_B z'$. From $b'\leq_B b$ it  therefore follows that $(1,b')\leq_{y\vvbar z} (1,b)$.  

Conversely, given a subconfiguration $y'$ of $y$ we have $y'\vvbar g(y')\in \G$ whence 
$y'\vvbar g(y') \cap z'$ is a subconfiguration of $y\vvbar z$ in $\G$.  From this the converse implication follows: if $(1,b')\leq_{y\vvbar z} (1,b)$ then $b'\leq_B b$.

Thus $(1,b')\leq_{y\vvbar z} (1,b)$  iff $b'\leq_B b$, for all $b, b'\in y$. 
That $(g_!)_1$ is partial rigid is a direct consequence. 
\end{proof}
  
\begin{lemma}\label{lem:affstcompeqrelcomp}
Let $f:A\to B$ and $g:B\to C$ be affine stable maps which determine stable families
${\F} = \set{x\vvbar y}{ f(x)\leb_B y} $  and 
${\G} = \set{y\vvbar z}{ g(y)\leb_C z}$, 
respectively.  
Then,  $\Pr({\G}) \scirc \Pr({\F}) \iso \Pr({\G}\circ {\F})$.
\end{lemma}
\begin{proof}
Recall, $\Pr({\G}) \scirc \Pr({\F})$ is obtained as $\Pr({\G} \sncirc {\F})$ followed by hiding the synchronisations over $B$.  
First consider ${\G} \sncirc {\F}$.   

A finite configuration of ${\G} \sncirc {\F}$, built as a pullback of stable families, has the form
$x\vvbar y\vvbar z$ where
$x\vvbar y \in \F$ and $y\vvbar z\in \G$  and the 
causal dependencies from $\F$ and $\G$ do not jointly introduce any causal loops.
However, from the observation of Proposition~\ref{prop:obs} and Lemma~\ref{lem:partlrigid-y*x} above,  it follows that there are no causal loops for such particular stable families.

It follows that for all $x\vvbar y \in \F$ and $y\vvbar z\in \G$ we have $x\vvbar y\vvbar z$ is a configuration of ${\G} \sncirc {\F}$.  
Thus we have a simple characterisation of the the stable family ${\G} \sncirc {\F}$:
 $${\G} \sncirc {\F} = \set{x\vvbar y\vvbar z\in\iconf{A^\perp\vvbar B\vvbar C}  }{\ 
x\vvbar y \in {\F }\ \&\ y\vvbar z\in {\G}}\,.$$

It remains to consider the effect of hiding the synchronisations over $B$ and show
$$
\Pr({\G}) \scirc \Pr({\F}) \iso \Pr({\G}\circ {\F})\,,
$$
where 
$$
{\G}\circ {\F} = 
\set{x\vvbar z\in\iconf{A^\perp\vvbar C}
}{ \exists y\in\iconf B.\ x\vvbar y\in {\F}  \ \&\ y\vvbar z\in{\G} }\,.
$$ 
(As we saw in the discussion preceding this lemma, this is the  stable family obtained from the composition $gf$.)  
To this end we define
$$
\theta: \Pr({\G}) \scirc \Pr({\F}) \to \Pr({\G}\circ {\F}) 
$$
and its putative mutual inverse 
$$
\phi:  \Pr({\G}\circ {\F}) \to \Pr({\G}) \scirc \Pr({\F})\,.
$$
{\it For simplicity of notation, to avoid indices, throughout this proof assume that the events $A$, $B$ and $C$ are pairwise disjoint and identify $x\vvbar y\vvbar z$ with $x\cup y\cup z$.}\\

The events of $\Pr({\G}) \scirc \Pr({\F})$ have the form
$[a]_{x\vvbar y\vvbar z}$, where $a\in x$,  or  $[c]_{x\vvbar y\vvbar z}$, where $c\in z$, and $x\vvbar y\vvbar z\in {\G} \sncirc {\F}$.  
The events of $\Pr({\G}\circ {\F})$ have the form
$[a]_{x\vvbar  z}$, where $a\in x$,  or  $[c]_{x\vvbar z}$, where $c\in z$, and $x\vvbar z\in {\G} \circ {\F}$.  Define
$$
\theta([d]_{x\vvbar y\vvbar z}) = [d]_{x\vvbar z}\  \hbox{ and }\
\phi([d]_{x\vvbar z}) = [d]_{x\vvbar f(x) \vvbar z}\,,
$$
on typical events $[d]_{x\vvbar y\vvbar z}\in \Pr({\G}\circ {\F}) $ and $[d]_{x\vvbar z}\in  \Pr({\G}\circ {\F})$. We should check $\theta$ and $\phi$ are well-defined functions.   
 In showing that 
 $\theta$ is well-defined we use that $x\vvbar y\vvbar z$ is a configuration of ${\G} \sncirc {\F}$ directly  implies 
 $x\vvbar z$ is a configuration of ${\G} \circ {\F}$.  In showing  $\phi$ is well-defined we need that 
 $x\vvbar z\in {\G} \circ {\F}$ implies $x\vvbar f(x) \vvbar z \in {\G} \sncirc {\F}$.  Assuming $x\vvbar z\in {\G} \circ {\F}$, 
  we have  $x\vvbar y\in {\F}$ and $y\vvbar z\in{\G}$ for some $y\in\iconf B$.  Then 
$f(x)\leb_B y$ and    $g(y) \leb_C z$.  Thus 
$gf(x)\leb_C g(y) \leb_C z$ whence $g(f(x))\leb_C z$ ensuring $f(x)\vvbar z\in \G$.
Clearly  $x\vvbar f(x) \in \F$, so $x\vvbar f(x) \vvbar z \in {\G} \sncirc {\F}$, as needed. 

We show $\theta$ and $\phi$ are mutual inverses.  
It is easy to see that $\theta\phi([d]_{x\vvbar z}) = [d]_{x\vvbar z}$.  By definition, 
$\phi\theta([d]_{x\vvbar y\vvbar z}) = [d]_{x\vvbar f(x) \vvbar z}$, where $x\vvbar y\vvbar z\in {\G} \sncirc {\F}$ and $d$ is an event of $x$ or $z$.  We require 
$$
[d]_{x\vvbar y \vvbar z} = [d]_{x\vvbar f(x) \vvbar z}\,.
$$
To this end we show $x\vvbar (y\cap f(x)) \vvbar z \in {\G} \sncirc {\F}$; once this is shown we have
$$
[d]_{x\vvbar y \vvbar z} = [d]_{x\vvbar (y\cap f(x)) \vvbar z} = [d]_{x\vvbar f(x) \vvbar z}
$$
---using twice the general fact that $[e]_v = [e]_w$ when $e$ is an event of 
compatible configurations  $v$ and $w$ of a stable family.  
To show $x\vvbar (y\cap f(x)) \vvbar z \in {\G} \sncirc {\F}$ we require 
$$
x\vvbar (y\cap f(x)) \in {\F} \ \hbox{ and }\ (y\cap f(x)) \vvbar z  {\G}\,.
$$
From $f(x)\leb_B y$ with $f(x)\leb_B f(x)$ we obtain $f(x)\leb_B (y\cap f(x))$; so $x\vvbar (y\cap f(x)) \in {\F}$. 
From $f(x)\leb_B y$ we get $g(f(x)\cap y) = g(f(x))\cap g(y)$.  But $g(f(x))\leb_C z$ and $g(y)\leb_C z$ ensuring  $g(f(x))\cap g(y)\leb_C z$
. Hence $g(f(x)\cap y) \leb_C z$ and $(y\cap f(x)) \vvbar z \in {\G}$, as required.  
This establishes a bijection between the events of $\Pr({\G}) \scirc \Pr({\F})$  and those of $\Pr({\G}\circ {\F})$.  

For an isomorphism, we require the bijection respects causal dependency and consistency.  The matching of a configuration
$x\vvbar  z$ in ${\G} \circ {\F}$ with a configuration $x\vvbar f(x) \vvbar z$ in  ${\G} \sncirc {\F}$ clearly respects inclusion.  This implies
$$
d' \leq_{x\vvbar z} d \iff d' \leq_{x\vvbar f(x) \vvbar z} d\,,
$$
for $d$, $d'$ in $x\in\iconf A$ or $z\in\iconf C$.  
This entails that the bijection on events given by $\theta$ and $\phi$ respects causal dependency.  

Via the matching of configurations, both $\theta$ and its inverse $\phi$   may be shown to preserve consistency.  This establishes the isomorphism of  the lemma.
 \end{proof}
 
\begin{corollary}
The operation $(\_)_!$ is a (pseudo) functor from the category of affine-stable maps  to concurrent strategies. 
\end{corollary}

\subsection{For Section~\ref{sec:adjn}, the adjunction}

\begin{proposition}\label{lem:stfamcmp}
Let $f$ be an additive-stable function from $A$ to $B$ between   \esswp. Define
$$
\eqalign{
F_!\eqdef&\set{x\vvbar y\in\iconf{A^\perp\vvbar B}}{f x\leb_B y}\,,\cr
F^*\eqdef &\set{y\vvbar x\in\iconf{B^\perp\vvbar A}}{y\leb_B f x}\,.
}
$$
Define
$\xymatrix{f_! :\Pr(F_!)\ar[r]^\max& A^\perp\vvbar B}$ and $\xymatrix{f^* : \Pr(F^*)\ar[r]^\max& B^\perp\vvbar A}$.  Then the composition of strategies $f^*\scirc f_!$ is isomorphic to  $$\xymatrix{\Pr(F^*\circ F_!)\ar[r]^\max& A^\perp\vvbar A}$$ 
and $ f_!\scirc f^* $ 
to  $$\xymatrix{\Pr(F_!\circ F^*)\ar[r]^\max& B^\perp\vvbar B}\,,$$ based on the relational composition of the stable families.  
\end{proposition}

\begin{theorem} Let $f$ be an additive-stable function from $A$ to $B$ between  \esswp. In the bicategory of strategies the strategies $f_!$ and $f^*$ form an adjunction $f_!\dashv f^*$.
\end{theorem}
\begin{proof}
It is easiest to carry out the arguments by considering the associated constructions on stable families.  We obtain the compositions $f^*\scirc f_!$  and $f_!\scirc f^*$  from  ``relational'' compositions of the stable families 
$$
F_!\eqdef\set{x\vvbar y\in\iconf{A^\perp\vvbar B}}{f x\leb_B y}
$$
for $f_!$
and 
$$
F^*\eqdef \set{y\vvbar x\in\iconf{B^\perp\vvbar A}}{y\leb_B f x}
$$
for $f^*$.

By Proposition~\ref{lem:stfamcmp}, the composition $ f^*\scirc f_!$ is the event structure $\Pr(F^*\circ F_!)$ derived from the stable family
$$
F^*\circ F_!= \set{x\vvbar x' \in \iconf{A^\perp\vvbar A}}{fx \leb_B fx'}
$$
---obtained as the relational composition of the stable families $F_!$ and $F^*$.  Recall, from Lemma~\ref{lem:CCbel},
that the stable family of $\cc_A$ is 
$$
C_A\eqdef \set{x\vvbar x'\in\iconf{A^\perp\vvbar A}}{x\leb_A x'}\,.
$$
Define the
unit $\eta: \cc_A \Rightarrow f^*\scirc f_!$ to be the map $\Pr(I)$ of \esswp~got from the inclusion of stable families
$$
I: C_A \hookrightarrow F^*\circ F_!\,;
$$
 clearly,  $x\vvbar x'\in C_A$, \ie~$x\leb_A x'$, implies $fx\leb_B fx'$, so  $x\vvbar x'\in F^*\circ F_!$.

By Proposition~\ref{lem:stfamcmp}, the composition  $ f_!\scirc f^* $ is the event structure  $\Pr(F_!\circ F^*)$  got from the stable family
$$
F_!\circ F^* = \set{y\vvbar y'\in \iconf{B^\perp\vvbar B}}{\exists x\in \iconf A.\ y\leb_B fx\ \&\ fx \leb_B y'}
$$
---obtained as the relational composition of the stable families $F^*$ and $F_!$. 
The 
  counit $\eps: f_!\scirc f^* \Rightarrow \cc_B$ is the
 the map $\Pr(J)$ got from the inclusion of stable families
$$
J:  F_!\circ F^* \hookrightarrow C_B\,;
$$
 clearly,  $y\vvbar y'\in F_!\circ F^*$, \ie~$y\leb_B fx$ and $fx\leb_B y'$, implies $y\leb_B y'$, so  $y\vvbar y'\in C_B$.

To obtain an adjunction $f_!\dashv f^*$ we require
(i)\ $(f^*\eps) (\eta f^*) = \id_{ f^*}$, 
\ie~the composition of the 2-cells
$$
\xymatrix{
B\ar@/^1.5pc/[rr]_{\Uparrow\eps}|-{\! +\!}^{\cc_B}\ar[r]|-{\! +\!}_{f^*}&A\ar@/_1.5pc/[rr]^{\Uparrow\eta}|-{\! +\!}_{\cc_A}\ar[r]|-{\! +\!}_{f_!}&B\ar[r]|-{\! +\!}_{f^*} &A\\
}
$$
is the identity 2-cell $\id_{f^*}: f^*\Rightarrow f^*$; 
and (ii)\ $(\eps f_!)(f_! \eta) = \id_{f_!}$, \ie~the composition of the 2-cells
$$
\xymatrix{
A\ar@/_1.5pc/[rr]^{\Uparrow\eta}|-{\! +\!}_{\cc_A}\ar[r]|-{\! +\!}_{f_!}&B\ar@/^1.5pc/[rr]_{\Uparrow\eps}|-{\! +\!}^{\cc_B}\ar[r]|-{\! +\!}_{f^*}&A\ar[r]|-{\! +\!}_{f_!} &B\\
}
$$
is the identity 2-cell $\id_{f_!}: f_!\Rightarrow f_!$. 

We establish (i) and (ii) by considering the companion diagrams for stable families---the diagrams (i) and (ii) are got by applying $\Pr$ to the  diagrams for stable families.  Consider the diagram for (i). It takes the form
$$
\xymatrix{
\iconf B\ar@/^1.5pc/[rr]_{\rotatebox{90}{$\subseteq$} }|-{\! +\!}^{C_B}\ar[r]|-{\! +\!}_{F^*}&\iconf A\ar@/_1.5pc/[rr]^{\rotatebox{90}{$\subseteq$} }|-{\! +\!}_{C_A}\ar[r]|-{\! +\!}_{F_!}&\iconf B\ar[r]|-{\! +\!}_{F^*} &\iconf A\,,\\
}
$$
yielding the inclusion $C_A\circ F^* \subseteq F^*\circ C_B$. We check this is the identity inclusion,  from which (i) follows, by showing the converse inclusion $F^*\circ C_B \subseteq C_A\circ F^*$.  
Suppose $y\vvbar x\in F^*\circ C_B$, \ie
$$
y\leb_B y' \ \& \ y'\leb_B f x\,,
$$
for some $y'\in\iconf B$.  Then,
$$
y\leb_B fx \ \&\ x\leb_A x\,,
$$
so $y\vvbar x\in C_A\circ F^*$.  

The diagram for (ii) takes the form
$$
\xymatrix{
\iconf A\ar@/_1.5pc/[rr]^{\rotatebox{90}{$\subseteq$}}|-{\! +\!}_{C_A}\ar[r]|-{\! +\!}_{F_!}&\iconf B\ar@/^1.5pc/[rr]_{\rotatebox{90}{$\subseteq$}}|-{\! +\!}^{C_B}\ar[r]|-{\! +\!}_{F^*}&\iconf A\ar[r]|-{\! +\!}_{F_!} &\iconf B\,,\\
}
$$
yielding the inclusion $F_!\circ C_A \subseteq C_B \circ F_!$.  To show (ii),  we check that the converse inclusion 
$C_B \circ F_! \subseteq F_!\circ C_A$ also holds.
Suppose $x\vvbar y \in C_B\scirc F_!$, \ie 
$$
fx \leb_B y' \ \&\ y'\leb y\,,
$$
for some $y'\in\iconf B$. 
Then,
$$
x\leb_A x \ \&\ fx\leb_B y\,,
$$
so $x\vvbar y \in F_! \circ C_A$.
 \end{proof}

\end{document}